\documentclass[pdflatex,sn-basic]{sn-jnl}% 
\usepackage{graphicx}%
\usepackage{multirow}%
\usepackage{amsmath,amssymb,amsfonts}%
\usepackage{amsthm}%
\usepackage[title]{appendix}%
\usepackage{xcolor}%
\usepackage{textcomp}%
\usepackage{manyfoot}%
\usepackage{booktabs}%
\usepackage{algorithm}%
\usepackage{algorithmicx}%
\usepackage{algpseudocode}%
\usepackage{listings}%
\usepackage{graphicx}
\usepackage{tikz}
\usepackage{empheq}
\usetikzlibrary{calc}

\usepackage{mathtools}
\usepackage{bookmark}

\theoremstyle{thmstyleone}%
\newtheorem{theorem}{Theorem}% 
\newtheorem{proposition}[theorem]{Proposition}% 

\newtheorem{corollary}[theorem]{Corollary}% 
\newtheorem{lemma}[theorem]{Lemma}% 

\theoremstyle{thmstyletwo}%
\newtheorem{remark}{Remark}%

\theoremstyle{thmstylethree}%

\newcommand{\ud}{\,d}%\mathrm{d}} % not common

\usepackage{amsmath,amssymb,mathtools}
\usepackage{xcolor}
\usepackage{tikz}
\usetikzlibrary{decorations.pathreplacing,calc,arrows.meta}
\usepackage[bottom]{footmisc}
\usepackage{soul} 
\usepackage{cancel}
 
\begin{document}

\title{Structured Population Models for Follicular Development}

\author[1]{\fnm{Edilbert} \sur{Christhuraj}}\email{edilbert.christhuraj@hs-anhalt.de}
\equalcont{These authors contributed equally to this work.}

\author[1]{\fnm{Claudio} \sur{Iuliano}}\email{claudio.iuliano@hs-anhalt.de}
\equalcont{These authors contributed equally to this work.}

\author*[1]{\fnm{Alexander} \sur{Lange}}\email{alexander.lange@hs-anhalt.de}
\equalcont{These authors contributed equally to this work.}

\affil*[1]{Anhalt University of Applied Sciences, Bernburger Str. 55, 06366 Köthen, Germany}

\abstract{
Ovarian follicle development has been modeled either by compartmental models, where developmental stages are represented by discrete compartments, or by physiologically structured population models (PSPMs), formulated as transport PDEs describing follicles along a continuous maturation variable. Compartmental models reproduce full-cycle hormonal and follicular dynamics but cannot describe cellular features. PSPMs capture multiscale behavior, yet lack a direct link to compartmental descriptions, making hormonal feedback difficult to incorporate. Here, we introduce a class of ODE systems derived as finite-dimensional reductions of PSPMs via a perturbative moment closure. These systems provide backbone approximations of the structured dynamics, retaining key nonlinear mechanisms governing recruitment, selection for dominance, and atresia while projecting maturation onto macroscopic variables. Zero- and first-order moments, derived from phenomenological assumptions regarding follicular maturation, allow us to realistically model follicle sizes and blood estradiol concentrations over one follicular wave. This application is supported by bovine data and likely generalizable to other mammalian species.}


\maketitle

\section{Introduction}\label{intro}

Over the past decades, various mathematical models have been developed to describe the complex multiscale process leading primordial ovarian follicles into ovulatory follicles, known as folliculogenesis \citep{biswas22}. These models play a crucial role in advancing our understanding of ovarian development in physiological scenarios as well as guiding diagnostic and therapeutic strategies for menstrual disorders such as polyendocrine metabolic ovarian syndrome (PMOS) \citep{TeedeKhomamiMormanEtAl2026}.

The central tasks in modeling folliculogenesis are to identify and to describe the hormonal regulation and the follicle growth dynamics underlying ovarian development. These biological mechanisms are typically investigated either from macroscopic (follicular length scale \(\sim \) 5--20\(\mathrm{mm}\)) or microscopic perspectives (cellular length scale \(\sim10\mathrm{\mu m}\)) \citep{gougeon1996regulation,ackert2001intercellular}%; \cite{biswas22,reny25,Orozco-Galindo2025-br}
. 
At the macroscopic level, follicular maturation manifests itself through distinct morphological stages, such as primordial, primary, pre-antral, antral, and pre-ovulatory, each characterized by specific structural and functional transformations such as sensitivity to follicle-stimulating hormone (FSH) and luteinizing hormone (LH), development of an antrum, layer formation, etc. 
\citep{Palermo2007-rq}.
At the microscopic level, different cell populations go through different developmental processes. For example, early follicular growth involves the proliferation and differentiation of pre-granulosa cells, which gradually organize into one or more layers of cuboidal granulosa cells surrounding the oocyte \citep{mcgee2000initial,WANG2025103614}. As the follicle becomes antral, a fluid-filled cavity---the antrum---forms due to the secretion of follicular fluid rich in hormones \citep{edwards,rodgers2010morphology}. The theca interna cells begins to synthesize androgens, which are aromatized by granulosa cells into estradiol, thereby reinforcing follicular growth and coordinating the endocrine transition toward dominance and ovulation \citep{Franks2018-jd}.
From the antral to the ovulatory phase, the follicle undergoes rapid enlargement accompanied by intensified hormonal activity. Mural granulosa cells acquire luteinizing hormone (LH) receptors and respond to the LH surge by triggering cumulus expansion, resumption of oocyte meiosis, and the synthesis of enzymes that degrade the follicular wall. Eventually, the follicular wall ruptures, releasing the mature oocyte \citep{russell2007molecular, richards2018ovulation}. In parallel, surrounding cells differentiate into the theca interna and theca externa layers, providing both mechanical support and endocrine function through steroidogenesis, which is essential for sustaining granulosa cell development \citep{Magoffin2002}. The interplay between cellular populations and the hormonal environment drives the selection of dominant follicles and ultimately leads to ovulation \citep{baird1987model, edson2009the_ovary}.

To realistically capture
the dynamical coordination between cellular differentiation, follicular maturation, and hormonal feedback, 
models of folliculogenesis need to cover multiple scales.
Two main approaches have been followed:
\begin{enumerate}
\item {\em Ordinary differential equation (ODE) models}, which describe the development at
the follicular scale. Numerous formulations and applications have been proposed by many research groups over half a century, e.g., \cite{Bogumil1972-cw, Lacker1981-js, SOBOLEVA200045, Clark2003-oe, FISCHERHOLZHAUSEN2022111150, shilo} 
to name a few that below we cite for their detailed contribution. 
\item {\em Partial differential equation (PDE) models}, which describe the development at 
the histological scale. 
To our knowledge, those have extensively been employed by the SISYPHE project-team of the Inria Institute (France); e.g., \cite{ECHENIM200557,reach, philippe, aymard16, Monniaux2016-yt, dimred}.
Their contributions span a period of more than a decade.
\end{enumerate}
The two approaches provide complementary perspectives on the development of ovarian follicles and their interaction with the hormonal environment.

ODE models are well established and can describe the dynamics of the entire menstrual or estrous cycle, including endocrine feedback within the hypothalamic--pituitary--ovarian axis, while---at least in part---tracking follicular development \citep{Clark2003-oe,BOER20115987}.
This strategy uses compartmental models, which describe follicular development through discrete stages, distributing the total follicular mass in the ovaries from pre-antral to ovulatory compartments \citep{HENDRIX201431,Graham2017-ye,cells11233908}. These ODE models realize the stage-dependent hormonal sensitivity of follicles by associating the hormone production with the discrete follicular stages. Despite being accurate in producing the correct endocrine feedback loop, they do not explicitly represent individual follicles and therefore cannot account for inter-follicular competition.

Another strategy in the ODE setting is to track the dynamics of individual follicles and their competitive interactions \citep{Lacker1981-js,SOBOLEVA200045,Lange2018-xf,shilo}.
In that, it predicts the amounts of estradiol (E2) produced by each follicle or, in more recent work, the individual follicular sizes, yet this strategy typically lacks a structured representation of developmental stages.  When individual follicle dynamics are coupled with the full endocrine feedback loop, the absence of follicular stages makes the coupling between hormones and follicles less natural to formulate: \cite{FISCHERHOLZHAUSEN2022111150}, e.g., used time-dependent coefficients to synchronize the follicular E2-production and the production of progesterone (P4).

The mentioned examples highlight a key limitation of existing ODE models: capturing both the internal structure of follicular development and the competitive dynamics among follicles is a challenging task within a purely compartmental model.
In contrast, the PDE framework provides a natural setting for incorporating a continuous representation of the macroscopic as well as the microscopic states.
While the former is included through averages (allowing for a temporal description of follicular masses or sizes, similar to the ODE approach), the latter enters the formalism through an additional independent variable that characterizes the histological development of the follicular cells.

PDE models that introduce an auxiliary stage or maturity variable fall into the class of
{\em physiologically structured population models} (PSPMs),\footnote{Depending on the context, PSPMs are often classified as mesoscopic models, since they describe the evolution of population densities rather than tracking individual entities. In the present work, we adopt a slightly different terminology similar to \cite{aymard16}: we regard the physiological structure as carrying the relevant microscopic information, so that the PSPM effectively encodes microscopic dynamics.}
which have been used in studying populations in various biological systems \citep{metz, DeRoos1990, Diekmann1995PerturbingES, inaba2017age, Diekmann2020, Diekmann2020b}.
In the field of ovarian folliculogenesis, the SISYPHE  team \citep{aymard16, Monniaux2016-yt}
has been describing and analyzing follicles as separate granulosa cell populations.  
Their PSPM for granulosa cells exhibits several realistic features. 
First, it relates the macroscopic notion of follicular maturity to FSH consumption by coupling local cellular dynamics to systemic endocrine feedback. Second, the model encodes follicular competition through a variable representing follicular maturity. The resulting PSPM is conceptually similar to competitive systems studied in other contexts \citep{deroos08,hart11}.
However, despite having a natural coupling between FSH and the granulosa cell number, the discontinuous nature of the involved population densities in the SISYPHE model---due to distinct phases of cell proliferation, death, and differentiation---complicates the numerical implementation.

For a comprehensive description of the competitive follicular dynamics neither the ODE nor the PDE approach is sufficient on its own.
While \cite{HENDRIX201431}, for example, 
explored the link between elevated androgen sensitivity and PMOS, their model cannot address
the competitive interactions of individual follicles.
Although, \cite{FISCHERHOLZHAUSEN2022111150}
embedded a competitive model \citep{Lange2018-xf} within the full hormonal cycle, 
their follicles are missing a characterization of maturation different from size.
Follicles at different developmental stages
are treated on the same footing
regarding their hormonal responsiveness.
To correct for the lacking stage-dependence
in the hormonal modeling, the authors introduced, as mentioned above, time-dependent coefficients to synchronize the E2 and P4 production.
In their approach, the missing maturation dynamics is absorbed into the hormones rather than being modeled explicitly.
This is different in the PDE approach by the SISYPHE team, which involves a notion that captures the cellular maturation,
although their maturity parameter is not formally associated with (real) histological changes. Besides, their large number of parameters makes model calibration computationally demanding \citep{aymard16}.
This results in expensive simulations and, as a consequence, renders the underlying mechanisms of follicular development and competition difficult to interpret.

In this work, we bridge the gap between PSPMs and ODE approaches trying to combine the strengths of both. We show how 
the PSPM framework naturally extends competitive ODE models by incorporating additional information on the follicular development
while remaining amenable to coupling with
hormones. 
Technically, we perform a low order moment expansion to obtain integrals of the PDEs that
satisfy the modeling ODEs.
Besides being able to evaluate PSPMs more systematically when knowing about their moments, we find a simple way to couple individual follicles (i.e., their zeroth and first order moments) with involved hormones.
Following this route, we think that at some point, one will be able to incorporate the full endocrine feedback loop into structured models of follicular dynamics.

\section{Results}\label{results}

Starting with a set of macroscopic observables that is governed by a particular system of ODEs
(representing the time evolution of cell numbers or sizes of follicles as specified by the concrete application),
we constrain the structure of the PSPM by requiring that the macroscopic dynamics is obtained through a finite truncation of the moment hierarchy associated with the  underlying PDE.\footnote{Similar ODE reductions in the context of PSPMs can be found in \cite{Diekmann2020, Diekmann2020b} or \cite{kooi2003physiologically}. Other truncation techniques have been used in evolutionary dynamics of phenotype-structured models \citep{LORENZI2015166, chrisholm, Villa2021,Villa2025,almeida19,Ardaseva2019-ub}.}
We show how zeroth-moment dynamics, as given by the ODEs, can be reconstructed from a PSPM setting together with some first-moment dynamics, 
where the first moments characterize
the developmental stage (or level of maturity) of a follicle.
The zeroth moment dynamics can be used to implement information about available amounts of FSH, as done in the SISYPHE approach.
In addition, we demonstrate that the difference of the zero and first order moments can efficiently be coupled with E2 concentrations.  
The data set in \cite{Cummins2012-oa} will support these findings, explicitly for cows.

Higher-order moments encode more detailed microscopic (histological) information, 
which we expect to have sub-leading contributions to the macroscopic observables. Therefore, we omit higher-order moments  from the present analysis. The resulting time-independent, first order approximation will be sufficient, as yielding a tractable dynamical system involving all relevant mechanisms: follicular growth, competition, selection, and coupling with hormones.

For now, we model the macroscopic variables only during the final stages of the follicular development (from antral
to ovulatory/atretic), therefore we assume that the PSPMs are given by constant parameters, effectively neglecting the back-reaction with the hormonal environment. Similar assumptions are included in the SISYPHE model. In the Discussion, we will suggest possible extensions.

\subsection{From follicular PSPMs to ODEs} \label{pspmtoode}

Previous ODE approaches suffer from two main limitations: first, the lack of follicular and histological resolution, which complicates the analysis of follicular competition, and second, the absence of an explicit follicular stage variable, which hinders a natural coupling with the endocrine feedback loop.
Existing PDE frameworks that model follicular development, on the other hand, allow to couple both follicle sizes and stages but are often difficult to interpret and to relate to ODE-based models, including those describing endocrine dynamics. In other words, integrating structured population models with established ODE frameworks remains a significant challenge.

To keep the strengths of established modeling approaches, we present a PDE framework in which each follicle is described by a cellular population governed by a PSPM. 
We identify key structural properties the PDE has to satisfy in order to be reducible to a simpler system of ODEs connecting follicular mass and stage. Such a reduction allows to design and better interpret a PDE model, especially if connected to a well-developed ODE system that models parts of the hormonal cycle.

In our modeling approach we do not distinguish between granulosa and theca cells, or other microscopic structures, which---if being of particular interest---could always be included by coupled PSPMs describing their mutual interactions.
We consider the follicle as a macroscopic biological entity that collectively develops through successive stages of maturation. In this setting, the density function can either be interpreted as the portion of the follicle that has the same level of maturation, or, in case of a finer microscopic resolution, as the number of cells with a given maturity. The former interpretation implicitly assumes a coarse‑grained, radially symmetric structure, allowing to treat all follicular components (granulosa, theca, and fluid) as parts of a single effective population. 

\paragraph{The structured population model}

We model the follicular development (from antral
to ovulatory/atretic stages) by a system of transport PDEs,
\begin{align} \label{pdesystem}
\partial_t \phi_i + \partial_m (g_i \phi_i)= (-\lambda_i+p_i)\phi_i,\qquad i=1,2,\dots, N,
\end{align}
with $N$ being the number of follicles and coefficients,
$g_i,\lambda_i,p_i$, representing
non-local functions. The involved mathematical expressions are defined and interpreted as follows:
\begin{enumerate}
\item $m\in[0,1]$ 
is a (dimensionless) local variable defining the progress of maturation.
That is, $m$ indicates the degree of cellular differentiation towards maturity.
 \item \label{item2} $\phi_i=\phi_i(t,m)$ is the density function of the $i$-th follicle. 
 Its units are chosen in accordance with the modeling task.
 %For example, 
In Section \ref{aymardconn}, for example, the quantity $\phi_i(t,m)\ud m$ %is
represents the number of cells having cellular maturation between  $m$ and $m+\ud m$ at time
$t\in[0,\infty)$.
In Section \ref{langeconn}, $\phi_i(t,m)\ud m$ represents the (spherical) portion of the follicle (measured in mm, the length unit of the follicles' diameter) having the same maturity $m$ at a given time $t$;
\item \label{item3} The zero and the first raw moments
\begin{equation}\label{EqDefMassMaturity}
x_i(t)=\int_{0}^{1} \phi_i(t,m)\,\ud m,\qquad s_i(t)=\int_{0}^1 m\,\phi_i(t,m)\,\ud m
\end{equation}
are referred to as follicular {\em mass} and {\em stage}, respectively. Note that $x_i$ and $s_i$ have the same physical units: in Section \ref{aymardconn}, the zeroth moment represents the total number of cells in a follicle, measured in millions \citep{aymard16}; in Section \ref{langeconn}, the zeroth moment represents the diameter of the follicle, measured in $\mathrm{mm}$ \citep{Lange2018-xf}. 
Independent of its units, the developmental stage $s_i(t)$ of the $i$-th follicle (e.g., pre-antral, antral, dominant)
can be associated with the changing follicular maturity.
The moments of all follicles are denoted as vectors: $x=(x_1,\dots,x_N)^T\in\mathbb R^N$ and $s=(s_1,\dots,s_N)^T\in\mathbb R^N$.
\item\label{itemmaturation} 
The transport coefficient, $g_i=g_i\left(t,m,x,s\right)$,
which we call {\em maturation function},
encodes the microscopic mechanisms that drive the microscopic population of follicle $i$ to develop and progress through different levels of maturation $m$.
We assume that $g_i(t,m,x,s)$ is continuous in all its variables.

\item\label{itemgrowdeath} 
The coefficients,
$\lambda_i=\lambda_i\left(t,m,x,s\right)$ and $p_i=p_i\left(t,m,x,s\right)$, supposed to be continuous in their variables as well, denote the local death and proliferation rates within the $i$-th follicle, respectively. Biologically, they reflect the probability of cell loss and cell division at a given maturation level, respectively. They are regulated by both intrinsic factors and follicular interactions.
\end{enumerate}
Because of Items \ref{itemmaturation}-\ref{itemgrowdeath}, as we should note, our model---if spelled out completely---is a non-linear partial integro-differential equation. However, we will keep referring to it as PDE.  

Before introducing explicit functional forms for the different terms, we first  
relate our formulation to existing models of follicular and cellular population dynamics. While pointing at similarities and differences,
we discuss conceptual and technical difficulties in existing models and outline how to overcome them.

\paragraph{Relation with previous models}

Similar to the SISYPHE
population model for cellular development in ovarian follicles \citep{ECHENIM200557, aymard16}, we assume that the rate of maturation, $g_i$, depends on both the local maturity $m$ and the global follicular variables $(x,s)$. Likewise, the growth and decay rates of each follicle arise from internal cellular duplication and death processes that are modulated by these macroscopic quantities, thereby incorporating competition for shared resources. In this framework, each cell consumes environmental resources according to both the overall state and mass of its own follicle and, indirectly, that of its neighboring follicles. These assumptions imply that, from a mathematical viewpoint, the nonlinearities in the PDE system are non-local and depend only on the zero and first raw moments.

Although, conceptually, our model is related to the structured populations of the SISYPHE team and similar to classical cellular dynamics models \citep{Bell1967-vj}, it diverges from them in several aspects. The first key difference is that our cell density depends only on time, maturation and macroscopic quantities, omitting the explicit dependence on cellular age. In the formulation of SISYPHE  project-team, cell replication occurs at a specific age, where the cell density exhibits a discontinuity. This assumption, while biologically reasonable, introduces the undesirable feature that newborn cells begin their cycle with a positive age. In contrast, the framework of Bell and Anderson avoids this issue by resetting newborn cells to age zero, leading to a PDE of the form \eqref{pdesystem} but together with an additional rescaling term associated with volume structure.\footnote{In their setting, they obtain a volume‑structured rather than maturity‑structured model.}  In our approach, we ignore such age and volume terms and instead describe cellular processes through continuous functions of follicular mass and maturity. By smoothing local division events into a continuous description, we avoid discontinuities and rescaling effects that, in our view, have limited influence on the large-scale structures of growth dynamics. This yields a simpler, but biologically consistent representation of follicular mass growth. A second major difference concerns the functional forms of the maturation, decay, and growth rates. Especially, the forms proposed in e.g. \cite{ECHENIM200557,Monniaux2016-yt}, have a complex structure, which makes the derivation of a closed ODE system for the macroscopic quantities a difficult task. In contrast, we specify these functions as low‑order polynomials in the maturation variable $m$, allowing an analytically tractable reduction of the PDE to an ODE system governing follicular mass and maturity. These simplifications eliminate the need of introducing an additional PDE system for each follicle, as done in \cite{philippe} and in \cite{dimred}. At any rate, the ODE reductions provide a conceptually and computationally more accessible representation of follicular dynamics, even though we sacrifice some microscopic details. However, we regard this trade‑off as valuable: the reduced system offers intuition about the functional dependencies and parameter roles in the structured PDE, which can be refined in future work.

Analytical requirements in mind, we propose that $\phi_i\in C^1( [0,T]; L^1(0,1) )$, because well-posedness of the Cauchy problem for the SISYPHE model has already been established by \cite{shang2010cauchyproblemmultiscaleconservation} in the context $\phi_i\in C^0( [0,T]; L^1(0,1) )$, i.e., $\phi_i$ being continuous from $[0,T]$ to the set of integrable function in the $m$ variable. For technical purposes, we also assume that the initial profile $\phi_i(0,m)$ is supported in $m\in (0,1) $. We do not aim to prove existence and uniqueness of continuous and differentiable solutions in time but, under some natural requirements on the functions $g,\, \lambda, \, p$, one can use well-known results to improve the regularity on $\phi_i(t,m)$. 
In Sections \ref{math_sisyphe}-\ref{math_lange},
we study existence, uniqueness, boundedness, and regularity of integrated quantities of the PDE models discussed in this work.
Based on these results, well-posedness of PDE solutions can be obtained using, e.g., integral and/or operator techniques together with standard fixed point arguments \citep{Gripenberg_Londen_Staffans_1990, yamaguchi,Perthame2008}.

\paragraph{ODEs associated to the truncated PSPM}

In order to relate the PSPM to an ODE system for mass and maturity as defined in \eqref{EqDefMassMaturity}, we introduce an approximation scheme with respect to the maturity parameter $m$ (similar, e.g., to the one discussed in \cite{Villa2025}). In particular, we assume that
\begin{itemize}
\item[(a1)]\label{fulla1} at any time $t\in [0,T]$, the essential support of $\phi_i$ does not include the boundaries, i.e.,\footnote{
Being a local maturity variable, $m$ can be interpreted as the level of maturation of single cells. We model each cell in the population $\phi_i$ as characterized by a value in the interval $(0,1)$, close to $0$ for being a young cell, and close to $1$ for being a fully developed and differentiated cells. For example, if there were $N=30\times10^6$ cells with $m\approx1$ in a given follicle, then the follicular stage (or maturity) would be $s\approx30\times10^6$.}
\begin{equation} 
\mathrm{ess~supp}~\phi_i(t,\cdot)\subset(0,1)\label{a1};
\end{equation}

\item[(a2)]\label{fulla2} %the main
contributions to growth, decay, and drift come from expressions with the following continuous dependencies,
\begin{subequations}\label{a2}
\begin{align}
p_i&=p_{i,0}(t,x,s), \qquad \lambda=\lambda_{i,0}(t,x,s)\\
g_i&=g_{i,0}(t,x,s)+m\, g_{i,1}(t,x,s),~~\text{resp.}
\label{a2_b}
\end{align}
\end{subequations}
i.e., growth and death rates are constant with respect to $m$, and the maturation function is linear with respect to $m$. That is,  \eqref{a2} can be interpreted as a truncation or approximation to more general functions $g_i$, $\lambda_i$ and $p_i$; see also Theorem \ref{higherorder}.
\end{itemize}
With these assumptions, we find (cf.~Thm. \ref{odethm} in Methods) that mass and maturity associated to the density $\phi_i$ solve the ODE system
\begin{subequations}\label{odesysmain}
\begin{empheq}[left=\empheqlbrace]{align}
\dot x_i&= x_i\left[-\lambda_{i,0}(t,x,s)+p_{i,0}(t,x,s)\right]\\
\dot s_i&= s_i\left[-\lambda_{i,0}(t,x,s)+p_{i,0}(t,x,s)\right]+g_{i,0}(t,x,s) x_i+g_{i,1}(t,x,s) s_i.
\label{odesysmain-s}
\end{empheq} 
\end{subequations}
This result establishes an explicit connection between the microscopic scale of cell population development, described by the PSPM and the macroscopic scale of follicle growth, described by the ODE system of mass and stage. The macroscopic dynamics encoded in PSPMs illustrate how developmental dynamics (described by \( \dot{s}_i \)) and mass evolution (described by \( \dot{x}_i \)) interact with each other. 
Therefore, we can understand each follicle as  only described (macroscopically) by a couple $(x_i,s_i)\in \mathbb R^2$.

In Section \ref{langeconn}, we will show how this interaction addresses several critical limitations of existing mass and estradiol models, such as those in \cite{Lacker1981-js,SOBOLEVA200045,Lange2018-xf,shilo}, and how it naturally links mass dynamics with hormonal production models.   

Furthermore, we will be able to
simplify the structure of the system \eqref{odesysmain} by analyzing the follicle's average (level of) maturity,
\begin{equation} \label{defn:avstage}
\bar s_i:=\frac{s_i}{x_i}=\dfrac{\int_0^1 m\, \phi_i(t,m)\ud m}{\int_0^1 \phi_i(t,m)\ud m}.
\end{equation}
If $\phi_i\,\ud m$ is interpreted as the portion of maturation (see Section \ref{langeconn}), then $\bar{s}_i$ is also referred to \emph{relative maturity}, to stress the defining ratio between the follicular level of maturity $s_i$ and its geometric size $x_i$, both expressed in units of length. 
Importantly, since we are assuming that $\mathrm{supp}\,\phi_i(t,m)\subset(0,1)$ for all $t\in[0,T]$, it follows that $\bar s_i\in(0,1)$. That is, the average maturity can be understood as a normalized dimensionless indicator of the developmental stage of the follicle.

Due to Corollary \ref{evolutionsbar} and Corollary \ref{redsbar}, we can analyze a simpler system
\begin{subequations}\label{uncoupledmain}
\begin{empheq}[left=\empheqlbrace]{align}
\dot{x}_i
&= x_i\left[-\lambda_{i,0}(t,x,x\circ\bar{s})
+p_{i,0}(t,x,x\circ\bar{s})\right],
\label{uncoupledmain-x}
\\
\dot{\bar{s}}_i
&= g_{i,0}(t,x,x\circ\bar{s})
+\bar{s}_i g_{i,1}(t,x,x\circ\bar{s}).
\label{uncoupledmain-s}
\end{empheq}
\end{subequations}
where the maturity equation is replaced by a simpler one, modeling the average maturity (which we sometimes also call relative maturity). Here, the $\circ$ means the Hadamard product between vectors $x\circ \bar s=(x_1\bar s_1,\dots,x_N\bar s_N)$.
When looking at the equation for the average stage \eqref{uncoupledmain-s}, we note that it resembles the characteristic equation of the PDE \eqref{pdesystem}, namely
\begin{equation}
\dot m^{(i)}_{a}=g_i\left(t,m^{(i)}_{a},x,s\right)\equiv g_{i,0}(t,x,x \circ \bar s)+m^{(i)}_{a}g_{i,1}(t,x,x \circ \bar s),\label{explicitchar}
\end{equation}
where the subscript \(a\) identifies a particular solution (e.g., via the initial value); see Appendix \ref{app:characteristics}.
These characteristic curves $m_a^{(i)}(t)$ 
can be interpreted as the maturation trajectories of cells in the $i$-th follicle that start with initial maturities $m_a^{(i)}(0)$.

The correspondence between the average-maturity equation \eqref{uncoupledmain-s} and the characteristic equation \eqref{explicitchar} implies that, whenever
\begin{equation*}
m_a^{(i)}(0)=\bar s_i(0),
\end{equation*}
the corresponding solutions satisfy (under suitable regularity assumptions)
\begin{equation}\label{eq:correspondence}
m_a^{(i)}(t)=\bar s_i(t),\quad \text{for all } t\geq 0.
\end{equation}
This follows from the fact that
\begin{equation}
\frac{\ud}{\ud t}\left(\bar s_i-m_a^{(i)}\right)=\left(\bar s_i-m_a^{(i)}\right)g_{i,1}(t,x,x \circ \bar s).
\end{equation}
For sufficiently regular $g_i$, by the standard existence and uniqueness theorems for ODEs, $m_a^{(i)}(t)$ and $\bar s_i(t)$ coincide whenever their initial conditions coincide. This establishes that the evolution of the statistical variable $\bar s_i$ for the population density can be identified with the microscopic dynamics of the maturity along individual characteristic curves, provided that the two are initially equal.

Finally, when assuming that the functions $g_i$ are continuous, as proposed in \eqref{a2},
we observe that
points on the characteristic curves
\eqref{explicitchar} always propagate with finite speed, provided that $x(t)$ and $s(t)$ are also continuous in time.
Consequently, if $\mathrm{supp}\,\phi_i(0,\cdot)\subset(0,1)$, then there exists a time $T>0$ such that $\mathrm{supp}~\phi_i(t,\cdot)\subset(0,1)$ for all $t\in[0,T]$, as required by \eqref{a1}.

\paragraph{ODE invariance under maturation function transformation}

While the PSPM determines the time evolution of the total mass and maturity, the converse is not true. That is, the knowledge of mass and maturity is not enough to reconstruct the PSPM. Different PSPMs, however, can give rise to the same evolution of the macroscopic quantities. This freedom (i.e., different choices of the maturation function $g_i$) 
is a major advantage of our framework,
allowing to reproduce a broad spectrum of phenomena (cf.~Sect.~\ref{aymardconn} and \ref{langeconn}), despite the strict and somewhat artificial assumption (a2), including \eqref{a2}.

The involved assumptions are verified by considering two different maturation functions, 
\begin{align}
g_i=g_{i,0}+m\,g_{i,1}&,\quad \tilde g_i=\tilde g_{i,0}+m\,\tilde g_{i,1},
\intertext{satisfying (a2), for which in addition} 
\label{macroinv}
g_{i,0} x_i+g_{i,1} s_i&= \tilde g_{i,0} x_i+\tilde g_{i,1} s_i.
\end{align}
For them, the macroscopic ODE system for mass and maturity \eqref{uncoupledmain} remains unchanged.
Therefore, the knowledge of the mass-stage dynamics does not uniquely specify the microscopic development of the population.

Under a transformation $g_i\to \tilde g_i$ satisfying 
\eqref{macroinv}, the time evolution of the population density differs only
at the level of higher-order raw moments 
($\alpha\geq2$),
\begin{align}
s_i^{(\alpha)}=\int_0^1 m^\alpha \phi_i(t,m)\ud m\,,
\end{align}
as one confirms
when applying the transformations
\begin{equation}\begin{split}
&g_{i,0}(t,x,s)\to g_{i,0}(t,x,s)+s_i f_i(t,x,s), \\ &g_{i,1}(t,x,s)\to g_{i,1}(t,x,s)-x_i f_i(t,x,s).
\end{split}\label{gtransfo}
\end{equation}
Then one obtains
\begin{subequations}\label{transfo}
\begin{align}
 \dot s^{(\alpha)}_i&=s^{(\alpha)}_i \left[-\lambda_{i,0}+p_{i,0}\right]+\alpha\left[ g_{i,0} s_i^{(\alpha-1)} +  g_{i,1} s_i^{(\alpha)}\right]
\\ &\to ~ \dot s^{(\alpha)}_i=s^{(\alpha)}_i \left[-\lambda_{i,0}+p_{i,0}\right]+\alpha\left[ g_{i,0} s_i^{(\alpha-1)} +  g_{i,1} s_i^{(\alpha)}\right]+\alpha f_i\left[s^{(\alpha-1)}_i s_i-s^{(\alpha)}_i x_i\right].\label{transfo_b}
\end{align}
\end{subequations}
Since $s^{(0)}_i=x_i$ and $s_i^{(1)}=s_i$,
an additional
contribution 
$\alpha f_i [...]$
to the higher-moment dynamics in \eqref{transfo_b} appears only for $\alpha\geq2$. 

Determining higher moments would require sufficiently detailed histomorphological data, including their numerical formalization through a (maturity) variable. A scientific approach for such a purpose, which could involve cellular morphology or even bio-chemical data, has not been developed yet.
Consequently, the maturation function cannot be 
further specified or constrained by the information currently available to us.

In the following section, we propose maturation functions designed to capture the main mechanisms driving follicular development. 
We also apply our truncated framework to the PSPM approach developed by the SISYPHE project-team  (cf.~Introduction).
That is, we reconstruct the macroscopic observables, mass and maturity, 
of their PDE-model by only using the ODE-system 
\eqref{uncoupledmain}.
As an immediate application,
we can identify microscopic features 
such as cellular synchronization in dominant follicles near maximal levels of maturity.

\subsection{ODEs related to a truncated SISYPHE model} \label{aymardconn}

As a first application of our framework, we numerically reproduce the follicular dynamics of the SISYPHE model.
Among the extensive body of work produced by the SISYPHE team, we consider the PSPM calibrated by \cite{aymard16}, 
which describe follicular growth during the final developmental stage preceding ovulation.
The model is based on the equation
\begin{equation} \label{aymardeqn}
\partial_t \tilde \phi_i+\partial_a( \tilde g_i \tilde\phi_i )+\partial_m( g_i \tilde\phi_i)=-\lambda_i \tilde\phi_i,
\end{equation}
where $\tilde \phi_i=\tilde\phi_i(t,m,a)$
denotes the density function describing the age and maturation structure of the follicles, with $a$ and $m$ representing age and maturation, respectively.
In its original formulation, this model aims to study the evolution of the 
granulosa cell population, represented by
$\tilde \phi_i(t,m,a)\ud a \ud m$,
consisting of cells with age between $a$ and $a+\ud a$ and maturity between $m$ and $m+\ud m$,  as well as its relation with FSH levels. In this setting, the variable $m$ is an auxiliary variable that qualitatively represents cell maturity through its responsiveness to FSH. We emphasize that, in the original model, $m$ is not normalized;
see Aymard et al.~for further details.

In the present section, we introduce several simplifications and approximations in order to capture the microscopic mechanisms underlying the dynamics of the macroscopic observables, $x_i$ and $s_i$, while remaining consistent with the analysis performed by Aymard et al.

The first observation, as preliminarily discussed in \cite{philippe}, is that the coefficients are simple step functions
of the age variable $a$, with discontinuities at the transitions between different stages of the cell cycle. Within each stage, these functions are independent of $a$.
We propose that, under certain approximations, these age-dependent transitions can be eliminated, thereby allowing the full equation to be integrated with respect to \(a\). In particular, if \(\tilde \phi_i(t,m,a)\) is assumed to be differentiable with compact support in \(a \in (0, +\infty)\), the integration can be performed by parts. However, in the original formulation, the solution is discontinuous at specific values \(a = a^*\), 
corresponding to transitions between cellular growth and replication stages. 
At these points, a ``doubling flux condition'' is imposed: when a cell reaches the age \(a^*\), it instantaneously duplicates, resulting in a jump of the density according to $\tilde \phi_i\mapsto 2\tilde \phi_i$. One can eliminate this discontinuity by replacing the death term \(-\lambda_i\) with \(-\lambda_i + p_i\), where \(p_i\) represents the continuous replication rate
\begin{equation}\label{replication_rate}
p_i = \alpha_{0,i}(1 - \bar{s}_i) \ln(2),
\end{equation}
with  a constant  $\alpha_{0,i}>0$. Note that in \cite{aymard16}, follicle growth is terminated when cells enter the so-called ``differentiation zone'',
defined through a threshold and implemented as a hard cutoff in the computational domain. In contrast, we assume that the cell replication rate decreases as cells approach the maximal average maturity, reflecting their progressive transition toward their maximally differentiated functionality.

By making the cellular replication rate independent of age, we effectively bypass the discontinuities associated with replication. Moreover, since the replication rate has been smoothed by making it $a$-independent\footnote{A similar approach was taken by \cite{philippe}.}, the step function at the transition between the growth and replication stages is eliminated.  It is therefore natural, within the present formulation, to assume that $g_i$ is also independent of the age variable.\footnote{In the original model \citep{Monniaux2016-yt}, the maturation function includes a step function in the age variable at the transition between the cellular growth and replication stages.} We can then integrate \eqref{aymardeqn} over \(a \in (0, +\infty)\) to eliminate the age structure. Introducing
\begin{equation}
\phi_i(t,m) = \int_{0}^{\infty} \tilde{\phi}_i(t,m,a) \,\ud a,
\end{equation}
we obtain the following PDE system for $\phi_i$:
\begin{equation}
\partial_t \phi_i + \partial_m (g_i \phi_i) = (-\lambda_i + p_i) \phi_i, \quad i = 1, 2, \dots, N, \label{aymardpde}
\end{equation}
where we have used the assumption that $g_i$ is independent of $a$. Although several variants of the maturation function $g_i$ and the apoptotic terms have been considered by the SISYPHE team, none of these formulations directly satisfies the requirements in \eqref{a2} formulated in the previous section. We therefore seek to recover the phenomenology of their model while imposing the constraints specified in \eqref{a2}.
In doing so, we perform a first-order truncation of the SISYPHE model
as a perturbative expansion in $m$.

Let us start with the maturation function $g_i$. The SISYPHE models assume that the maturation function depends quadratically on the maturation variable $m$ and exponentially on the total ovarian maturity $\sum_j s_j$.  
As, so far, our approach (via Theorem \ref{odethm}) does not allow for quadratic maturation functions in $m$, we employ a relation between the characteristic equation for $m^{(i)}_{a}$ and the average $\bar s_i\equiv s_i/x_i$, as discussed in the paragraph below \eqref{explicitchar}.
In particular, we assume that
\begin{equation}
g_i\left(m,x,\bar s\right)=(1-c_{i,1})\theta\left(\bar s\right)+m\left[-\bar s_i+c_{i,1}\theta\left(\bar s \right)\right] , \qquad 0<c_{i,1}<1 \label{aymardmatlin}
\end{equation}
where $\theta$ is a function of $s\in\mathbb R^N$ such that
\begin{equation*}
\theta(0)=0, \qquad \lim_{\bar s\to 1}\theta(\bar s)=1.
\end{equation*}
With this specification, the average maturity follows the system of equations
\begin{equation}
\dot{\bar s}_i=-\bar s_i^2+\left[c_{i,1} \bar s_i -c_{i,1}+1\right]\theta\left(\bar s\right).
\end{equation}
In Theorem \ref{thm:aymard} we show that the average stage is bounded $\bar s_i(t)\subset (0,1)$ for all times, as expected from the statistical description, see comment after \eqref{defn:avstage}. 

We now emphasize a modeling choice that differs slightly from the one adopted by the SISYPHE team, but which we argue does not alter the core mechanisms underlying their proposed follicular dynamics. In their original model, the maturation function $g_i$ depends on an additional follicle-specific parameter, denoted $c_{i,2}$. Biologically, the pair $(c_{i,1}, c_{i,2})$ encodes follicle-specific sensitivities to hormonal signals. Mathematically, the combination of these follicle-dependent parameters determines a follicle-specific maximal average maturity $\bar{s}_i$.
Consequently, the selection of ovulatory follicles---which we expect to correspond to follicles attaining high average maturity\footnote{We recall that, in the original SISYPHE model, there is no natural scale for the local and follicular maturities $m$ and $s_i$. Consequently, distinguishing ovulatory from atretic follicles requires the introduction of additional thresholds \citep{aymard16}.}---could be driven not directly by follicular competition for environmental (hormonal) resources, but rather by intrinsic differences in the maximal average maturity attainable by individual follicles.

In Appendix~\ref{num}, using the model from \cite{aymard16}, we explicitly show that, for a fixed maximal FSH level, (i) the upper bound on the maximal average maturity depends on $(c_{i,1}, c_{i,2})$; (ii) for the calibrated parameters associated with atretic follicles, the equilibrium of the characteristic curves lies within, or numerically close to, the high-cell-loss regime; and (iii) the calibrated model admits scenarios in which every follicle reaches its follicle-specific upper bound for $\bar{s}_i$ (see also Figure~\ref{fig:aymardsoverx}).
In our view, these observations suggest that the calibrated follicular parameters $(c_{i,1},c_{i,2})$ may effectively pre-identify the qualitative fate of individual follicles, distinguishing ovulatory---or, more precisely, non-declining-mass---dynamics from atretic or declining-mass dynamics.
Thus, this distinction may be determined largely by intrinsic, follicle-specific parameters rather than emerging solely from follicular competition. See also \citep{aymard16,dimred}.

In our current model, we aim to reproduce the observation that follicles approach a maximal average maturity, while normalizing this value to be the same across all follicles. In particular, Theorem \ref{onestablefixed} shows that, in our model, all follicles can asymptotically approach the maximal average maturity $\bar{s}_i=1$. Thus, our modeling choice amounts to normalizing the follicle-dependent maximal average maturity by fixing $c_{i,2}=1-c_{i,1}$. This simple normalization of the maximal average maturity reproduces the macroscopic mass and maturity dynamics (see Fig.~\ref{fig:celldyn}), and we therefore do not expect it to affect the overall mechanisms underlying the calibrated dynamics.
With this choice, we eliminate the parameter $c_{i,2}$, thereby simplifying the dynamics, and ensure that all follicles approach their maximal average maturity, which we normalize to $\bar{s}=1$.\footnote{If one is interested in the effect of an extra parameter $c_{i,2}$, one can replace the maturation function as
\begin{equation}
g_i\left(m,x,s\right)=c_{i,2}\theta\left(s\right)+m\left[-\frac{s_i}{x_i}+c_{i,1}\theta\left(s\right)\right]\implies\dot{\bar s}_i=-\bar s_i^2+\left[c_{i,1} \bar s_i +c_{i,2}\right]\theta\left(s\right). \label{footnoteeqn}
\end{equation}}

Following the same argument used to replace nonlinear terms in $m$ with functions of $\bar{s}_i$ (see the comment following \eqref{explicitchar}), we model the apoptosis term as
\begin{equation} \label{apopaymard}
\lambda_{i,0}=K_i \exp\left[-\left(\frac{\bar s_i-s_{i,m}}{s_{i,v}}\right)^2\right]\left(1-\frac{U\left(\sum_j \bar s_j\right)}{U_M}\right),
\end{equation}
where $U\left(\sum_j\bar  s_j(t)\right)$  represents a sigmoid-like function with respect to time,
while $s_{i,v}$, $s_{i,m}$, and $K_i$ are constant parameters. At this point, we further simplify the functional forms of $\theta$ and $U$ by choosing
\begin{equation} \label{stangehill}
\theta(\bar s)=\frac{1}{N}\sum_j \bar s_j, \quad U\left(\sum_j \bar s_j\right)=U_M-\left(U_M-U_m\right)\frac{1}{N}\sum_j \bar s_j,
\end{equation}
where $N$ is the number of follicles, $U$ denotes the available FSH level, and $U_M$ and $U_m$ denote its maximal and minimal values, respectively. In what follows, we discuss to what extent our choice of $U$ is consistent with the sigmoid-like function used in \cite{aymard16}. An interesting feature of their model is that the average follicular maturity is directly related to FSH levels. 

With these choices, we obtain the system
\begin{subequations}\label{aymardode}
\begin{align}
&\dot x_i=x_i\left[\alpha_{0,i}(1-\bar s_i)\ln(2) -K_i \exp\left[-\left(\frac{\bar s_i-s_{i,m}}{s_{i,v}}\right)^2\right]\left(1-\frac{U_m}{U_M}\right)\frac{1}{N}\sum_j \bar s_j\right], \label{eqn:aymardode-x}\\
&\dot{\bar s}_i=-\bar s_i^2+\left[c_{i,1} \bar s_i -c_{i,1}+1\right]\frac{1}{N}\sum_j \bar s_j. \label{eqn:aymardode-s}
\end{align} 
\end{subequations}
Despite \eqref{aymardode} being a nonlinear system of $2N$ equations, the average maturity variables $\bar s_i$ in \eqref{eqn:aymardode-s} are independent of the mass variables $x_i$ and therefore satisfy a much simpler system of $N$ coupled equations. In Theorem \ref{thm:aymard} and Proposition \ref{prop:contaymard}, we prove the continuity of the average maturity equation.
Moreover, in Theorem \ref{onestablefixed}, we show that the maximal average maturity $\bar{s}_{\mathrm{max}}=1$ is the unique locally stable fixed point for all follicles.
Therefore, under suitable conditions, all follicles can asymptotically approach their corresponding maximal maturity, which in our normalized formulation is $\bar{s}_{\mathrm{max}}=1$, in agreement with the qualitative behavior of the calibrated model.

Since the equations for $\bar s_i$ are independent of the mass variables $x_i$, we can first solve \eqref{eqn:aymardode-s} for $\bar s_i$. Substituting this solution into \eqref{eqn:aymardode-x} then yields $x_i$, and hence the total maturity $s_i(t)=x_i(t)\bar s_i(t)$.
For instance, in the case $N=1$, the average maturity equation has the solution
\begin{equation} \label{s1follicle}
\bar s_1(t)=\frac{\bar s_1(0) e^t}{\bar s_1(0)e^t+(1-\bar s_1(0))e^{c_1 t}} \implies x_1(t)=x_1(0)\exp{\int_0^t} \left[-\lambda_1\left(\bar s_1(\tau)\right)+p_1\left(\bar s_1(\tau)\right)\right]\ud\tau.
\end{equation}
The average maturity in \eqref{s1follicle} has a sigmoid-like form, which motivates the functional form of our ansatz for $U$ in \eqref{stangehill}. Moreover, the follicle size $x_1$ at time $t$ depends on its initial size $x_1(0)$: at any finite time, the integral in the exponent in \eqref{s1follicle} is finite and independent of $x_1(0)$, so different initial sizes lead to different values of $x_1(t)$.

In the general case of $N$ follicles, near the locally stable point $\bar s_i=1$ for all $i$,\footnote{Note that $\bar s_i=1$ for all $i$ is the only locally stable equilibrium point for the average maturity equation. See Theorem \ref{thm:aymard} and Theorem \ref{onestablefixed}.} the characteristic equation behaves as
\begin{equation}\label{simaymardmat}
\dot m^{(i)}_{a}\big\vert_{\bar s_i=1\, \forall i}=g_i\left(m^{(i)}_{a},x,\bar s\right)\big\vert_{\bar s_i=1\, \forall i}=(1-m^{(i)}_{a})(1-c_{i,1})\implies m_a^{(i)}(t)=  1-(a-1) e^{-t( 1-c_{i,1} )}
\end{equation}
leading to characteristics curves $m^{(i)}_{a}(t)$ asymptotically converging to $m^{(i)}_{a}(t)\to 1$ with zero velocity. From the biological perspective, this indicates the synchronization of cellular maturity near maximal maturity and this phenomenon is also observed by the shrinking of $\mathrm{supp}\tilde \phi_i$ along the $m$-direction in Fig. A.3 in the Supplemental Material of \cite{aymard16}.

In Figure \ref{fig:celldyn}, we show the dynamics of the $10$ follicles, with random initial data and parameters summarized in the figure and its caption.
\begin{figure}[t]
\centering
\includegraphics[width= .32\textwidth]{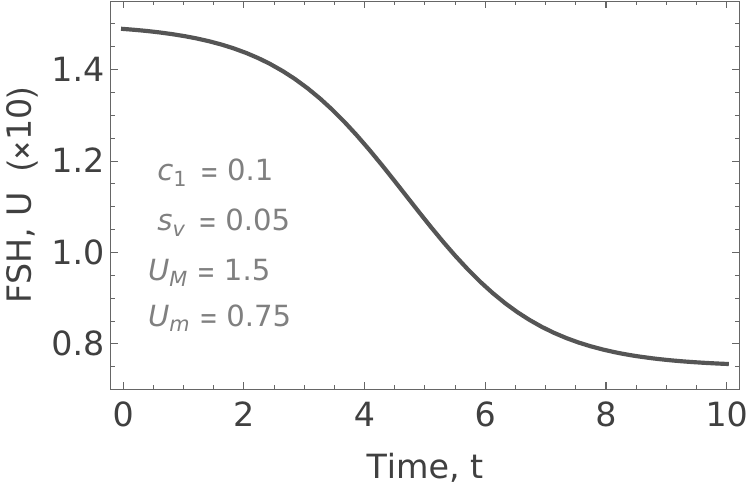}
\,
\includegraphics[width= .32\textwidth]{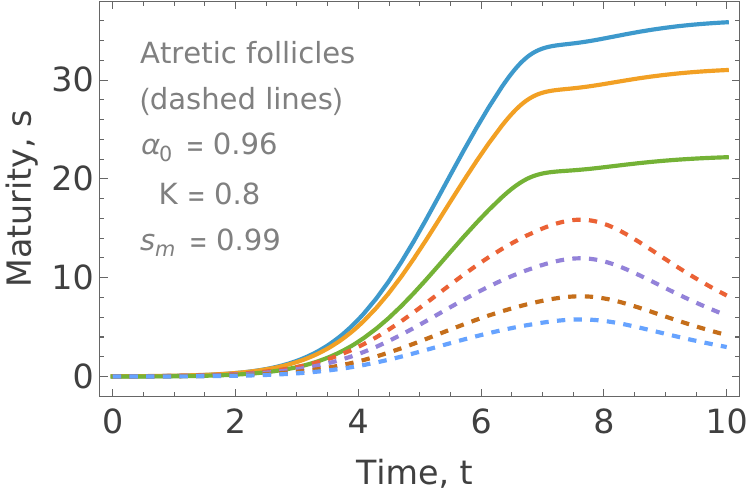}
\,
\includegraphics[width= .32\textwidth]{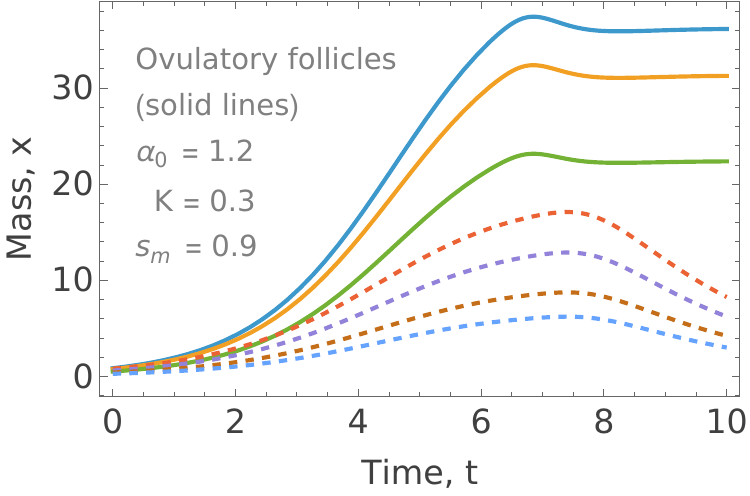}
\caption{Follicle dynamics of the truncated SISYPHE model. Depicted are $7$ follicles solving \eqref{aymardode}. Initial values are chosen randomly: $\bar s_i(0)\in (0,0.03)$ and $x_i(0)\in(0.2,0.9)$. The scales are arbitrary; a realistic mass scale, for example, would be defined by millions of cells. Parameter values used for ovulatory, atretic, and for all follicles are shown in the diagrams (right, center, and left, respectively). The FSH level (left) is given by the $U(t)$ function in \eqref{stangehill}. Notice the similarity of the resulting $m,x$-trajectories with those obtained by \cite{aymard16}, Fig.~6. 
}
\label{fig:celldyn}
\end{figure} 
We see that the resulting dynamics resembles the one obtained in the calibrated models for the macroscopic observables (FSH level, maturity and mass).

Let us emphasize some differences between the model presented here and the calibrated model. An important difference concerns the behavior of the average-stage trajectories. As discussed above, our model can reproduce the fact that all follicles reach the same maximal average stage, which we normalize to $\bar s_i=1$. As a consequence of this normalization, for all follicles sharing the same $c_{i,1}$, the corresponding average-stage trajectories cannot cross in our model. In contrast, in the SISYPHE model, the maximal stage is follicle-dependent through the parameters $(c_{i,1},c_{i,2})$. Moreover, each follicle is characterized by an additional maturation time scale $\tau_i$, see, e.g., \cite{aymard16} and \eqref{aymardcharfull}, which further contributes to the maturation velocity. The combination of a follicle-dependent maximal stage and maturation time scale allows crossings between average-stage trajectories, as illustrated in Figure~\ref{fig:aymardsoverx}.

A further difference concerns the so-called ``rescue'' dynamics observed in some cases in \cite{aymard16}, whereby increasing the initial mass of a follicle can cause a trajectory that would otherwise decline to become non-declining. We do not expect this mechanism to be reproduced by the simplified model considered here. Indeed, as shown in Theorem~\ref{thm:aymard} and Theorem~\ref{onestablefixed}, the average-stage equation admits only one locally stable fixed point. Consequently, all average-stage trajectories are attracted to the same equilibrium, which may lie either within the cell-loss regime or outside it. Thus, in the present model, changes in the initial mass cannot lead to qualitatively different long-term maturation states through convergence to different equilibria. A possible way to recover such rescue dynamics would therefore be to introduce additional locally stable fixed points into the average-stage dynamics, thereby allowing the long-term behavior to depend on the initial condition. Exploring such an extension is left for future work.

Finally, we note that, up to the follicle-specific dependencies on $i$, reproducing the macroscopic dynamics of the calibrated model requires only six parameters, namely $c_1$, $\alpha_0$, $K_0$, $s_m$, $s_v$, and the ratio $U_m/U_M$, compared with the fifteen parameters used in their model: twelve governing the different aging, maturation, and apoptosis functions, and three specifying the subdivision of the domain.

\paragraph{Microscopic behavior}

Beyond the obvious difference between the 1D transport equation considered here and the 2D model proposed by the SISYPHE team, an important distinction between their model and our approach lies in the density
distribution $\phi$.\footnote{The model in \cite{dimred} also analyzes a 1D version of \cite{aymard16,Monniaux2016-yt}; however, no closed ODE reduction is performed. It would be interesting to compare their PDE solutions with those presented here.}
For instance, while our distribution, for suitable initial conditions, only allows for cellular synchronization, the distribution in \cite{aymard16} can develop two or more maturation clusters due to the so-called ``waterproof boundary conditions'' (see Fig.~A.3 in their Supplementary Material). Mathematically, this means that starting from a distribution with connected support, their evolution can produce a distribution with disconnected support.

In our model, the microscopic information is encoded in the following equation:
\begin{equation}\label{aymard_simp}
\begin{split}
&\partial_t \phi_i +\partial_m \left[(1-c_{i,1})\theta\left(\bar s\right)+m(-\bar s_i+c_{i,1} \theta\left(\bar s \right))\right]\phi_i\\&=\phi_i\left[\alpha_{0,i}(1-\bar s_i)\ln(2) -K_i \exp\left[-\left(\frac{\bar s_i-s_{i,m}}{s_{i,v}}\right)^2\right]\left(1-\frac{U_m}{U_M}\right)\theta(\bar s)\right].
\end{split}
\end{equation}

In Figure~\ref{fig:aymardALL} (left panels), we plot the solution of the model with two follicles with the same initial condition
\begin{equation}\label{initpop}
\phi_1(0,m)=\phi_2(0,m)=30\chi_{[0.01,0.3]},
\end{equation}
where $\chi_{[a,b]}$ is the indicator function of the interval $[a,b]$, and with parameters as in Figure~\ref{fig:celldyn}. As anticipated by \eqref{simaymardmat}, we reproduce synchronization at late times through the narrowing of the distribution, reflected in the shrinkage of its support, similar to \cite{aymard16}. This narrowing follows from the dynamics of the characteristics near $\bar s_i= 1$. In this regime,
\begin{equation*}
\dot m^{(i)}_a\big\vert_{\bar s_i=1\, \forall i}=(1-m^{(i)}_a)(1-c_{i,1}),
\end{equation*}
so that maturation becomes increasingly concentrated around $m=1$. In other words, cells with different maturation levels are progressively driven toward the same value, resulting in cellular synchronization, independently of whether the follicle is ovulatory or atretic.

\paragraph{Microscopic behavior: a quadratic approximation}

Up to this point, we have considered a model with a maturation function $g_i$ given in \eqref{aymardmatlin} that is linear in $m$. We now compare its microscopic dynamics with those of a model in which $g_i$, given in \eqref{aymardmatquadratic}, is quadratic in $m$, while keeping the definitions of mass $x$ and follicular maturity $s$ unchanged. 
In doing so, we identify the transformation that keeps the macroscopic output unchanged, even when the maturation function is quadratic rather than linear in $m$. This analysis allows us to bridge the gap between our model and the model in \cite{aymard16}.

In the setting of Theorem~\ref{higherorder}, we can choose the following polynomial:
\begin{equation} \label{aymardmatquadratic}
g_i=g_{i,0} + m g_{i,1} +m^2 g_{i,2}= c_{i,2} \theta(\bar s) + m c_{i,1} \theta(\bar s)-m^2.
\end{equation}
This polynomial resembles the maturation function in the SISYPHE model with respect to its order and structure, up to the functional form of $\theta(\bar s)$ and a time-scale factor $\tau_i$; see \eqref{aymardcharfull}. In order to reproduce the macroscopic output defined in \eqref{aymardode}, we define
\begin{equation}\label{gammas}
\begin{split}
&\Gamma_{i,0}=\alpha_{0,i}(1-\bar s_i)\ln(2) -K_i \exp\left[-\left(\frac{\bar s_i-s_{i,m}}{s_{i,v}}\right)^2\right]\left(1-\frac{U_m}{U_M}\right)\theta(\bar s)-\bar s_i\\
&\Gamma_{i,1}=1.
\end{split}
\end{equation}
See \eqref{eqs:thm4} in Theorem~\ref{higherorder} for the relation  between $\Gamma_{i,j}$ and the growth and death rates. 
With these definitions, 
substituting \eqref{aymardmatquadratic} and \eqref{gammas} into \eqref{infinitesystem} of Theorem~\ref{higherorder},
and focusing only on the zero and first moments $(\alpha=0,1)$, one can verify that the dynamical system for the mass $(\alpha=0)$ \eqref{eqn:aymardode-x} remains unchanged, and that the equation for the average maturity ($\bar s_i=s_i^{(1)}/s_i^{(0)}$) is
\begin{equation*}
\dot{\bar s}_i= -\bar s_i^2 +(c_{i,1} \bar s_i +c_{i,2})\theta(\bar s),
\end{equation*}
which is equivalent to \eqref{eqn:aymardode-s} if choosing $c_{i,2}=1-c_{1,i}$ as proposed earlier.

We again normalize the maximal average maturity to $\bar s_i=1$ by fixing $c_{i,2}=1-c_{1,i}$. Therefore, the characteristic equation near $\bar s_i= 1$ reads
\begin{equation}
\begin{split}
\begin{dcases}
&\dot{m}_a^{(i)}\big\vert_{\bar s_i=1\, \forall i}= (1-c_{i,1}) -m_a^{(i)}(m_a^{(i)}-c_{i,1})\\
&m_a^{(i)}(0)=a
\end{dcases}
\end{split}
\end{equation}
which has an explicit solution given by
\begin{equation}
m_a^{(i)}(t)=\frac{(c_{i,1}-1) (a-1) e^{(c_{i,1}-2) t}+c_{i,1}-a-1}{(a-1) e^{(c_{i,1}-2) t}+c_{i,1}-a-1}.
\label{eqn:simaymardmat-square}
\end{equation}
Similar to \eqref{simaymardmat}, \eqref{eqn:simaymardmat-square} shows the convergence $m_a^{(i)}(t)\to 1$ as $t\to\infty$, i.e, synchronization effect towards maximal maturity.

\begin{figure}[t!]
\centering
\hspace*{.5em}
{\sf\small Linear maturation function} \hspace{1.3em}  {\sf\small Quadratic maturation function}
\vspace{.2cm}

\includegraphics[width=.62\textwidth]{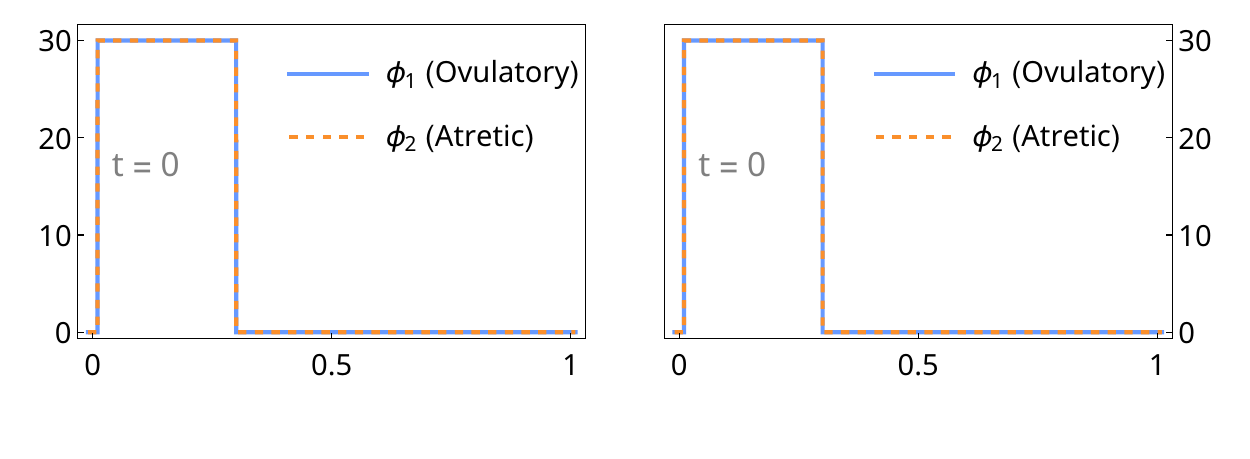}\\[-4ex]
\includegraphics[width=.68\textwidth]{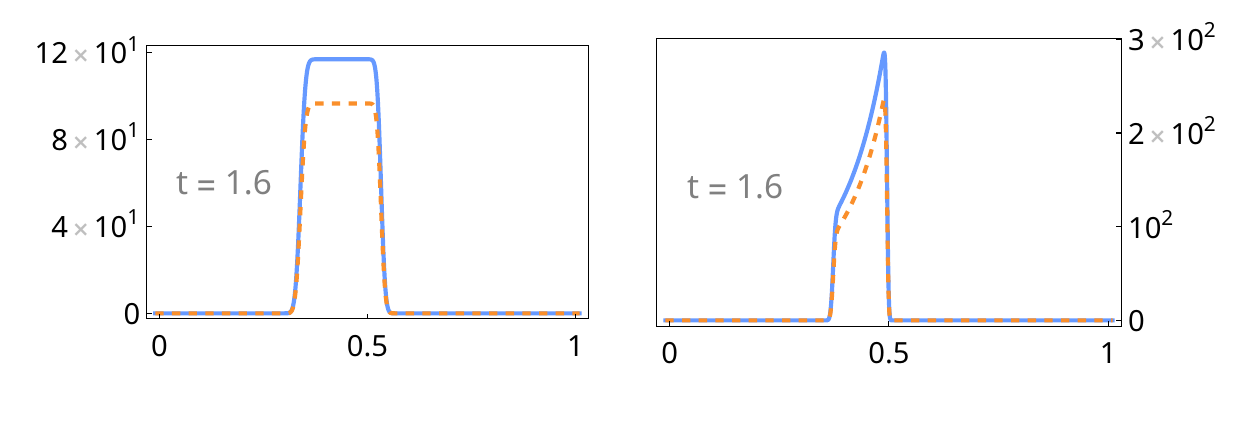}\\[-4ex]
\includegraphics[width=.7 \textwidth]{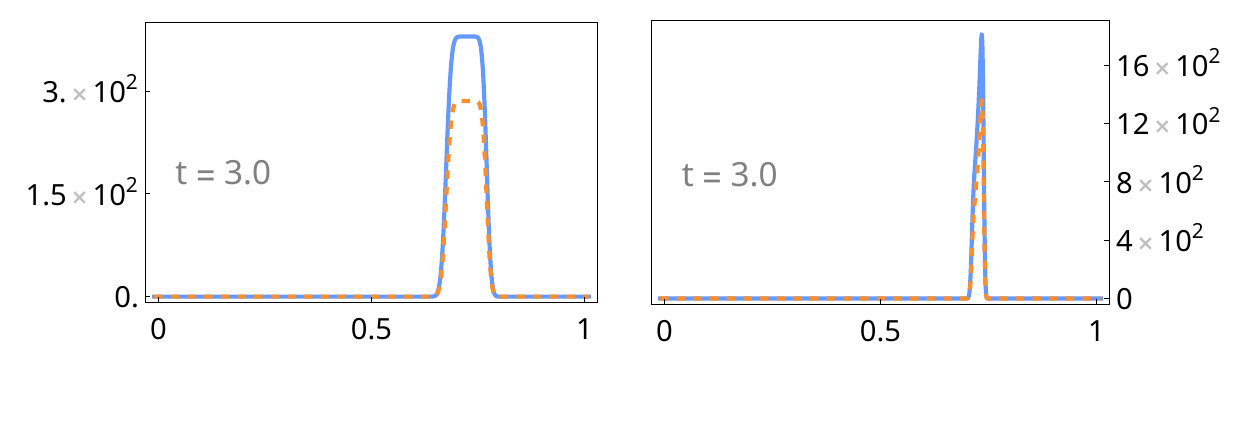}\\[-4ex]
\includegraphics[width=.7 \textwidth]{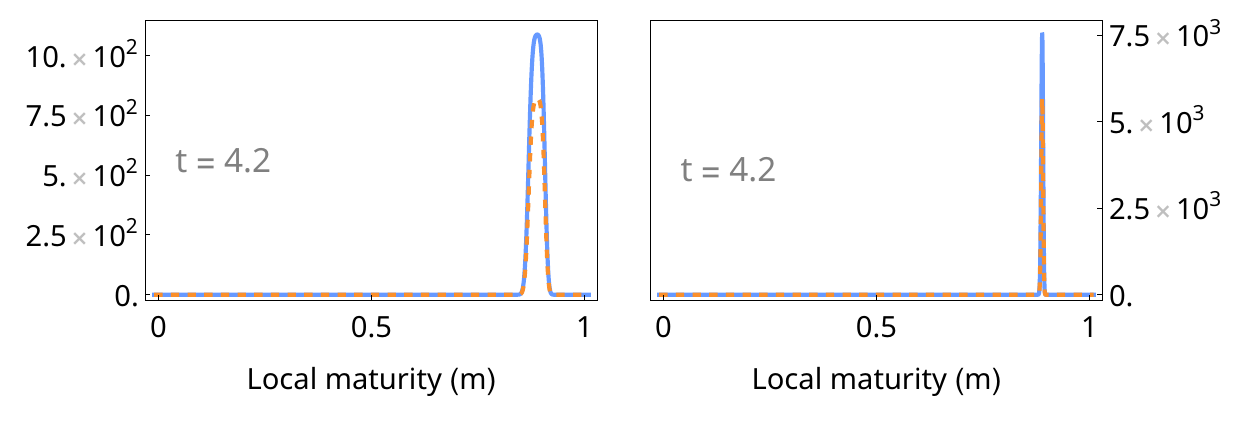}\\[2ex]
 \caption{Snapshots of the follicle maturity density at four consecutive time points. Each panel displays the densities of an atretic follicle ($\phi_2$, dashed red curve) and an ovulatory follicle ($\phi_1$, solid blue curve). Trajectories represent numerical solutions to \eqref{aymard_simp} obtained via a finite volume method paired with a fourth-order Runge-Kutta time-integration scheme (detailed in Section~\ref{num}). Left panels illustrate the dynamics under a linear maturity function \eqref{aymardmatlin}, while right panels show the dynamics for a quadratic maturity function \eqref{aymardmatquadratic}. The initial conditions and baseline parameters correspond to \eqref{initpop} and Fig.~\ref{fig:celldyn}, respectively. Time progresses chronologically from top to bottom (expressed in arbitrary units).}
\label{fig:aymardALL}
\end{figure}

Figure~\ref{fig:aymardALL} (right panels) presents numerical simulations for the quadratic maturation functions. In both linear and quadratic cases, the distributions narrow towards the final time, indicating cell synchronization. However, their evolutions at intermediate times differ between the two cases. In the case of the quadratic maturation function, the quadratic term $m^2$ breaks the symmetry of the distribution, predicting that highly mature subpopulations are more susceptible to reproduction than less developed cells. A rigorous comparison of the effects of linear and quadratic maturation functions on cellular synchronization, which may not be immediately apparent from Figure~\ref{fig:aymardALL}, is provided in Section~\ref{linvsquad}.

Despite the lack of histomorphological data, our framework is sufficiently versatile to simplify the PSPM proposed by \cite{Monniaux2016-yt}, while recovering both its macroscopic dynamics (i.e., the total number of granulosa cells) and some of its predicted histomorphological features, such as the trend toward maximal normalized local maturity, $m_a^{(i)}\to 1$, as well as cellular synchronization toward the maximal average stage.

\subsection{Modeling size and maturity} 
\label{langeconn}

In the literature, ODE models describing individual follicles usually have the form
\begin{equation}\label{genODE}
\dot x_i= x_i \,\Gamma_i\left(x_i, \sum x_j\right),
\end{equation}
where $\Gamma$ is a smooth function, and
$x_i$ either represents the estradiol produced by the $i$-th follicle or its size.
Examples in which $x_i$ models estradiol concentrations were proposed by \cite{Lacker1981-js,Mariana1994,SOBOLEVA200045}, while
examples in which $x_i$ represents follicular sizes were given by \cite{Lange2018-xf,shilo}.
Here, we follow the more recent models, and in concrete applications we identify $x_i$ with follicular diameter.

To construct a population model and understand how the PDE framework encodes information about follicular size, we start with the ansatz
\begin{equation}\label{PDEgenODE}
\partial_t \phi_i + \partial_m \left[g_i(m,x, s) \phi_i \right]=\Gamma_i\left(x_i, \sum x_j\right)\phi_i.
\end{equation}
Provided that solutions to \eqref{PDEgenODE} exist (a0) and that the assumptions (a1, a2) made in Section \ref{pspmtoode} are satisfied, one can immediately conclude---by substituting $\Gamma_i=p_i-\lambda_i$ and inspecting \eqref{finalmasseqn} in Theorem~\ref{odethm}---that the zeroth-order moment satisfies the ODE system \eqref{genODE} for the follicular sizes,
\begin{equation*}
x_i(t)=s_i^{(0)}(t)\equiv\int_0^1\phi_i(t,m)\ud m.
\end{equation*} 
Needless to say, information is lost in such a reduction, as we focus only on the zeroth-order moment rather than on the distribution of cells over the maturity variable $m$. 
In particular, even the average follicular maturity can no longer be inferred.

Nevertheless, models like
\eqref{genODE} work perfectly adequate if one is
only interested in macroscopic variables such as follicular sizes.
The model introduced by \cite{Lange2018-xf}, 
\begin{align} \label{langeeq}
\dot{x}_i = x_i (\xi - x_i) \left[ \gamma - \kappa \left( \eta x_i^\nu + \sum_{j \neq i} x_j^\nu \right) \right],
\end{align}
which we will employ here, is a good example  
where solutions to \eqref{langeeq} were found to fit bovine data surprisingly well; see Figures~\ref{fig:bovine419} and \ref{fig:bovine3674}.
The applicability of \eqref{langeeq} extends beyond mono-ovulatory mammals, as the ratio between growth $\gamma$ and competition $\kappa$  (i.e., $d=[\gamma/\kappa/\xi^\nu+1-\eta]$, cf.~Eq.~(10) in Lange et al.) determines the number of dominant follicles and can, of course, be greater than one.
Remarkably, 
this number is independent of both the total number of growing follicles and their initial sizes, in a mathematically precise manner. 

Even if this model
does not incorporate interaction with the endocrine system, it is utilized as part of
the model studied by \cite{FISCHERHOLZHAUSEN2022111150}, where
parameters are coupled to the hormonal environment. 
Despite accurately modeling the first follicular wave, the model by Lange et al.~fails to produce crossing trajectories,
a phenomenon we have observed in the analyzed bovine data (e.g., cow no.~419, second wave).
To obtain crossing trajectories, Fischer-Holzhausen and Röblitz, initialize follicular growth at random times and with random sensitivities to FSH.\footnote{On this note,
we have observed that the SISYPHE model and, more general, PSPMs allow for crossing trajectories, which we will explore in an upcoming work \citep{next}.}

When reproducing the ODE model by Lange et al.~from a PSPM, we interpret the density $\phi_i(t,m)\ud m$ as the portion of the follicle, measured on the scale of its diameter, whose local maturation level lies between $m$ and $m+\ud m$. The size dynamics in \eqref{langeeq} is recovered by assuming that the death and proliferation rates are independent of the follicular maturity $s_i$, such that
\begin{equation}\label{gamme_lange}
\Gamma_i\left(x_i,\sum x_j\right)=(\xi-x_i)\left[\gamma-\kappa\left(\eta x_i^\nu+\sum_{j\neq i} x_j^\nu\right)\right].
\end{equation}   
%As already discussed, 
The model by Lange et al.~lacks a maturity variable $s_i$ that would indicate the developmental stage of a follicle (preantral, antral, dominant, atretic, or somewhere in between), leaving the maturation function $g_i$ in the corresponding PSPM unspecified. We therefore proceed under the following modeling assumptions.
\begin{itemize}
\item  The average maturity equation of the $i$-th follicle does not directly depend on the variables describing the other follicles. Biologically, this reflects the assumption that cells belonging to different follicles do not directly compete or interact with one another. Instead, each cell population responds only to its local environment and to the state of the follicle to which it belongs. Mathematically, this 
means that the maturation function $g_i$ depends only on the size and maturity of the $i$-th follicle, i.e.
\begin{equation}
g_i=g_i(t,m,x_i,s_i).
\end{equation}
We note that such an assumption has not been made by the SISYPHE team.
\item Furthermore, we assume that the maturation coefficients do not depend on the size $x_i$, i.e., $g_i=g_i(t,m,s_i)$, and that $\bar s_i=0$ and $\bar s_i=1$ are equilibrium points for the average maturity variable. This can be ensured by choosing maturation functions satisfying
\begin{equation}\label{genmatfunct}
\begin{dcases}
&g_{i,0}(t,s_i) + \bar s_i g_{i,1}(t,s_i)= f_i(t,1-\bar s_i) h_i(t,s_i)
\\ &  g_{i,0}(t,0)\geq 0
\\
&f_i(t,0)=0, \quad f_i(t,r)\in C^\infty([0,T], [0,1]), \quad h_i(t,s_i)\in C^\infty([0,T], \mathbb R_+).
\end{dcases}
\end{equation}
Indeed, the first condition in \eqref{genmatfunct} implies that $\dot{\bar s}_i(t)=0$ whenever $\bar s_i=1$, while the second condition ensures that $\dot{\bar s}_i(t)\geq 0$ at $\bar s_i=0$. Hence, $\bar s_i(t)\in[0,1]$ for all times whenever $\bar s_i(0)\in[0,1]$.

As already noted, the assumption that the maturation dynamics are independent of the %total 
follicular
size also %appears in 
applies to the SISYPHE model.
\item For technical reasons, we assume that
\begin{equation}\label{eq:gInCInf}
g_{i,0}(t,s_i),\, g_{i,1}(t,s_i),\in C^{\infty}([0,T]\times[0,\xi]).
\end{equation}
Because of 
(1st) %(increasing)
monotonicity and boundedness of  $x_i(t)\in(0,\xi]$ (cf.~Lange et al.), which implies  Lipschitz continuity $x_i(t)\in C^1([0,T];(0,\xi])$, (2nd) continuity of the maturation functions \eqref{eq:gInCInf},
and (3rd) boundedness of $\bar s_i(t)\in[0,1]$,%
\footnote{Because of \eqref{genmatfunct}, the equation $$\bar s_i= g_0(t,s) + \bar s\,g_1(t,s)= g_0(t,x\bar s) + \bar s\, g_1(t,x\bar s)=\tilde g_0(t,\bar s) + \bar s\, \tilde g_1(t,\bar s)$$ leads to continuous and bounded $\bar s_i$, which makes $\bar s$ bounded, and $\tilde g_i$ are continuous functions of $(t,\bar s)$.
}
we have ensured
that $\bar s_i(t)\in C^1([0,T];[0,1])$ 
is continuous as well and $s_i(t)\in C^1([0,T];[0,\xi])$, for all $t\geq0$.
By bootstrap arguments, i.e., analyzing higher order derivatives of $x_i(t)$ and $s_i(t)$, the regularity results for size and maturity can be improved. 
\end{itemize}

Before we discuss explicit examples, we recall that both the follicular radius $x_i$ and the follicular maturity $s_i$ have the dimension of length. While this is obvious for the radius, this is not so clear
for the maturity $s_i$. From a general perspective, the maturity $s_i$ indicates the developmental stage of a given follicle, i.e., its functional, structural, and morphological properties. In the present setting, however, $s_i$ can be interpreted as an ``effective radius": unlike the geometric quantity $x_i$, it encodes the stage of the follicle and thus provides a more physiologically appropriate quantity to couple with the endocrine system. For instance, there can be large but atretic follicles whose influence would be overestimated if the coupling were based solely on the geometric size. In other words, the stage variable $s_i$ tracks the physiological structure of the follicle, rather than its mere geometrical size.

In following two subsections, we study
two simple examples of maturation functions $g_i$ and show what type of stage dynamics they encode.

\paragraph{Maturity-independent cellular evolution} 

We start with the simplest possible example,    
where the maturation functions (i.e., the coefficients of the population model \eqref{pdesystem}) are constants, 
\begin{equation}\label{simple_lange}
g_{i,0}(x,s)=\alpha, \qquad g_{i,1}(x,s)=-\alpha.
\end{equation}
With this choice, the characteristic equation \eqref{explicitchar} and the average stage equation \eqref{uncoupledmain-s}  read,
\begin{equation*}
\dot m^{(i)}_{a}= g_i(m^{(i)}_{a},x,s)=\alpha(1-m^{(i)}_{a}), \quad \dot{\bar s}_i=\alpha(1-\bar s_i).
\end{equation*}
These equations imply that the characteristic curves reach the boundary only at infinite time.
This model mimics the cellular maturation dynamics in the final stages of follicular development as proposed by the SISYPHE model, where the cell population is not regressing at a local maturation level. 
All cells mature towards a normalized maximal average stage, $\bar s_i\to1$. The follicle's death is only modeled by the loss of size, i.e., the decay in the density $\phi$, and not via cell regression. 
This means that even if all follicles have maximal average maturation, not all follicles will have the correct ovulatory size $x_i=\xi$. 
    
Therefore, the maturation \eqref{simple_lange} together with \eqref{gamme_lange}, lead to the PSPM
\begin{equation*}
\partial_t \phi_i + \partial_m\left[(1-m)  \phi_i\right]= (\xi-x_i)\left[\gamma-\kappa\left(\eta x_i^\nu+\sum_{j\neq i} x_j^\nu\right)\right]\phi_i,\qquad i=1,2,\dots, N,
\end{equation*}
which recovers
the size equation of
Lange et al.~%\cite{Lange2018-xf}
but yields a new stage equation,
\begin{equation*} 
\begin{cases}
\dot x_i&= x_i(\xi-x_i)\left[\gamma-\kappa\left(\eta x_i^\nu+\sum_{j\neq i} x_j^\nu\right)\right]\\
\dot s_i&= s_i(\xi-x_i)\left[\gamma-\kappa\left(\eta x_i^\nu+\sum_{j\neq i} x_j^\nu\right)\right]+ \alpha(x_i- s_i).
\end{cases}
\end{equation*}
We recall that the relevant equilibrium points for the size equation correspond to individuals of maximal size \( x_i = \xi \) and vanishing size \( x_i = 0 \). These sizes are associated with the maximal stage \( s_i = \xi \) and the vanishing stage \( s_i = 0 \), respectively.

\paragraph{Maturity-dependent cellular evolution}  

Another simple example is given by linear maturation functions,
\begin{equation}\label{nicematfunction}
g_{i,0}(x,s)=\alpha s_i, \qquad g_{i,1}(x,s)=-\alpha s_i.
\end{equation}
The corresponding characteristic and average maturity equations read
\begin{equation} \label{nicechar}
\dot m^{(i)}_{a}= g_i(m^{(i)}_{a},s_i)=\alpha(1-m^{(i)}_{a})s_i \qquad \dot{\bar s}_i=\alpha (1-\bar s_i) s_i.
\end{equation}
When starting with reasonable initial conditions (for the average maturity, i.e., $0<s_i(0)/x_i(0)\equiv\bar s_i(0)\le1$), one can prove that $0<s_i(t)<x_i(t)<\xi$, for all $t\ge0$;
see Proposition \ref{propstage}.
Furthermore, one concludes that for an atretic follicle (i.e., $s_i\to 0$), its cell population stops developing and the average stage monotonically approaches a constant value, 
$\bar s_i(t)\to\bar s_i^c$ as $t\to\infty$. That is, in this model, atretic follicles 
lead to a constant average maturity,
which is similar to the SISYPHE model (considering the cell population), but---in contrast to the SISYPHE model---its value is not necessarily maximal. This captures the common understanding  that the microscopic structures ``dies in place'' rather than regressing to an ``earlier'' developmental stage. From a histomorphological perspective, the population retains its developed structure, even though individual cells are continuously removed from the system by different clearing mechanisms.
See \cite{regan18,Stringer2023-nz} and the references therein for different apoptotic and atretic mechanisms. 

If a follicle reaches the equilibrium at $s_i=\xi$, or its maturity does not vanish, $s_i\neq0$, for a sufficiently long time, then the cells in that follicle
approach the maximal value, $m^{(i)}_{a}=1$.
Due to the defining maturation functions \eqref{nicematfunction}, many follicles can
approach $\bar s_i=1$, even if $x_i<\xi$. This means that for such small-sized follicles the average maturity is maximal, but the their small size---consistent with a low number of cells---makes these follicles atretic.
As a consequence, ovulatory (or dominant) follicles are only characterized by their total size, which has to be close to $x_i=\xi$.       

The PDE for this model is given by
\begin{equation} 
\partial_t \phi_i + \partial_m\left[\alpha(1-m)s_i  \phi_i\right]= (\xi-x_i)\left[\gamma-\kappa\left(\eta x_i^\nu+\sum_{j\neq i} x_j^\nu\right)\right]\phi_i,\qquad i=1,2,\dots, N, \label{langepspm}
\end{equation}
which leads to 
\begin{subequations} \label{langebas}
\begin{empheq}[left=\empheqlbrace]{align}\label{langebas-m}
\dot x_i&= x_i(\xi-x_i)\left[\gamma-\kappa\left(\eta x_i^\nu+
\textstyle \sum_{j\neq i} x_j^\nu\right)\right]\\
\dot s_i&= s_i(\xi-x_i)\left[\gamma-\kappa\left(\eta x_i^\nu+\textstyle \sum_{j\neq i} x_j^\nu\right)\right]+ \alpha s_i (x_i- s_i).
\label{langebas-b}
\end{empheq}
\end{subequations}

\begin{figure}[t]
\centering
\includegraphics[width=0.45 \textwidth]{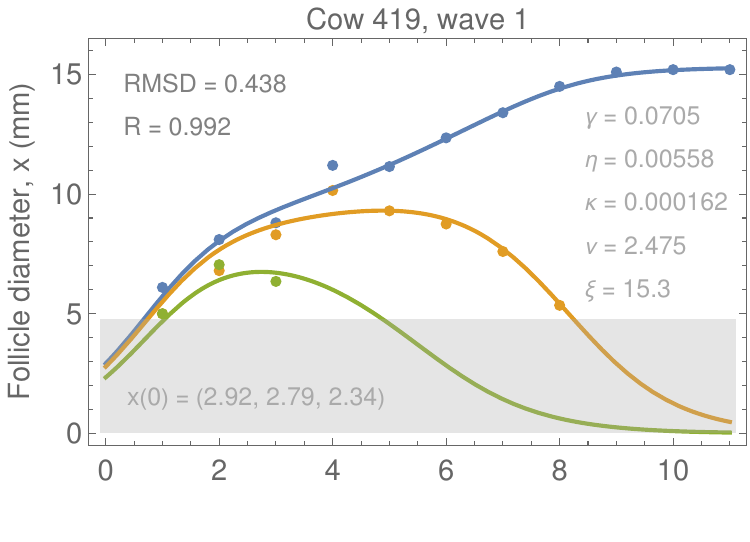} \quad
\includegraphics[width=0.45 \textwidth]{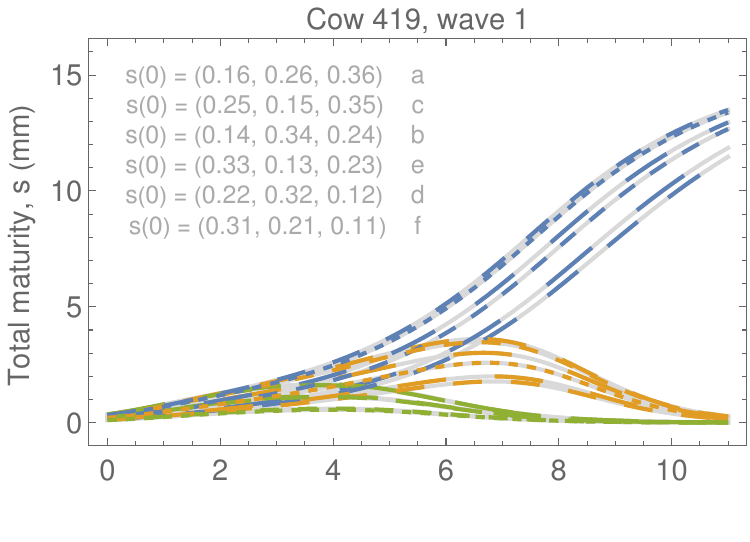}
\\[-4ex]
\includegraphics[width=0.45 \textwidth]{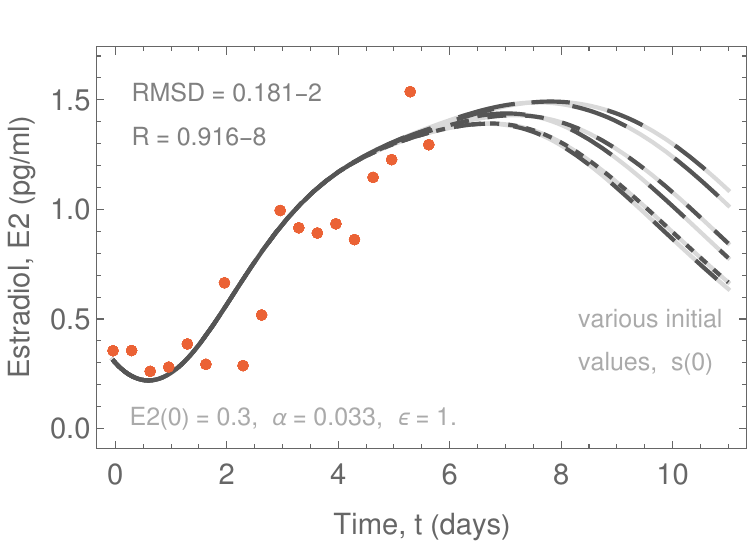} \quad
\includegraphics[width=0.45 \textwidth]{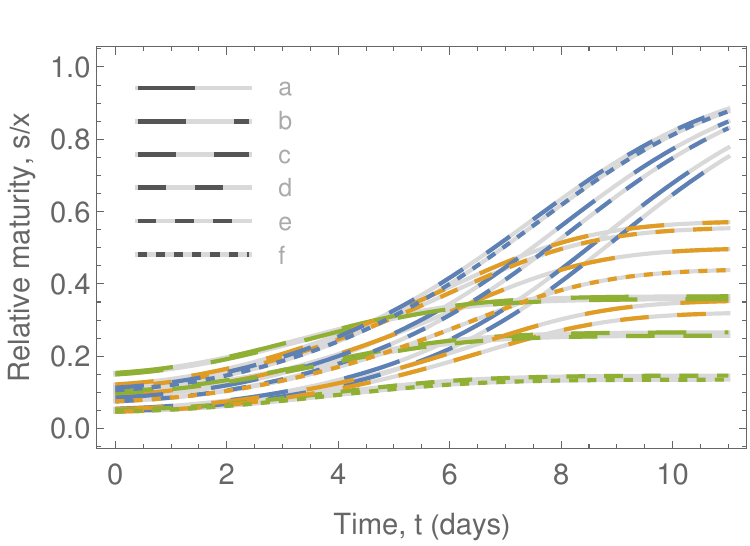} 
\\[2ex]
\caption{Fitting ovarian follicle sizes and blood estradiol concentrations to cow data. Depicted is the rising part of the first wave past ovulation in cow no.~419 from \cite{Cummins2012-oa}.
Follicle sizes $x_i(t)$ over time (depicted in the top left panel) have already been analyzed by \cite{Lange2018-xf}, yet with additional depletion terms covering a longer time period; fitted parameters (printed next to the curves) therefore do not coincide, and the quality of the fit, evaluated through the rooted mean square deviations (RMSD) and the correlation (R), appears slightly better here.
Based on the associated follicular maturity curves, $s_i(t)$ or equivalently $\bar s_i(t)$ (top and bottom right, respectively), simulated for six  (a,b,...,f) initial value triples (cf.~top right panel),
estradiol blood data (left bottom) have been fitted via an estimate for the effective concentration of androgens \eqref{EqEffAndro} and the equation and parameters \eqref{EqE2BovCycle} modeling estradiol 
in the bovine estrous cycle
of \cite{STOTZEL20121415}.
In Appendix \ref{additionalcow}, we provide data and similar fitting results for another cow (Fig.~\ref{fig:bovine3674}).
In contrast to the follicle size (top left), one observes that the curves representing the total and the relative cellular maturation (top and bottom right, respectively) may cross each other.
Note that the modeled follicle sizes are independent of the corresponding maturation. 
The gray area indicates diameters where follicles cannot be detected by ultrasound.}
\label{fig:bovine419}
\end{figure}    

In Figure \ref{fig:bovine419}, we present an example of three bovine follicles (forming the first wave after ovulation) based on the data in \cite{Cummins2012-oa}. 
As anticipated, the size dynamics is unchanged despite the extra follicular maturity variable, $\dot{ s}_i$. 
The maturity equation \eqref{langebas-b} shows that the dominant follicle, %with $x_i=\xi$, 
for which $x_i\to\xi$,
approaches the maximal maturity, %$\bar s_i=1$.
$\bar s_i\to1$.
In this limit, \eqref{langebas-b} reduces to
\begin{equation*}
\dot s_i= \alpha s_i(\xi-s_i),
\end{equation*}
which is a logistic growth equation with characteristic timescale $1/(\alpha\xi)$. The other atretic follicles stop developing at smaller values of average maturity, modeling atresia as a coordinated process in which both geometric and functional components of the follicle degenerate in parallel. 
     
We would like to mention
that for the size dynamics, %as
adopted from the model in Lange et al., a ``no-crossing" property holds for the resulting
size trajectories.
Therefore, dominant and atretic follicles can be predicted entirely by their initial sizes, independently of the initial values of the maturity variables. 
In an upcoming work \citep{next}, we will elaborate on how the size equation can be improved by including knowledge about the level of maturity. 
We are optimistic that refinements of this kind will, besides improving the data fit, allow for 
modeling clinical phenomena like PMOS (in the case of women) via the fixed point structure of the equations.

\paragraph{Microscopical behavior: desyncronization effects}

Now we will have a look at the microscopical content of the PSPM associated with the size-stage dynamics in \eqref{langebas}. 
To this end, we study the solutions of
\begin{equation} \label{langepde}
\begin{split}
\partial_t \phi_i + &\partial_m\big\{\left[\alpha(1-m)s_i -\beta\left(\xi-x_i\right)(\bar s_i- m  )\right]\phi_i\big\}\\&\hspace{1cm}= (\xi-x_i)\left[\gamma-\kappa\left(\eta x_i^\nu+\textstyle \sum_{j\neq i} x_j^\nu\right)\right]\phi_i,\qquad i=1,2,\dots, N,
\end{split}
\end{equation}
which is a modified version of \eqref{nicechar}, where the microscopic dynamics is changed via
the transformation \eqref{gtransfo} with
\begin{equation*}
f_i=-\beta\left(\frac{\xi}{x_i}-1\right),
\end{equation*}
but the macroscopic (size-maturity) dynamics \eqref{langebas} stays the same.
Furthermore, we will analyze different synchronization and  desynchronization effects 
of the density function $\phi_i$ and its support
induced by the parameter $\beta$.

In the case of $\beta=0$, for atretic follicles (i.e., $s_i\to0$), the characteristic equation reduces to
\begin{equation*}
\dot m^{(i)}_a=0,
\end{equation*}
so that the support remains constant with respect to the maturation variable and only partial synchronization occurs. In contrast, for 
dominant follicles, the support of $\phi$ contracts towards large time and maturity $m\to1$, leading to full cellular synchronization. Thus, while atretic follicles retain their dispersion in maturation as they undergo atresia, cells in dominant follicles progressively synchronize toward maximal maturity.
 
\begin{figure}[t]
\centering
\includegraphics[width= .9\textwidth]{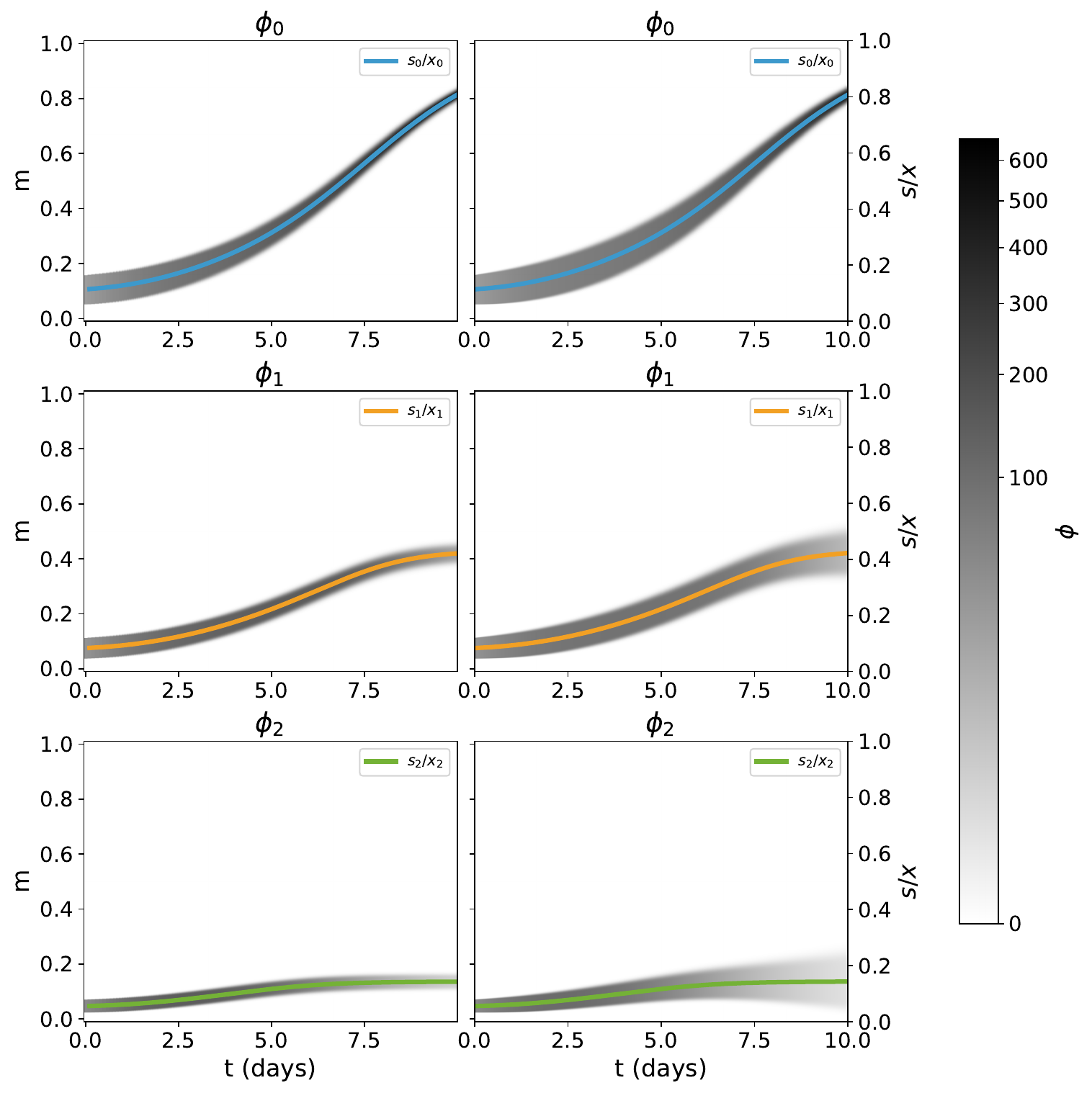}
\caption{Heat maps of three-follicle dynamics in \eqref{langepde}, with parameters as in Figure \ref{fig:bovine419} and initial conditions given by \eqref{init_pde_lange}. For the left panels, $\beta=0$; for the right panels, $\beta=1.2\times 10^{-2}>0$.}
\label{fig:lange}
\end{figure} 
 
On the other hand, for $\beta>0$, the characteristic equation satisfies
\begin{equation}\label{betachar}
\dot m^{(i)}_a=
\alpha\bigl(1-m^{(i)}_a\bigr)s_i
-\beta(\xi-x_i)\bigl(\bar{s}_i-m^{(i)}_a\bigr).
\end{equation}
For each fixed characteristic parameter $\xi$, this is a time-dependent
linear ODE for $m^{(i)}_a$. Since
$x_i(t)$, $s_i(t)$, and $\bar{s}_i(t)$ are bounded over time, the
coefficients of \eqref{betachar} are bounded. Consequently, its solution
$m^{(i)}_a(t)$ is continuous in time.

For dominant follicles (i.e., $x_i\to\xi$, $s_i\to1$), in the limit where $x_i$ and $s_i$ are close to their equilibrium points, the characteristic equation \eqref{betachar} reduces to%
\footnote{Here, '$\sim$' means up to order $O(\epsilon_i)+O(\sigma_i)$ where $0<\epsilon_i=x_i^*-x_i(t)\ll1$ and $0<\sigma_i=s_i^*-s_i(t)\ll1$, with $x_i^*$ and $s_i^*$ being stable equilibrium points.}
\begin{equation*}
\dot m^{(i)}_a\sim\alpha(1-m^{(i)}_a)\ \ \implies \ \ m^{(i)}_a\to1.
\end{equation*}
In contrast, for atretic follicles 
(i.e., $x_i\to0$, $s_i\to0$) satisfying $\bar s_i\to\bar s_i^c<1$, there exists a time $t_0>0$ for which we can approximate the characteristic equation \eqref{betachar} as
\begin{equation}
\dot m^{(i)}_a\sim-\beta\xi(\bar s_i^c-m^{(i)}_a)\implies m^{(i)}_a(t)\sim \bar s_i^c+ e^{\beta\xi t}\left[m_a^{(i)}(t_0)-\bar s_i^c\right].\label{atreticbeta}
\end{equation}
Thus, assuming the support of $\phi_i$ for an atretic follicle remains finite during the final phase, cells are driven away from the average maturation level $\bar s_i^c$, resulting in desynchronization and a progressive expansion of $\operatorname{supp}\phi_i(t,\cdot)$. Because $\beta>0$, the characteristic $m^{(i)}_a$ are not constrained to be bounded between $[0,1]$, which means that in presence of atretic follicles with $\bar s_i^c<1$ the PDE only admits a solution up to a finite time,
$T=\min \big\{t\in \mathbb R_+: \exists j ~m^{(j)}_a(t)\in\{0,1\} \wedge \forall i~m^{(i)}_a(0)\in \mathrm{supp}~\phi_i(0,\cdot)\big\}>0$.

In Figure~\ref{fig:lange}, we show the population dynamics of the three follicles for different values of $\beta$, using initial conditions of the form
\begin{equation}\label{init_pde_lange}
\phi_i(0,m)=\mu_i\chi_{[a_i/3,a_i]},
\end{equation}
where $\chi_{[a,b]}$ denotes the indicator function of the interval $[a,b]$, and
\begin{equation*}
\mu_i=\frac{x_i^2(0)}{s_i(0)},\qquad
a_i=\frac{3s_i(0)}{2x_i(0)}.
\end{equation*}
The numerical simulations were performed using the finite-volume scheme described in Appendix~\ref{finite-volume}, with the corresponding C++ implementation available in the source code repository~\cite{github_follicular-dynamics}.
The initial sizes $x_i(0)$ are taken from the values listed in Figure~\ref{fig:bovine419} (top left), while the initial maturities $s_i(0)$ are taken from list (f) in Figure~\ref{fig:bovine419} (top right).
We observe that \( \beta\neq 0 \) does not significantly affect the dynamics of the dominant follicle (top row), which, at later stages, exhibits clear cellular synchronization. In contrast, for the atretic follicles (middle and bottom rows), a greater dispersion in cellular maturity is evident.
In particular, the middle follicle shows a milder degree of desynchronization compared to the bottom follicle, which can be explained by its relatively high average maturity prior to atresia. Conversely, the bottom follicle remains small and underdeveloped throughout the simulation, with pronounced cellular desynchronization, reflecting its unhealthy physiological state. 

Finally, we observe that the colored lines in Figure \ref{fig:lange}, represent the average maturity $\bar s_i$. In Section \ref{error} we compute the difference between the average stage computed from \eqref{langebas}, using the Mathematica ODE solver, and our numerical scheme for the integration of the PDE \eqref{langepde}. %As expected
We find a good match between both approaches, which is validating our numerical schemes and the ODE reduction.

\subsection{Hormonal feedback: the production of estradiol} \label{sect:BovCycle}

The early mathematical studies on follicular growth, \cite{Lacker1981-js} and \cite{LackerAkin1988}, 
for example, have modeled 
the concentration of
the hormone estradiol 
secreted by individual follicles, not the sizes of individual follicles as we have been doing in the previous sections.
These studies start with a linear ODE of the form,
\begin{equation}\label{Eq1Lacker}
\frac{d}{dt} E2(t) = \sum_{i\le N} \sigma_i(E2(t),t) - c_{E2} E2(t),
\end{equation}
where $\sigma_i$ represents the secretion rate of the $i$-th follicle, $E2(t)$ the resulting estradiol concentration in the blood at time $t$, and $c_{E2}$ a corresponding clearance rate.
Later in their analyses, the authors propose that $\sigma_i$ represents a maturity marker of the $i$-th follicle and, after implementing certain estimates, derive an explicit equation for the estradiol concentration
over time, without $\sigma_i$.

In this section, we as well aim at modeling the estradiol dynamics using an equation similar to \eqref{Eq1Lacker}. 
However, having proposed explicit equations for the size $x_i(t)$ and the maturity $s_i(t)$, our approach follows a different route. 
Both these variables have been claimed to contribute to the secretion rate: $x$ by \cite{Lange2018-xf} and $s$, as mentioned, by Lacker. The model by Lange et al.~requires a discontinuous and rather artificial dependence on $x$ to stop secretion once $x_i(t)$ reaches its maximum, but also the implementation of the maturity concept by Lacker fails to predict the pre-ovulatory estradiol peak and its subsequent drop at ovulation; cf.~\cite{Lu2018-vs}, \cite{LINDNER202686}. To overcome these limitations, one may consider the difference of the two variables.

Given the explicit equations for $x_i$ and $s_i$, one can employ two properties of the difference,
\begin{equation*}
0<x_i(t)-s_i(t)=x_i(t)\Big(1-\bar s_i(t)\Big)\to_{t\to\infty}0\,,
\end{equation*}
positivity and the limiting behavior at large times (cf.~Prop. \ref{propstage}, below).
The two properties make this difference a unimodal function over time, which initially increases (exactly as the size of the follicle) but,
as soon as competition is over and follicle $i$ has either turned dominant ($x_i\to\xi$, $s_i\to\xi$) or atretic ($x_i\to0$, $s_i\to0$),
starts to decrease and converges to zero 
(via $\xi-\xi$ or $0-0$).
At least phenomenologically, this is how secretion of estradiol can be described 
(i.e., via $\sigma_i\propto x_i-s_i$) when both size and maturity, the two variables our formalism \eqref{langebas}
is built on, are employed.   
Before committing to a slightly more realistic expression
(as the pure difference is certainly the simplest possible characterization),
we are looking at other hormones involved.

Depending on the presence of LH and FSH in the ovaries, 
follicles produce androgens and estrogens, in layers of theca and granulosa cells, respectively \citep{hillier1981modulation}. 
Induced by LH, which is synthesized in the pituitary and transported to the ovaries through the blood, theca cells convert cholesterol into androgens.
This conversion is ongoing in all developmental stages of a follicle, although more intensely in later stages \citep{Magoffin2002,young2010theca}.
The produced androgens diffuse into the granulosa cells of the same follicle, and likely to other follicles as well \citep{mcnatty1976concentration}.
In the granulosa cells, the androgens are converted into estrogens such as E2. This only happens, however, with the help of the enzyme aromatase, which is expressed in the granulosa cells through FSH, which again is synthesized in the pituitary in accordance with the regulatory hormonal feedback \citep{dorrington1975estradiol,simpson2002aromatase}.
Without being converted into estrogens, androgens are detrimental to the granulosa cells; certain accumulated amounts induce apoptosis.
That is, without FSH being sufficiently available, the androgen concentration increases and the follicle undergoes atresia \citep{billig1993estrogens, kaipia1997regulation}.
The amount of FSH needed to keep a follicle developing is inversely related to its size and maturity. More developed follicles are more sensitive to FSH and thus function properly also at smaller FSH concentrations \citep{brown1978pituitary}.
Many authors, the SISYPHE team included, argue in terms of the limited resource FSH when explaining competition for dominance between follicles.
We think it is more natural to assign this function to the androgens, because their accumulation directly suppresses follicular survival and, importantly, the accumulated amount informs about the number and maturity of the surrounding follicles, as all of them produce androgens.%
\footnote{
In contrast, FSH is only produced in the pituitary and not by the follicles, which would allow for competition only with a large time delay (on the scale of days).}

Moreover, as the production of androgens and consequently their presence is directly related to the follicle size, one can express follicular competition
in terms of a single equation \eqref{langeeq}, only modeling size. 
The contribution therein representing the involved suppression of growth is given by the sum over all sizes,
where the size of the focal follicle (with index $i$) is included with a smaller weight ($\eta<1$).
This sum can be understood as an effective description of the hormones involved,
presumably representing harming androgens
as discussed by \cite{Lange2018-xf}.
The function of androgens in this context has in fact been attributed to transmitting information about the size and, as we will add here, the maturity of the surrounding follicles.
In particular, the change of the androgen concentration due to follicle $i$ has been associated by Lange et al.~with a positive change in its size $(\dot x_i^\nu)_+$.
When including maturity as a limiting factor for growth, one can express the androgens (produced by the theca cells of follicle $i$) through a version of the continuous difference operation,
\begin{equation}\label{eqAi}
 A_i=\frac{1}{\xi^\nu}\,(x_i^\nu-s_i^\nu),  
\end{equation}
which we have suggested for the secretion of estradiol (at the beginning of this section).
Here we raise both variables, size $x_i$ and maturity $s_i$, to some power $\nu$. This $\nu$ is the same power as used in the equation \eqref{langeeq} for the competing sizes,
$\sum_ix_i^\nu$.
We expect values for $\nu$ to lie between 1 and 3, to represent a volume of fractal dimension where androgen production happens.\footnote{Data of two cows led to values of  2.475 and 1.65 for $\nu$, cf.~Fig.~\ref{fig:bovine419} and Fig.~\ref{fig:bovine3674}, resp.}
We normalize these follicular androgen contributions, $A_i\le1$, via the parameter $\xi$, encoding the maximal values of size and maturity.

The rate at which estradiol is secreted by follicle $i$,
resulting from direct conversion of androgens
in its granulosa cells, can be modeled by a multiple of \eqref{eqAi},
\begin{align}
\sigma_i=c_{Foll}^{E2}\,\varepsilon^\nu A_i,
\end{align}
where the first factor $c_{Foll}^{E2}$ is taken from the literature (containing all relevant units) and the second (hence dimensionless) factor $\varepsilon$ is adjusted in accordance with data.
Again, its form is plausible as larger follicles convert more androgens than smaller ones, but when they become mature they either ovulate or degenerate and stop converting androgens.

The secretion rate parameter $c_{Foll}^{E2}$ is taken from \cite{STOTZEL20121415}. In fact, we aim at employing their bovine estrous cycle, in particular their estradiol equation,
\begin{equation}\label{EqE2BovCycle}
\frac{d}{dt}E2(t)=c_\text{\sl Foll}^\text{\sl E2}~\text{\sl Foll\,}^2(t)-c_\text{\sl E2}~E2(t),
\qquad\begin{cases}
c_\text{\sl Foll}^\text{\sl E2}\hspace*{-1em}&=~2.19,\\
c_\text{\sl E2}&=~1.23,
\end{cases}
\end{equation}
to compare our model against E2 data in cows.
Despite hormone concentrations and other state variables being normalized (in artificial units with maximal values around one), 
their time scales are realistic, with period lengths around 23 days and multiple waves.
The explicit form (of the square) of their state variable, 
\begin{equation}\label{EqEffAndro}
\text{\sl Foll\,}^2=
\left(\frac{\varepsilon}{\xi}\right)^\nu
\sum_{i\le N}\,(x_i^\nu-s_i^\nu),
\end{equation}
supposed to model all growing follicles effectively in one variable, can now be interpreted as the total amount of androgens convertible to estradiol at a given time.

Three parameters have been adjusted to fit E2 data, which---in addition to the follicle sizes---we have taken from \cite{Cummins2012-oa}; see Figure~\ref{fig:bovine419}.
These parameters are the maturation rate $\alpha$,
the dimensionless E2 secretion rate multiplier $\varepsilon$, 
and the initial estradiol concentration $E2(0)$.
For the first few days into a follicular wave, all possible $E2(t)$ trajectories turn out to be virtually the same curve, independent of the initial values $s_i(0)$ of the underlying maturity. During those first days, the estradiol concentration increases. Only when the concentration starts to decrease again, the E2 trajectories resulting from different initial maturities separate and follow different paths.

As an immediate application, this behavior could be employed to numerically determine the initial maturities, $s_i(0)$, starting the wave. 
Let us recall that even if in Item \ref{item3},  Section \ref{results}, we associate the follicular maturity $s_i(t)$ with morphological and developmental stages,\footnote{similar to the discrete compartmental stages in, e.g., \cite{MARGOLSKEE201321} or \cite{HENDRIX201431}}
there is no proper data or scale available, at least up to now.  
This is also the reason why there are no data points on the right hand side panels of Figure~\ref{fig:bovine419} and only a bunch of (possible) curves. That is, in contrast to the directly observable follicle size, maturity is kind of a hidden variable.
However, if E2 data was available over the whole wave,
one could narrow down the initial values of the maturities and hence reconstruct their full temporal course via the modeling ODEs \eqref{langebas}.

Alternatively, one can compare the multiple  $E2(t)$ curves in Figure~\ref{fig:bovine419} with simulated E2 data. 
\cite{cells11233908} have performed such a simulation, based on follicle data while relying on a classification of discrete follicular stages with different E2-contributions. Importantly, they use data from the same cow, no. 419 in \cite{Cummins2012-oa}.
The model by \cite{STOTZEL20121415} also produces $E2(t)$ curves, but their curves are idealized and do not reflect data of a particular cow. A similar description applies to a simulation by \cite{BoerEtAl2011}, which has been used for comparison by \cite{cells11233908}.
Figure~\ref{fig:McEvoy} shows the simulated data by McEvoy et al., Boer et al.,~and Stötzel et al., as well as an $E2(t)$ curve following our construction with hand-picked initial data, $s(0)$.
All four curves are comparable in shape, although the simulation by Boer et al. and Stötzel et al.~deviate from the other two by up to two days.

\begin{figure}[t]
\centering
\hspace*{-5em}
\begin{minipage}[c]{20em}
\includegraphics[width= 1.05\textwidth]{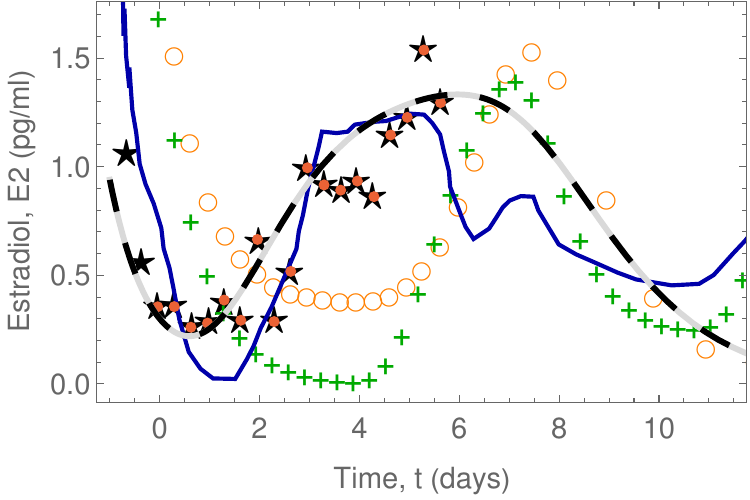}       
\end{minipage}
\qquad~
\begin{minipage}[c]{10em}
\includegraphics[width=1.5\textwidth]{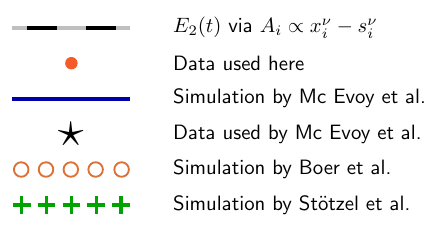}
\end{minipage}\\[-2ex]
\hspace*{-2em}
\begin{minipage}[b]{20em}
%\vspace*{1ex}
\footnotesize
\begin{align*}
\intertext{\sf Fitting:}
RMSD&=0.183~({\rm mm})\\
R&=0.913\\
\intertext{\sf Parameters:}
E2(0)&=0.3~({\rm pg/ml})\\
s(0)&=(0.3,0.2,0.1)~({\rm mm})\\
\alpha&=0.05~({\rm 1/mm/days})\\
\varepsilon&=1\\
\end{align*}
\vspace*{-.5ex}
\end{minipage}
\vspace*{-1em}
\begin{minipage}[b]{10em}
\hspace*{-2em}
 \includegraphics[width= 1.6\textwidth]{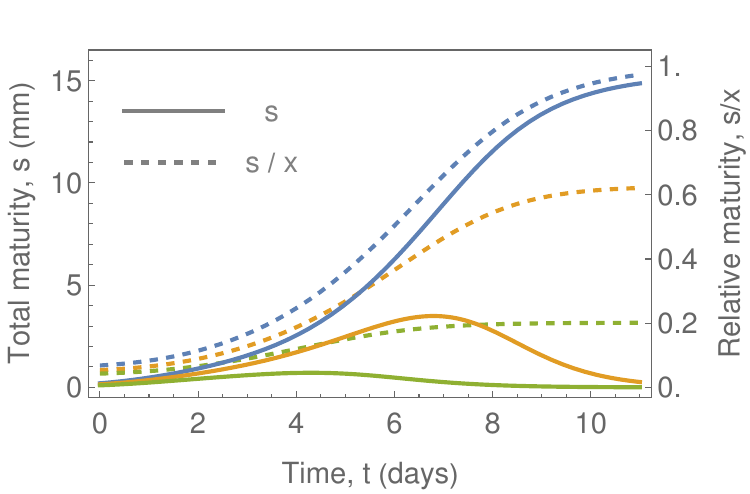}\\[1ex]
\end{minipage}
\caption{Comparison of various simulated blood estradiol concentrations over time. The underlying data (black stars, as used by \cite{cells11233908} and red dots, two data points less, as used here) are associated with Cow 419 from \cite{Cummins2012-oa}. 
Two simulated estradiol concentrations (blue line and orange circles in the upper left panel) are taken from Fig.~13 of McEvoy et al. and transformed onto a realistic scale (pg/ml). Parameters for our $E2(t)$ estimation (via $A_i\propto x_i^\nu-s_i^\nu=x_i^\nu(1-\bar s_i^\nu)$, black dashed line) are chosen to fit the data and the simulated data best. Follicle sizes $x_i(t)$ and associated parameters are chosen as in Fig.~\ref{fig:bovine419}. Model-based simulations (orange circles and green pluses), generated for idealized cows, are taken from \cite{BoerEtAl2011} and \cite{STOTZEL20121415}, resp. Maturity levels, total $s_i$ and relative $\bar s_i\equiv s_i/x_i$, underlying our computations of $E2$ are shown as well over the duration of the simulated first wave of Cow 419 
(lower right panel).}
\label{fig:McEvoy}
\end{figure} 

\section{Discussion }\label{sect:discussion}

Inspired by earlier work on physiologically structured population models (PSPMs) and on stage-structured models,
we have introduced a PDE framework for describing follicular development by explicitly incorporating the concept of maturity on both the level of cells and follicles 
(encoded by the variables $m$ and $s$, respectively).
We have proposed a set of five assumptions to connect the histomorphological information encoded by the PDE to macroscopic observables described by ODE models through so-called moment-closure relations (e.g., identifying $s$ to be the first moment with respect to $m$).

For now, our approach yields a closed first-order moment system without explicitly tracking higher-order moments.
In that we assume that the first two moments, representing size (or cell count) and maturity, provide a sufficiently accurate description of the biological system that can be 
adapted to account for the different 
mechanisms underlying follicular development.%
\footnote{A similar assumption underlies the SISYPHE model, in which follicular mass and maturity are the primary variables used to characterize a follicle as ovulatory or non-ovulatory.
Note that the maturation function $g_i$ in \eqref{aymardcharfull} depends only on the local maturity $m$ and the follicular maturity $s_i$; also see the discussion in Appendix \ref{num}.}

The first-order closure can be interpreted as a mean-field approximation, in which heterogeneity in cellular maturity plays a secondary role. Mathematically, this amounts to assuming that the leading contributions to the maturation, death, and proliferation functions ($g_i,\lambda_i,p_i$, resp.) are captured by a finite polynomial expansion in their
maturity $(m)$, see e.g., \eqref{a2_b} or \eqref{aymardmatquadratic} and \eqref{gammas}.

However, the resulting ODE framework  
captures the leading-order dynamics of follicular cohorts while preserving  
information on their underlying histological structure. Despite its simplicity, the framework is able to reproduce key qualitative features observed in both theoretical and experimental studies.
As a first successful example, we can name the implementation of the estradiol production, which we have demonstrated for bovine data.

The proposed closure may not be regarded as an intrinsic limitation of the framework
but as a natural lowest-order approximation. Whenever additional information on maturity heterogeneity becomes available, the same PDE-to-ODE framework can be systematically refined by incorporating higher-order moments and suitable closure relations.
That is, the framework provides a route toward a more detailed description of cellular maturity connecting microscopic population structure and macroscopic follicular dynamics.

\paragraph{Structured populations replacing compartmental models}
The structured population model 
\eqref{pdesystem}
gives rise to an ODE system
that couples follicle sizes with developmental stages in an explicit form \eqref{odesysmain}. This construction removes the need for discrete ovarian stages and allows to track the mass and stage (i.e., size and maturity) of individual follicles continuously in time. Even if the reconstruction of the underlying 
histomorphological dynamics 
is not uniquely determined, the structure of \eqref{odesysmain} can be used to constrain
admissible maturation mechanisms (through maturation functions $g_i$) and competitive interactions (through death and proliferation functions, $\lambda_i$ and $p_i$).

The application of this framework has clarified
the relationship between early PDE approaches involving cellular maturity and discrete compartmental models. In the original 
approach by the SISYPHE team \citep{Monniaux2016-yt},  
the maturity variable and the total maturation are introduced as auxiliary quantities linked to FSH levels, without providing an independent physiological or morphological interpretation. 
Their definition of atretic and ovulatory follicles requires the introduction of
of a maturity threshold to distinguish between the two fates. In Section~\ref{aymardconn}, we argue that, in the calibrated SISYPHE models  \citep{aymard16,dimred}, most of the follicles can reach a non-normalized, follicle-dependent 
maximal  average maturity $\bar s_i$  depending on
two parameters. This feature makes the distinct roles of local $m$ versus follicular maturity $s_i$, and consequently average maturity $\bar s_i$ difficult to disentangle in these models.%
\footnote{Note that we do not claim that follicular fate is generally predetermined in the SISYPHE models. These models allow for rescue dynamics, in which the fate of a follicle depends on its initial mass relative to that of the other follicles, with a sufficiently large initial mass favoring rescue.}

In our framework, the local maturity variable $m$, the follicular size $x_i$, 
and the follicular maturity $s_i$
are normalized and can have a direct interpretation in terms of physiological, morphological, and structural properties  
\citep{Gougeon1986-xc,Faddy1995-yh,Clement1997-op,Da_Silva-Buttkus2008-kz,Richard2024-pz}. 
Furthermore, our explicit model \eqref{langepde} provides a natural description of atretic follicles, for which the %average
relative maturity $\bar s_i\equiv s_i/x_i$ remains approximately constant at 
large times, while both
size and maturity
vanish. Biologically, this suggests that apoptotic cells may preserve their structural properties until they are removed from the system, rather than undergoing a regression in maturation.

\paragraph{Maturity controlling follicular estradiol production}

When only interested in the time evolution of a
single parameter related to follicular size or secretion of estradiol, it may suffice to formulate an ODE model for just that 
parameter.
Such an approach has been  
carried out by many authors, including \cite{Lacker1981-js,Mariana1994,SOBOLEVA200045,Lange2018-xf,shilo}, all of which successfully model 
the time evolution of at least one
follicular wave.
To extend these models and incorporate microscopic features of the cell population, 
one can utilize our PSPM \eqref{PDEgenODE} as 
being constructed to ensure the
follicular development
to be consistent with the underlying
ODE dynamics \eqref{genODE}.
This particular connection between PDE and ODE descriptions underscores the 
potential for further applications.

The maturity variable $s_i$ allows us to overcome the ``no-crossing'' property inherent in the older ODE models.
In PSPMs with stage-independent growth and death rates, e.g.~\eqref{simple_lange}, crossings of trajectories $x_i(t)$ cannot occur. Consistently with other PDE models, including those proposed by the SISYPHE team, the stage dependence, as introduced in \eqref{nicematfunction}, removes this restriction, as we will explore in an upcoming work \citep{next}.

Additionally, the introduction of the maturity variable allows for implementing
hormonal production,
which has been utilized by SISYPHE team to couple cellular and FSH dynamics via an algebraic relation \eqref{stangehill}.
Through our reduction, the implementation of hormonal feedback 
via the maturity variable
extends to the corresponding ODE models.
In Section \ref{sect:BovCycle}, we have demonstrated that
maturity-extended ODEs can
model estradiol production,
thereby extending previous approaches 
that exclusively couple follicular sizes with hormone levels \citep{FISCHERHOLZHAUSEN2022111150}.
In particular, we proposed the secretion of estradiol to be proportional to roughly the difference between follicular size and maturity,
and we tested this ansatz against available estradiol data in cows \citep{Cummins2012-oa},
reproducing estradiol concentrations via a bovine model \citep{STOTZEL20121415}.
Our approach turned out to be consistent with the data while appearing much simpler than the estradiol production model proposed in \cite{cells11233908}.
However, the limited temporal coverage, with data points spanning  only half a follicular wave, cannot confirm the functional form with full confidence.

Our reduced ODE model suggests that maturity $s_i(t)$ is delayed with respect to size $x_i(t)$; see Figure~\ref{fig:McEvoy}. 
Therefore, our estradiol model correctly predicts that ovulatory follicles
(for which $s_i/x_i \approx 1$)
produce relatively little estradiol, thereby phenomenologically reproducing the estradiol drops observed in literature.

\paragraph{Coupling follicles with the full cycle}

A natural continuation of this work is the inclusion of the follicular dynamics, formulated by PSPMs, in a fully functional endocrine model of the hormone cycle. Although we have given a hormonal interpretation of maturity when identifying it with estradiol production, the proposed follicular dynamics remains only partially coupled to the endocrine system. The follicle development
is not yet explicitly linked to the key regulatory hormones such as FSH, LH, progesterone, and estradiol. Nevertheless, the ODE reduction of the proposed PSPM provides a robust basis for coupling follicular development at the cellular level with hormonal feedback, as demonstrated in \cite{fischerb,FISCHERHOLZHAUSEN2022111150} and in the full endocrine-cycle models of \cite{Clark2003-oe,MARGOLSKEE201321,HENDRIX201431}.
Following the approach of \cite{BOER20115987,STOTZEL20121415}, and \cite{FISCHERHOLZHAUSEN2022111150}, one may replace the constant coefficients in \eqref{langepde} by time-dependent functions driven by the hormonal cycle. In 
such an approach, progesterone $P_4(t)$ should be modeled through the dynamics of the corpus luteum, which we suggest to associate 
either with the declining phase of the
size--maturity dynamics, 
which the current model is lacking, 
or with the plateau phase of the average maturity dynamics; see Figure \ref{fig:bovine419} (bottom-right).

After estradiol and progesterone have been introduced as functions of the maturity variable, the remaining hormones, 
such as FSH, LH, and inhibins, can be modeled by standard feedback-delay equations or Hill-type functions, as in the aforementioned endocrine ODE frameworks.
The parameters appearing in the mass equation \eqref{langebas} can then be linked as
in \cite{FISCHERHOLZHAUSEN2022111150}, while the maturity variable serves as a central coupling term between the follicular state and the different hormones. In particular, the additional parameter $\alpha$ in the maturity equation might be encoded as
positively correlated with FSH and LH levels, since elevated concentrations of these hormones increase the rate of follicular development. The implementation of a fully coupled endocrine–ovarian model, together with the investigation of de/-synchronization, multimodal maturity distributions, and PMOS scenarios, will be addressed in future work, thereby extending the present framework toward a more comprehensive description of the ovarian cycle.

\paragraph{Higher order maturation functions and future histomorphological data}

The maturity variable $m$ has been motivated to serve as an index for the histological state of the cell population, encoding structural and functional changes that accompany growth and atresia of a follicle; for different microscopic mechanisms see \cite{Salmon2004-ij,DIAZ200691,GOUGEON2010132,sugiura}.
Within this interpretation, the maturation function in \eqref{langepde} captures de-/synchronization effects as demonstrated via the parameter $\beta$, but it does not resolve finer features of the underlying cell development. Higher-order moments or diffusion terms of the form $\partial_m^2 \phi_i$ would allow one to better represent mechanisms such as differential susceptibility to atresia, cohort desynchronization, and nonlinear hormonal responses. In the present work, we treat these effects as subleading during follicular development, but one could
incorporate them systematically via a perturbative expansion.How to do this expansion in detail remains an open question for now.

The main obstacle in constructing an accurate multiscale model of follicular development is the lack of structured histomorphological data that can be directly linked to macroscopic observables, such as follicular mass and maturity. A more comprehensive and biologically realistic formulation would require quantitative measurements of cell populations throughout all stages of follicular growth and degeneration. Such data would %also
make it possible to identify and quantify the role of higher-order moments in follicular dynamics.
In this sense, the present framework provides 
a mathematically tractable baseline for future refinements when more detailed histological and endocrine data become available.

\section{Methods}\label{math}

Here, we prove the mathematical statements required for the Results section.

\subsection{Mathematical background}

Starting from the transport PDE \eqref{pdesystem}, we derive a system of ODEs governing the associated moments. 
The moments of order zero and one (i.e., mass and stage) are determined by
\begin{theorem}[ODE reduction]\label{odethm}
Let  $T>0$. Assume
\begin{enumerate}%[label=(\text{h}\arabic*), ref=\text{h}\arabic*]
\item[(a0)] \label{h0} 
$\phi_i\in C^1([0,T];L^1(0,1))$ weakly
solves
\begin{equation} \label{pdesystemmath}
\partial_t \phi_i + \partial_m (g_i \phi_i)= (-\lambda_i+p_i)\phi_i,\qquad i=1,2,\dots, N,
\end{equation}
and
\begin{equation}
x_i(0)=\int_0^1 \phi_i(0,m)\ud m\in\Omega,\quad s_i(0)=\int_0^1  m\, \phi_i(0,m) \ud m\in\Omega,
\end{equation}
where $\Omega=(0,\xi)$ for some $\xi>0$;
\item[(a1)] at any time $t\in [0,T]$, 
\begin{equation} 
\mathrm{ess~supp}~\phi_i(t,\cdot)\subset(0,1);
\label{a1math}
\end{equation}
\item[(a2)] \label{h2} growth, decay, and maturation are continuous real functions on $[0,T]\times\Omega^{2N}$ that can be written as follows,
\begin{subequations} \label{a2math}
\begin{align} \label{a2math-a}
p_i&=p_{i,0}(t,x,s),\quad \lambda=\lambda_{i,0}(t,x,s),\quad\text{and}\\
g_i&=g_{i,0}(t,x,s)+m\, g_{i,1}(t,x,s),\quad\text{resp.} \label{a2math-b}
\end{align}
\end{subequations}
\end{enumerate} 
Then the raw moments
\begin{equation}
x_i(t)=\int_0^1 \phi_i(t,m) \ud m, \quad s_i(t)=\int_0^1 m\, \phi_i(t,m) \ud m \label{defnmoments}
\end{equation}
satisfy the ODE system
\begin{subequations} \label{odesys}
\begin{align}
\dot x_i&= x_i\left[-\lambda_{i,0}(t,x,s)+p_{i,0}(t,x,s)\right] \label{finalmasseqn}\\
\dot s_i&= s_i\left[-\lambda_{i,0}(t,x,s)+p_{i,0}(t,x,s)\right]+g_{i,0}(t,x,s) x_i+g_{i,1}(t,x,s) s_i\label{finalstageeqn}
\end{align} 
\end{subequations}
where $x_i(t)$ and $s_i(t)$ are continuous function on $[0,T_{\min})$ with values in $\Omega$, 
for some $T_{\min}$ with $0< T_{\min}\leq T$.
\begin{proof} 
By assuming that a solution $\phi_i(t,m)$ to the PDE \eqref{pdesystemmath} exists, being continuous in $t$ and integrable with respect to $m$ as proposed in \textit{(a0)}, we integrate the three additive terms of the PDE,
\begin{align*}
\partial_{t}\int_0^1\phi_i(t,m)\ud m
+ \int_0^1\partial_{m}\bigl[g_i(m,x(t),s(t))\,\phi_i(t,m)\bigr]\ud m
&= \int_0^1(-\lambda_i+p_i)\,\phi_i(t,m)\ud m,
\end{align*}    
apply the definition \eqref{defnmoments} of the raw moment $x_i$, perform an integration by parts, use the functional independence of $\lambda_i$ and $p_i$ from $m$ as proposed in \eqref{a2math-a}, and end up with an ODE in the variable $t$,
\begin{align*}
\frac{d}{dt}\,x_i(t)
+ g_i(m,x(t),s(t))\,\phi_i(t,m)\Big|_{m=0}^{m=1}
&= (-\lambda_i+p_i)\,x_i(t).
\end{align*}
By employing the support property  \eqref{a1math} to the boundary term we finally obtain an ODE for the zeroth moment,
\begin{align*}
\dot x_i + 0 
&=x_i[-\lambda_i+p_i],
\end{align*}
which is in fact \eqref{finalmasseqn}.
To derive the ODE \eqref{finalstageeqn}
for the first moment, we apply the same steps as for the zeroth moment, now applied to the PDE \eqref{pdesystemmath} multiplied by $m$, i.e.,
\begin{align*}
\partial_{t}\int_0^1m\,\phi_i\ud m
- \int_0^1\!\bigl[\,g_{i,0} + m\, g_{i,1}\,\bigr]\phi_i\ud m
&= (-\lambda_i+p_i)\int_0^1 m\,\phi_i\ud m,\\
\intertext{while employing $g_i$ represented by a linear function of $m$, as proposed in \eqref{a2math-b}, which leads to}
\dot s_i
- \bigl[ x_i g_{i,0} + s_i\,g_{i,1}\,\bigr]
&= s_i[-\lambda_i+p_i].%\label{stage11}
\end{align*}
Since the right-hand sides of \eqref{odesys} are continuous on $[0,T]\times\Omega^2$, it follows that the solutions $(x_i(t),s_i(t))$ are continuous on an interval $[0, T_{\min})$, where $T_{\min}>0$ is defined as
\begin{align*}%\label{eq_Tmin}
T_{\min} = \sup \left\{ t \in (0,T) \colon (x_i(t),s_i(t)) \in \Omega^2, \forall i \right\}.
\end{align*}
\end{proof}
\end{theorem}

\begin{remark}
The assumptions (a1) and (a2) in the theorem coincide with those made in Section \ref{pspmtoode}. 
\end{remark}

\begin{remark}
In a classical setting with $\phi_i\in C^k([0,T]\times[0,1])$ and $k\geq 1$, the assumption in \eqref{a1math} can be replaced by  Dirichlet boundary condition $\phi_i(t,0)=\phi_i(t,1)=0$. 
\end{remark}

\begin{remark}
In the theorem, we assume the existence of a weak solution. Establishing local existence and uniqueness for the PSPM solution requires conditions for the specific maturation, growth, and decay functions. 
We do not intent to specify those,
therefore we cannot pursue well-posedness in full generality.
Instead, we outline one possible strategy for obtaining existence and uniqueness if Dirichlet boundary conditions were imposed. 
Starting from the problem
\begin{subequations}
\begin{empheq}[left=\empheqlbrace]{align}
&\hspace{.2cm}\partial_t \phi_i
+ \partial_m \left[
\bigl(g_{i,0}(x,s) + m g_{i,1}(x,s)\bigr)\phi_i
\right]
= \bigl(-\lambda_i(x,s) + p_i(x,s)\bigr)\phi_i\label{pdefull}
\\[0.5em]
&\hspace{.2cm}\dot{x}_i = x_i\bigl(-\lambda_i(x,s) + p_i(x,s)\bigr) \label{ode-1}
\\[0.5em]
&\hspace{.2cm}\dot{s}_i
= s_i\bigl(-\lambda_i(x,s) + p_i(x,s)\bigr)
+ x_i g_{i,0}(x,s)
+ s_i g_{i,1}(x,s) \label{ode-2}
\\[0.5em]
&\hspace{.2cm}s_i(0)
= \displaystyle\int_0^1 m\,\phi_i(0,m)\,\ud m
\leq x_i(0)
= \displaystyle\int_0^1 \phi_i(0,m)\,\ud m
< \xi \label{init_cond}
\\[0.5em]
& \hspace{.2cm} \phi_i(t,0)=\phi_i(t,1)=0,
\end{empheq}
\end{subequations}
and assuming suitable regularity and local Lipschitz conditions on
$g_i$, $p_i$, and $\lambda_i$, one can first establish the existence and
uniqueness of bounded solutions $(x_i(t),s_i(t))$ to the ODE system \eqref{ode-1}-\eqref{ode-2}, at
least on a sufficiently small time interval.
Substituting these ODE solutions into the PDE \eqref{pdefull} yields
\begin{equation}
\partial_t\phi_i
+ A_i(t,m)\,\partial_m\phi_i
= f_i(t)\,\phi_i,
\end{equation}
where
\begin{equation}
A_i(t,m)
= g_{i,0}\Big(x(t),s(t)\Big)
+ m g_{i,1}\Big(x(t),s(t)\Big)
\end{equation}
and
\begin{equation}
f_i(t)
= -\lambda_i\Big(x(t),s(t)\Big)
+ p_i\Big(x(t),s(t)\Big)
- g_{i,1}\Big(x(t),s(t)\Big).
\end{equation}
Thus, after solving the ODE system, one obtains a family of linear,
time-dependent transport equations for the functions $\phi_i$. Local
existence and uniqueness for $\phi_i$ can then be established, for
example, by the method of characteristics in the case of classical
solutions or by fixed-point arguments in the case of weak solutions. Having obtained the existence and uniqueness results for the ODE system \eqref{ode-1} and \eqref{ode-2} and the PDE \eqref{pdefull}, we have to identify the first two moments of $\phi_i$ with the variables $x_i$ and $s_i$. Integrating the PDE as shown in Theorem \ref{odethm} yields the evolution equations for the zeroth and first raw moments, respectively. 
These equations are
\begin{subequations}\label{fullproblem}
\begin{empheq}[left=\empheqlbrace]{align}
&\frac{\ud}{\ud t} \int_0^1 \phi_i(t,m)\,\ud m = \bigl[-\lambda_i(x,s) + p_i(x,s)\bigr] \int_0^1 \phi_i(t,m)\,\ud m, \label{moment_0} 
\\[.5em]
&\frac{\ud}{\ud t} \int_0^1 m\phi_i(t,m)\,\ud m
= \bigl[-\lambda_i(x,s) + p_i(x,s)\bigr] \int_0^1 m\phi_i(t,m)\,\ud m
\nonumber  \\&\hspace{2cm}+  g_{i,0}(x,s) \int_0^1 \phi_i(t,m)\,\ud m
+  g_{i,1}(x,s) \int_0^1 m\phi_i(t,m)\,\ud m, \label{moment_1}
\end{empheq}
\end{subequations}
which are comparable to \eqref{ode-1}-\eqref{ode-2}, up to the identification between the raw moments and $x_i$, $s_i$. To finally achieve the relations between raw moments and $x_i,s_i$, we first subtract \eqref{moment_0} from \eqref{ode-1} and using initial conditions \eqref{init_cond}, by existence and uniqueness, it follows
\begin{equation}
x_i(t)=\int_0^1 \phi_i(t,m)\ud m,\ \ \text{for } t\geq0.
\end{equation}
Similarly, subtracting \eqref{moment_1} from \eqref{ode-2}, using initial conditions \eqref{init_cond}, by existence and uniqueness is follows
\begin{equation}
s_i(t)=\int_0^1 m\phi_i(t,m)\ud m,\ \ \text{for } t\geq0.
\end{equation}
\end{remark}

For $x_i(t)>0$, we define the average stage (or {\em relative maturity}) by
\begin{equation*}
\bar s_i(t):=\frac{s_i(t)}{x_i(t)}.
\end{equation*}
It follows
\begin{corollary}\label{evolutionsbar}
Let \((x_i(t),s_i(t))\) for $i \in \{1, 2, \dots, N\}$ satisfy \eqref{odesys} and $x_i(t)>0$,
then the average‑stage
$\bar s_i(t)$
obeys
\begin{equation}\label{xsrelation}
\dot{\bar s}_i
\;=\;
g_{i,0}\!\bigl(t,x,x\circ\bar s\bigr)
\;+\;
\bar s_i\,g_{i,1}\!\bigl(t,x,x\circ\bar s\bigr).
\end{equation}
\begin{proof}
Since the pair \((x_i(t),s_i(t))\) satisfies the ODE system \eqref{odesys},
we have \(x_i(t)>0\) on the interval of interest, so \(\bar s_i\) is well defined.
A direct differentiation gives
\begin{align}
\frac{\mathrm d}{\mathrm d t}\!\bigl(\bar s_i\bigr)
&=\frac{\mathrm d}{\mathrm d t}\!\left(\frac{s_i}{x_i}\right)
=\frac{1}{x_i^{2}}\bigl(\dot s_i\,x_i-s_i\,\dot x_i\bigr) \nonumber
\label{eq:deriv1}\\[4pt]
&=\frac{1}{x_i^{2}}\Bigl\{
\bigl[\,s_i\bigl(-\lambda_{i,0}+p_{i,0}\bigr)
+g_{i,0}\,x_i+g_{i,1}\,s_i\bigr]\,x_i
-s_i\,x_i\bigl(-\lambda_{i,0}+p_{i,0}\bigr)
\Bigr\} \nonumber\\[4pt]
&= g_{i,0}\bigl(t,x,x\circ\bar s\bigr)
+\bar s_i\,g_{i,1}\bigl(t,x,x\circ\bar s\bigr).
\end{align}
In the second line above we substituted \(\dot x_i\) and \(\dot s_i\) from
\eqref{finalmasseqn}--\eqref{finalstageeqn} and used the identity
\(s_i=x_i\bar s_i\).
Equation \eqref{eq:deriv1} coincides with \eqref{xsrelation},
which proves the claimed evolution equation for the average‑stage.
\end{proof}
\end{corollary}

It is natural to assume that the average‑stage dynamics of each population are independent of those of the other populations. Under this hypothesis the system \eqref{odesys} can be simplified, as stated in the following corollary.

\begin{corollary}\label{redsbar}
Assume that, for every $i\in\{1,\dots,N\}$, the maturation functions depend only on the
own population’s mass and stage, i.e.
\begin{equation}\label{hyp:local-g}
g_{i,0}(t,x,s)=g_{i,0}\!\bigl(t,x_i,s_i\bigr),\qquad
g_{i,1}(t,x,s)=g_{i,1}\!\bigl(t,x_i,s_i\bigr)
\quad\text{for all }t\in[0,T].
\end{equation}
Then the ODE system \eqref{odesys} is equivalent to the decoupled system
\begin{subequations}\label{uncoupled}
\begin{align}
\dot x_i &= x_i\Bigl[-\lambda_{i,0}\!\bigl(t,x,x\circ\bar s\bigr)
+p_{i,0}\!\bigl(t,x,x\circ\bar s\bigr)\Bigr],
\label{uncoupled-x}\\[4pt]
\dot{\bar s}_i &= 
g_{i,0}\!\bigl(t,x_i,x_i\bar s_i\bigr)
\;+\;
\bar s_i\,g_{i,1}\!\bigl(t,x_i,x_i\bar s_i\bigr),
\label{uncoupled-g}
\end{align}
\end{subequations}
where $\displaystyle\bar s_i(t) =\frac{s_i(t)}{x_i(t)}$ with $x_i(t)>0$.
\end{corollary}

These results can be seen as a perturbative analysis of PSPMs,
bridging the gap between multiscale models formulated as PDE systems and macroscopic models described by ODE systems. They will allow us to better understand the competition between populations, equilibrium points, bifurcations, and, eventually, the calibration of microscopic dynamics from macroscopic outputs.

Let us define the macroscopic variables
\begin{equation*}
s_i^{(\alpha)}= \int_0^1 m^\alpha \phi_i \, dm,
\end{equation*}
representing the $\alpha$-th moment.
Then, we have the following 
\begin{theorem} \label{higherorder} 
Assume that hypotheses (a0) and (a1) of Theorem \ref{odethm} hold, and let 
\begin{equation*}
s^{(\alpha)}(t) = \left(s_1^{(\alpha)}(t), \dots, s_N^{(\alpha)}(t)\right)
\end{equation*}
be a sequence of macroscopic variables indexed by $\alpha \in \mathbb{N}$, with initial conditions satisfying $s_i^{(\alpha)}(0) \in \Omega$ for all $i \in \{1, \dots, N\}$.    
Furthermore, assume that the growth $p$, decay $\lambda$, and maturation functions $g$ are real-valued, continuous on $[0,T] \times [0,1] \times \Omega^{r_M+1}$, depend on a finite number of macroscopic variables, and admit the power series expansions
\begin{subequations}\label{eqs:thm4}
\begin{align}
g_i\left(t, m, \{s^{(r)}\}_{r\in R}\right) &= \sum_{j=0}^{\infty} m^j g_{i,j}\left(t, \{s^{(r)}\}_{r\in R}\right), \\
-\lambda_i\left(t, m, \{s^{(r)}\}_{r\in R}\right) + p_i\left(t, m, \{s^{(r)}\}_{r\in R}\right) &= \sum_{j=0}^{\infty} m^j \Gamma_{i,j}\left(t, \{s^{(r)}\}_{r\in R}\right),
\end{align}
\end{subequations}
where $R = \{0, 1, \dots, r_M\}$ for a finite integer $r_M \geq 0$. 
Then, on the time interval $[0, T_m] \subseteq [0, T]$, the macroscopic variables $s_i^{(\alpha)}$ satisfy the infinite-dimensional dynamical system
\begin{equation} \label{infinitesystem}
\dot{s}_i^{(\alpha)} = \sum_{n=0}^{\infty} \left[ \Gamma_{i,n} s_i^{(n+\alpha)} + \alpha g_{i,n} s_i^{(n-1+\alpha)} \right].
\end{equation}
\begin{proof}
The proof follows a similar argument to that of Theorem \ref{odethm}. Multiplying \eqref{pdesystemmath} by $m^\alpha$, integrating over the domain, and applying the support properties \eqref{a1math}, yields the claim.
\end{proof}
\end{theorem}

Under generic assumptions for a perturbative expansion (e.g., small higher-order moments or small values of $g_{i,j}$ and $\Gamma_{i,j}$ as $j$ becomes large), the infinite-dimensional system \eqref{infinitesystem} can be truncated into a finite one. See also \cite{Villa2025} for alternative closure and truncation strategies.

\subsection{Proofs related to our rephrased SISYPHE model}\label{math_sisyphe}

Here, we derive various results regarding the truncated version of the SISYPHE model presented in Section \ref{aymardconn}. We start with the following the theorem.
\begin{theorem} \label{thm:aymard}
Let $0 < c_{i,1} < 1$ for all $i \in \{1, \dots, N\}$. For any initial condition 
$\bar{s}(0) = \big(\bar{s}_1(0), \dots, \bar{s}_N(0)\big) \in (0,1)^N$, 
the system of differential equations
\begin{equation} \label{avmaturityaymard}
\dot{\bar{s}}_i = -\bar{s}_i^2 + (c_{i,1}\bar{s}_i - c_{i,1} + 1) \frac{1}{N} \sum_{j=1}^N \bar{s}_j
\end{equation}
admits a unique, continuous solution $\bar{s}(t)$ defined for all $t \geq 0$. 
Moreover, the unit hypercube $[0,1]^N$ is positively invariant; i.e., 
$\bar{s}_i(t) \in [0,1]$ for all $i$ and all $t \geq 0$.   
Furthermore, %within the state space $[0,1]^N$:
\begin{enumerate}
\item The origin $\bar{s}^* = (0, \dots, 0)$ is an unstable equilibrium point.
\item The point $\bar{s}^* = (1, \dots, 1)$ is a locally asymptotically stable equilibrium in $(0,1]^N$.
\end{enumerate}
\begin{proof}
We first prove boundedness in the closed hypercube $[0,1]^N$.
Assume that there exists a first 
time $t_0 > 0$ such that $\bar{s}_{i_0}(t_0) = 0$ for indices $i_0 \in I_0 \subset \{1, \dots, N\}$, while $\bar{s}_i(t_0) \in (0,1)$ 
for all $i \in \{1, \dots, N\} \setminus I_0$. Evaluating \eqref{avmaturityaymard} at $t = t_0$, we obtain
\begin{equation*}
\dot{\bar{s}}_{i_0}(t_0) = (1 - c_{i_0,1}) \frac{1}{N} \sum_{j=1}^N \bar{s}_j(t_0) > 0,
\end{equation*}
since $0 < c_{i_0,1} < 1$ and $\bar{s}_j(t_0) \geq 0$ for all $j$, with at least one strictly positive component associated with $j \in \{1, \dots, N\} \setminus I_0$. Therefore, $\bar{s}_{i_0}$ is strictly increasing at $t=t_0$, and hence $\bar{s}_i(t)$ cannot cross the lower boundary of the interval $[0,1]$.
\newline
Similarly, assume there exists a time $t_1 > 0$ such that $\bar{s}_{i_1}(t_1) = 1$ for indices $i_1 \in I_1 \subset \{1, \dots, N\}$, while $\bar{s}_i(t_1) \in (0,1)$ for all $i \in \{1, \dots, N\} \setminus I_1$. Evaluating \eqref{avmaturityaymard} at $t = t_1$ yields
\begin{equation*}
\dot{\bar{s}}_{i_1}(t_1) = -1 + \frac{1}{N} \sum_{j=1}^N \bar{s}_j(t_1) < 0,
\end{equation*}
since $\bar{s}_j(t_1) \leq 1$ for all $j$.
Thus, $\bar{s}_{i_1}$ is strictly decreasing at $t = t_1$, preventing trajectories from crossing through the upper boundary. This establishes that if $\bar{s}(0) \in [0,1]^N$, then $\bar{s}(t) \in (0,1)^N$ for all $t \geq 0$.\newline
Local existence and uniqueness of a solution around $t = 0$ follow directly from the Picard--Lindelöf theorem, as the right-hand side of \eqref{avmaturityaymard} is locally Lipschitz continuous. Global-in-time existence follows from the boundedness of the vector field on the invariant set $[0,1]^N$, where we have
\begin{equation*}
|\dot{\bar{s}}_i(t)| \leq M_i := \max_{y \in [0,1]^N} \left| -y_i^2 + (c_{i,1}y_i - c_{i,1} + 1) \frac{1}{N} \sum_{j=1}^N y_j \right| < +\infty.
\end{equation*}
The continuity of the solution $\bar{s}(t)$ is thereby guaranteed for all finite times $t \geq 0$.\newline
To analyze the stability of the equilibrium points, we first establish that the origin $\bar{s}^* = (0, \dots, 0)$ is unstable. Indeed, introducing the parametrization $\bar{s}_i(t) = \epsilon \tilde{\sigma}_i(t)$, where $0 < \epsilon \ll 1$, $\tilde{\sigma}_i = O(1)$, and substituting this expression into \eqref{avmaturityaymard}, by neglecting higher-order contributions in $\epsilon$, we obtain  the linearized system
\begin{equation*}
\dot{\tilde{\sigma}}_i = (1 - c_{i,1}) \frac{1}{N} \sum_{j=1}^N \tilde{\sigma}_j > 0,\end{equation*}
which forces trajectories away from the origin.
\newline
To evaluate the local stability of the equilibrium point $\bar{s}^* = (1, \dots, 1)$, we introduce the perturbation variables $\sigma_i(t)$ defined via $\bar{s}_i(t) = 1 - \epsilon \sigma_i(t)$, where $0 < \epsilon \ll 1$ and $\sigma_i(t) = O(1)$. Substituting this expression into \eqref{avmaturityaymard} and neglecting higher-order contributions yields the linearized system
\begin{equation} \label{eq:linearized}
\dot{\sigma}_i = -(2 - c_{i,1})\sigma_i + \frac{1}{N} \sum_{j=1}^N \sigma_j.
\end{equation}
We determine the stability of the trivial equilibrium $\sigma = (0, \dots, 0)$ using the candidate Lyapunov function 
\begin{equation*}
V(\sigma) = \frac{1}{2} \sum_{i=1}^N \sigma_i^2.
\end{equation*}
Taking the time derivative along the trajectories of \eqref{eq:linearized}, we obtain
\begin{equation} \label{lyap}
\begin{split}
\dot{V}(\sigma) &= \frac{1}{N} \left( \sum_{i=1}^N \sigma_i \right)^2 - \sum_{i=1}^N (2 - c_{i,1})\sigma_i^2 \\
&\leq \sum_{i=1}^N \sigma_i^2 - \sum_{i=1}^N (2 - c_{i,1})\sigma_i^2 \\
&= - \sum_{i=1}^N (1 - c_{i,1})\sigma_i^2 \\&\leq 0,
\end{split}
\end{equation}
where we used the property that
\begin{equation*}
\frac{1}{N}\left(\sum_{i=1}^N \sigma_i\right)^2 \leq \frac{1}{N}\left(\sum_{i=1}^N \sigma_i^2\right)\left(\sum_{i=1}^N 1\right) = \sum_{i=1}^N \sigma_i^2.
\end{equation*}
Since $0 < c_{i,1} < 1$, the right-hand side of \eqref{lyap} is strictly negative for all $\sigma_i \neq 0$, and equality holds if and only if $\sigma = (0, \dots, 0)$. By Lyapunov's direct method, the origin of the system is asymptotically stable, implying that $\bar{s}^* = (1, \dots, 1)$ is locally asymptotically stable.
\end{proof}
\end{theorem}

This result leads to the following proposition.

\begin{proposition} \label{prop:contaymard}
Under the same assumptions as in Theorem \ref{thm:aymard}, together with the initial conditions 
$x_i(0) > 0$ for all $i \in \{1, \dots, N\}$, the solution $x_i(t)$ to the mass equation in \eqref{aymardode} remains non-negative and continuous for all finite times $t \geq 0$.
\begin{proof}
Non-negativity follows from the fact that if $x_i(t_0)=0$ then $\dot x_i(t_0)=0$, therefore, the function $x_i$ is not decreasing near $x_i(t_0)=0$. On the other hand, we have that
\begin{equation*}
\left\vert \dot x_i \right\vert = \left\vert x_i\right\vert \left\vert-\lambda_i(\bar s(t)) + p_i(\bar s(t)) \right\vert\leq \left\vert x_i\right\vert C_i
\end{equation*}
where 
\begin{equation*}
C_i=\max_{\bar s\in [0,1]^N} \left\vert-\lambda_i(\bar s) + p_i(\bar s) \right\vert<\infty.
\end{equation*}
The existence of $C_i$ follows from the continuity of $\lambda_i$, $p_i$ and $\bar s_i$. Local continuity follows from the Picard-Lindel\"of theorem, while global existence from the Gr\"onwall's lemma. Indeed, the function $x_i(t)$ is bounded at any finite time by
\begin{equation*}
x_i(t)\leq x_i(0) e^{C_i t}.
\end{equation*}
Continuity at any finite times follows again from the  Picard-Lindel\"of theorem and global existence.
\end{proof}
\end{proposition}

Finally, we are in a position to prove the following result.

\begin{theorem} \label{onestablefixed}
Under the assumptions of Theorem~\ref{thm:aymard}, the unique locally asymptotically stable equilibrium point $\bar{s} \in [0,1]^N$ for the average maturity dynamics in \eqref{aymardode} is given by $\bar{s}_i = 1$ for all $i \in \{1, \dots, N\}$.
\begin{proof} 
We have already shown that the origin $\bar s^*=(0,\dots,0)$ is unstable and that the equilibrium point $\bar s^*=(1,\dots,1)$ is locally asymptotically stable. It remains to show that there are no other equilibrium points in the open hypercube $(0,1)^N$. This is equivalent to proving that the system of equations
\begin{equation}\label{eqcondition1}
\bar s_i^2 - (c_{i,1} \bar s_i -c_{i,1}+1)\frac{1}{N}\sum_{j=1}^N \bar s_j=0, \quad i = 1, \dots, N,
\end{equation}
does not admit any solutions in $(0,1)^N$. To this end, let us define the polynomial
\begin{equation*}
p_i(\bar s_i,\theta):=  \bar s_i^2 - (c_{i,1} \bar s_i -c_{i,1}+1)\theta
\end{equation*}
and observe that \eqref{eqcondition1} is equivalent to the system
\begin{equation*}
\begin{dcases}
\theta = \frac{1}{N}\sum_{j=1}^N \bar s_j, \\
p_i\left(\bar s_i,\theta\right)=0.
\end{dcases}
\end{equation*}
We aim to analyze the zero set of $p_i$ under the restrictions $\bar s_i\in[0,1]$ and $\theta\in(0,1]$, corresponding to the requirement that
\begin{equation*}
\frac{1}{N}\sum_{j=1}^N \bar s_j\in(0,1].
\end{equation*}
Note that we exclude the trivial solution $\bar s=(0,\dots,0)$ since it has already been shown to be unstable.
For a fixed value of $\theta$, equation \eqref{eqcondition1} is quadratic with respect to $\bar s_i$, meaning it admits exactly two roots. Evaluating the polynomial at the boundaries of the domain $[0,1]$ yields
\begin{equation} \label{posneg1}
p_i(0,\theta)=  (c_{i,1}-1)\theta < 0, \quad p_i(1,\theta)=1 - \theta \geq 0.
\end{equation}
By the Intermediate Value Theorem, there exists at least one solution $\bar s_i\in(0,1]$ satisfying $p_i(\bar s_i,\theta)=0$ for any fixed $\theta\in(0,1]$. Furthermore, if $\bar s^{(1)}_i$ and $\bar s^{(2)}_i$ denote the two roots of $p_i(\bar s_i,\theta)=0$, Vieta's formulas imply that
\begin{equation*}
\bar s^{(1)}_i \bar s^{(2)}_i=(c_{i,1}-1) \theta < 0,
\end{equation*}
which indicates that the roots must have opposite signs for $\theta\in(0,1]$. This guarantees that there is exactly one unique solution $\bar s_i\in(0,1]$ satisfying $p_i(\bar s_i,\theta)=0$ given $\theta\in(0,1]$ and $c_{i,1}\in(0,1)$.
Having established the existence of a unique positive solution in $(0,1]$, we can further restrict its range. Indeed, evaluating the polynomial at $\bar s_i = \theta$ gives
\begin{equation}
p_i(\theta,\theta)=\theta(\theta  -1)(1- c_{i,1})\leq0, \label{mcondition1}
\end{equation}
where equality holds if and only if $\theta=1$. This inequality implies that the unique positive solution to $p_i(\bar s_i,\theta)=0$ for a fixed $\theta$ must satisfy $\bar s_i\in [\theta,1]$.
Explicitly solving for this root yields
\begin{equation*}
f_i(\theta)= \frac{1}{2} \left(c_{i,1} \theta+\sqrt{ c_{i,1}^2 \theta^2+4\theta(1-c_{i,1})}\right).
\end{equation*} By construction, the function $f_i(\theta)$ satisfies
\begin{equation*}
p_i\Bigl(f_i(\theta),\theta\Bigr)=0, \quad \text{and} \quad f_i(\theta)\geq \theta \quad \text{for } \theta\in(0,1],
\end{equation*}
where equality holds exclusively at the upper boundary, namely,
\begin{equation} \label{graterm1}
f_i(\theta)=\theta \implies \theta=1.
\end{equation}
To determine the full equilibrium state, we impose the  consistency condition:
\begin{equation*}
\bar s_i = f_i(\theta) \implies \theta = \frac{1}{N}\sum_{j=1}^N \bar s_j = \frac{1}{N}\sum_{j=1}^N f_j(\theta) =: F(\theta).
\end{equation*}
Thus, the problem reduces to finding the fixed points of the mapping
\begin{equation*}
\theta=F(\theta), \quad \theta\in(0,1].
\end{equation*}
By definition, at the upper boundary we find
\begin{equation*}
F(1)=\frac{1}{N}\sum_{j=1}^N f_j(1)=\frac{1}{N}\sum_{j=1}^N 1 = 1.
\end{equation*}
As expected, $\theta=1$ is a fixed point of $F$, corresponding to the equilibrium state $\bar s_i=1$ for all $i$. To prove that no other solutions exist in the open interval $\theta\in(0,1)$, we utilize the strict inequality
\begin{equation*}
f_i(\theta) > \theta \quad \text{for } \theta\in(0,1),
\end{equation*}
which follows directly from \eqref{mcondition1} and \eqref{graterm1}. Consequently, for any $\theta \in (0,1)$, we have
\begin{equation*}
F(\theta)=\frac{1}{N}\sum_{j=1}^N f_j(\theta) > \frac{1}{N}\sum_{j=1}^N \theta = \theta,
\end{equation*}
which strictly rules out the existence of any other fixed points in $(0,1)$.
This concludes the proof that $\bar s^*=(1,\dots,1)$ is the unique, locally asymptotically semi-stable fixed point for the average maturity dynamics in \eqref{aymardode} within the set $[0,1]^N$.
\end{proof}
\end{theorem}

Several results presented here, as we would like to point out, have counterparts in the $1\text{D}$ model discussed in the appendix of \cite{dimred}. This leaves open the question of which aspects of follicle maturation are preserved or lost when integrating over maturity (as in \citep{dimred}) rather than over age (as in our current setting). We leave a systematic comparison of similarities and differences to future work.

\paragraph{Linear versus quadratic maturation functions}\label{linvsquad}

Here, we explicitly show the effect of the different maturation functions on the support of the solution. While this distinction is unclear from the snapshots in Figure~\ref{fig:aymardALL}, in this section we analytically demonstrate the main difference. 

To simplify the discussion, we again assume a system of two follicles characterized by the parameters in Figure~\ref{fig:celldyn} and identical initial conditions:
\begin{equation}\label{initlinquad}
\phi_1(0,m) = \phi_2(0,m) = \phi(0,m) = \varphi \chi_{[a,b]},
\end{equation}
where $\varphi > 0$ and $\chi_{[a,b]}$ denotes the indicator function of the interval $[a,b] \subset (0,1)$. 
With this choice, we have
\begin{equation} \label{initsimp}
\begin{split}
x_1(0) &= x_2(0) = x_0 = \varphi (b-a), \\
s_1(0) &= s_2(0) = s_0 = \frac{\varphi}{2} (b^2-a^2), \\
\bar s_1(0) &= \bar s_2(0) = \bar s_0 = \frac{1}{2}(b+a).
\end{split}
\end{equation}
With this setting, we establish the following lemma.

\begin{lemma}\label{lem:unique_maturity}
If $c_{i,1} = c_{j,1} = c$ and $\bar s_i(0) = \bar s_j(0)$, then
\begin{equation*}
\bar s_i(t) = \bar s_j(t) = \bar s(t), \quad \forall t \geq 0,
\end{equation*}
for every pair $\bar s_i(t)$ and $\bar s_j(t)$ solving \eqref{avmaturityaymard}.

\begin{proof}
Starting from \eqref{avmaturityaymard}, we have the system
\begin{equation*}
\begin{dcases}
\dot{\bar s}_i = -\bar s_i^2 + (c\bar s_i - c + 1) \frac{1}{N}\sum_{n=1}^N \bar s_n, \\
\dot{\bar s}_j = -\bar s_j^2 + (c\bar s_j - c + 1) \frac{1}{N}\sum_{n=1}^N \bar s_n.
\end{dcases}
\end{equation*}
Therefore, the variable
\begin{equation*}
\tilde s_{ij} = \frac{\bar s_i}{\bar s_j}
\end{equation*}
satisfies the differential equation
\begin{equation} \label{ratio_dynamics}
\dot{\tilde s}_{ij} = -\frac{(\tilde s_{ij}-1)}{\bar s_j} \left[ \bar s_j^2 \tilde s_{ij} + (1-c)\frac{1}{N}\sum_{n=1}^N \bar s_n \right],
\end{equation}
which possesses an equilibrium point at $\tilde s_{ij} = 1$. Since the initial conditions satisfy $\tilde s_{ij}(0) = 1$, uniqueness of solutions implies that $\tilde s_{ij}(t) = 1$, and consequently, $\bar s_i(t) = \bar s_j(t)$ for all $t \geq 0$.
\end{proof}
\end{lemma}

Using Lemma \eqref{lem:unique_maturity}, we can easily analyze the different maturation functions. In particular, under the initial conditions \eqref{initsimp}, the functions in \eqref{aymardmatlin} and \eqref{aymardmatquadratic} reduce to \eqref{linearg} and \eqref{quadraticg}, respectively:
\begin{subequations}
\begin{align}
g_i^{(\mathrm{lin})}\left(m,x,\bar s\right) &= (1-c)\bar s + m c \bar s - m\bar s, \label{linearg} \\
g_i^{(\mathrm{quad})}\left(m,x,\bar s\right) &= (1-c)\bar s + m c \bar s - m^2. \label{quadraticg}
\end{align}
\end{subequations}
The average stage admits an exact solution since in both cases we have
\begin{equation} \label{solutions}
\dot{\bar s} = (1-c)\bar s\left(1-\bar s\right) \implies \bar s(t) = \frac{\bar s_0 e^t}{\bar s_0 e^t - (\bar s_0-1) e^{c t}}.
\end{equation}
Substituting the solution \eqref{solutions} into \eqref{linearg} and solving the characteristic equation yields
\begin{equation}\label{charsolution}
\dot{m}_a^{(i)}(t) = g_i^{(\mathrm{lin})}\left(m_a^{(i)},x,\bar s\right) \implies m_a^{(i)}(t) = \frac{e^{c t} (\bar s_0-a) - \bar s_0 e^t}{(\bar s_0-1) e^{c t} - \bar s_0 e^t}.
\end{equation}

In contrast, no closed-form analytical solution exists when \eqref{solutions} is substituted into \eqref{quadraticg}, due to the nonlinearity of the equation with respect to $m$. To solve this equation, we rely on the numerical ODE solver in Mathematica.
\begin{figure}
\centering
\includegraphics[width=.45\textwidth]{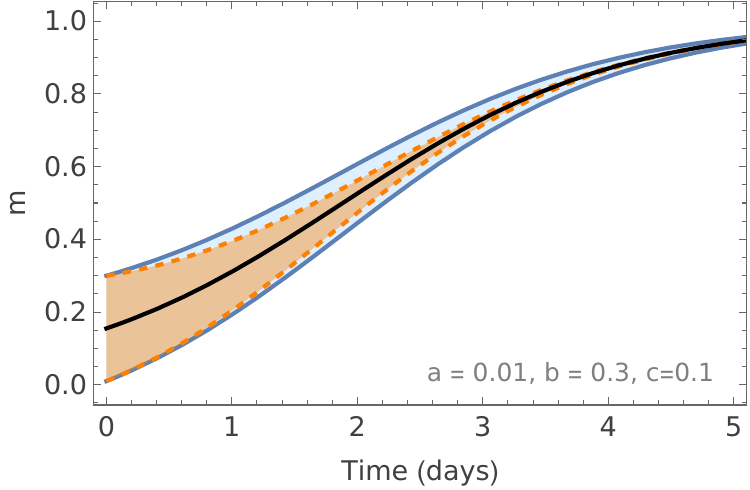}
\quad
\includegraphics[width=.45\textwidth]{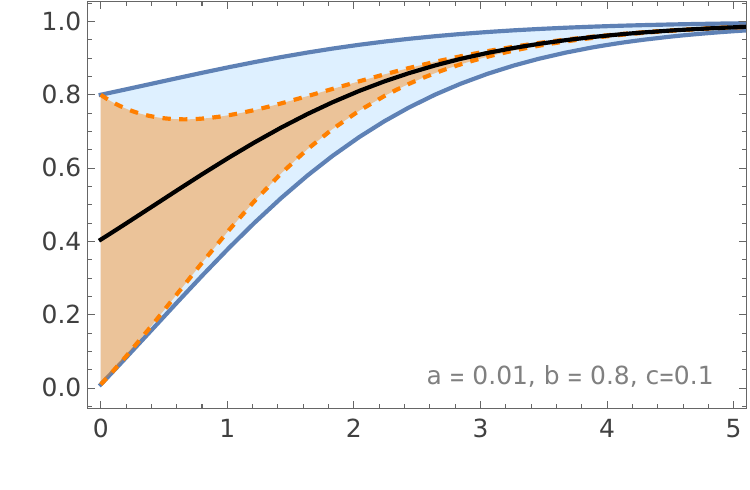}
\caption{Evolution of the support for $\phi(t,m)$ under the initial conditions \eqref{initlinquad}, with varying parameters $a$ and $b$ but an identical value of $c=0.1$. The blue area corresponds to the evolution associated with the linear maturation function, whereas the orange area corresponds to the evolution associated with the quadratic maturation function. The black curve represents $\bar s(t)$ in \eqref{solutions}.}
\label{fig:linvsinit}
\end{figure}

The behavior of the support of $\phi(t,m)$ from \eqref{initlinquad} is plotted in Figure~\ref{fig:linvsinit}. The solid blue lines correspond to the solutions of
\begin{equation}\label{linearchar}
\dot{m}_{a/b}^{(i)}(t) = g_i^{(\mathrm{lin})}\left(m_{a/b}^{(i)},x,\bar s\right) \quad \text{(bottom/top boundaries)},
\end{equation}
while the dashed orange curves correspond to
\begin{equation}\label{quadchar}
\dot{m}_{a/b}^{(i)}(t) = g_i^{(\mathrm{quad})}\left(m_{a/b}^{(i)},x,\bar s\right) \quad \text{(bottom/top boundaries)}.
\end{equation}

In general, we can observe a faster convergence to the average value $\bar s(t)$ in the case of the quadratic maturation function. Note that this faster shrinkage is associated with a higher peak in the PDE solution to preserve the values of $x_i$ and $s_i$ (see Fig.~\ref{fig:aymardALL}).\footnote{Recall that the linear model and the quadratic model with a modified $\Gamma_{i,j}$ as in \eqref{gammas} lead to the same mass-maturity dynamics \eqref{aymardode}.} However, for a small initial support, this difference is somewhat minimal, supporting the idea that the linear (first-order perturbative) approximation of the maturation function introduces a negligible error during the evolution. In particular, if the support is sufficiently small, we can approximate $m^{(i)}_{a/b} = \bar s \pm \varepsilon(t)$, with a positive $\varepsilon(t) \ll 1$, which implies 
\begin{equation*}
\begin{split}
g_i^{(\mathrm{quad})}\left(\bar s \pm \varepsilon, x, \bar s\right) &= (1-c)\bar s + (\bar s \pm \varepsilon) c \bar s - (\bar s \pm \varepsilon)^2 \\
&= (1-c)\bar s + (\bar s \pm \varepsilon) c \bar s - (\bar s^2 \pm 2\varepsilon \bar s) + \mathcal{O}(\varepsilon^2) \\
&= g_i^{(\mathrm{lin})}\left(\bar s \pm \varepsilon, x, \bar s\right) \mp \varepsilon \bar s + \mathcal{O}(\varepsilon^2).
\end{split}
\end{equation*}
Neglecting higher-order terms in $\varepsilon$, the quadratic characteristic equation \eqref{quadchar} can be approximated by the linear one \eqref{linearchar}.

On the other hand, for a large support, the difference becomes evident. In the quadratic case, cells that are initialized with a sufficiently high maturation level (with respect to the average maturity $\bar s_i(0)$) are more susceptible to death. Analytically, this is seen from the fact that for $b = m_{b}^{(i)}(0) = 0.8$ and $\bar s_0$ given as in \eqref{initsimp}, we have
\begin{equation*}
\dot{m}_{b}^{(i)}(t)\Big\vert_{t=0} = g_i^{(\mathrm{quad})}\left(m_{b}^{(i)},x,\bar s\right)\Big\vert_{t=0} = (1-c)\bar s_0 + b c \bar s_0 - b^2 \approx -0.2 < 0.
\end{equation*}

Finally, from the right panel of Figure~\ref{fig:linvsinit}, we can observe the stronger asymmetric evolution of the distribution discussed %in the main text
in Figure~\ref{fig:aymardALL} (Right). Indeed, during the late stages of follicular development, the average maturity $\bar{s}(t)$ approaches the upper boundary $m_b^{(i)}(t)$ faster than the lower boundary $m_a^{(i)}(t)$. This indicates that the statistical average of the distribution shifts toward the upper end of its support rather than remaining centered. On the other hand, for the linear maturation function, the support is transported according to the solution of the characteristic equation \eqref{charsolution},
\begin{equation*}
\mathrm{supp}~\phi(0,\cdot)=[a,b]\mapsto \mathrm{supp}~\phi(t,\cdot)
= \left[
\frac{e^{c t} (\bar s_0-a)-\bar s_0 e^t}
{(\bar s_0-1) e^{c t}-\bar s_0 e^t},
\frac{e^{c t} (\bar s_0-b)-\bar s_0 e^t}
{(\bar s_0-1) e^{c t}-\bar s_0 e^t}
\right]=:[a(t),b(t)]
\end{equation*}
from which we find that the average maturity $\bar{s}(t)$ in \eqref{solutions} coincides with the midpoint of the support,
\begin{equation*}
\frac{a(t)+b(t)}{2}=\bar{s}(t).
\end{equation*}

\subsection{Proofs related to the stage-dependent maturation function}\label{math_lange}

Here, we analyze the coupled equations \eqref{langebas} in Section~\ref{langeconn}. 
Since the mass equation is independent of $s_i$, we can use the results obtained by \cite{Lange2018-xf}. In particular, we recall the following result
\begin{theorem} \label{langethm}
Let $0<\eta<1\leq\nu$ and $n\geq2$. Then, any solution of the mass equation \eqref{langebas-m} satisfying  initial conditions $\xi>x_1(0)>\cdots > x_n(0)$, is continuous and converges to a stable fixed point, 
\begin{equation*}
x_i(t)\to x_i^*\in[0,\xi], ~~ \text{as}~~ t\to+\infty.
\end{equation*}
\end{theorem}
\noindent
See Lange et al.~for the precise statement and the derivation of the stable fixed points. 
For the stage equation, we obtain

\begin{proposition}\label{propstage}
Assume parameters and initial conditions as in Theorem \ref{langethm}. Moreover, let $\alpha>0$ and $0<s_i(0)<x_i(0)$. Then:
\begin{enumerate}
\item $s_i(t)$ and $\bar s_i(t)$ are continuous;
\item $\bar s_i(t)\in (0,1]$ for all $t\geq0$;
\item the fixed point of the stage dynamics coincides with $s_i^*=x_i^*$, as in Theorem \ref{langethm}.
\end{enumerate}
\begin{proof}
Statements 1 and 2 follow from the time-dependent logistic growth of the average stage equation,
\begin{equation*}
\dot{\bar s}_i= \alpha x_i (1-\bar s_i) \bar s_i,
\end{equation*}
and the continuity and boundedness of $x_i$.

Statement 3 can be proved as follows. If $x_i \to x_i^*$ with $0 < x_i^* \leq \xi$, then the average stage equation implies that $\bar s_i \to 1$, which means $s_i \to x_i^*$. On the other hand, if $x_i \to 0$, it follows from the average stage equation and Statement 2 that $\bar s_i \to s_i^*$ for some constant $s_i^* \in (0,1]$. Furthermore, by definition,
\begin{equation*}
\bar s_i(t)=\frac{s_i(t)}{x_i(t)}.
\end{equation*}
Given the continuity and boundedness of $\bar s_i(t)$, if $x_i(t) \to 0$ and $\bar s_i(t) \to s_i^* \in (0,1]$, it must hold that $s_i(t) \to 0$ as well.
\end{proof}
\end{proposition}

%\bmhead
\section*{Acknowledgements}

This work was supported by the European Regional Development Fund (ERDF) within the funding programme 'Sachsen-Anhalt WISSENSCHAFT Forschung und Innovation (EFRE)' under project number ZS/2024/01/183368.

\bibliography{sn-bibliography}

\appendix

\section{Numerical methods} \label{num}

To ensure the manuscript remains self-contained, we provide supplemental material on the original SISYPHE model for follicular development.
Furthermore, we detail our numerical scheme, which is based on the finite-volume method and is employed in the main text.

\subsection{The SISYPHE model 
}\label{originalsisyphe}

In this section, we recall some details of the original model%
\footnote{The core components of the model were first introduced by \cite{ECHENIM200557}; see also the review by \cite{Monniaux2016-yt}.}
and its calibration as provided by \cite{aymard16}. 
The age-maturity structured population model is given by
\begin{equation}
\partial_t \tilde \phi_i
+ \partial_a(\tilde g_i \tilde \phi_i)
+ \partial_m(g_i \tilde \phi_i)
= \tilde \lambda_i \tilde \phi_i.
\label{eq:aymardeqn}
\end{equation}

The integration domain $(a,m)\in U\subset\mathbb R_+\times\mathbb R_+$ is divided into $2D+1$ rectangular sectors, where $D$ denotes the so-called cell cycle. The different transport and loss coefficients are of the form
\begin{equation}
\tilde g_i
= \tilde g_i\left(a,m,\sum_i s_i\right),
\qquad
g_i
= g_i\left(a,m,\sum_i s_i\right),
\qquad
\tilde\lambda_i
= \tilde\lambda_i\left(a,m,\sum_i s_i\right).
\end{equation}

At the boundaries of each sector, the aging function $\tilde g_i$ and the maturation function $g_i$ exhibit discontinuities, and the solutions $\tilde \phi_i$ satisfy different boundary conditions; for the precise definitions, we refer to \cite{aymard16}.

The sink term $\tilde\lambda_i$ is supported only in a thin strip $V=[m_-,m_+]$, where $0<m_-<m_+$. Following their choice, we assume $V=[0.45,0.55]$. The strip $V$ is referred to as the vulnerable zone and represents the maturation region in which cells undergo apoptosis. Moreover, their maturation function is given by
\begin{equation}
g_i\left(a,m,\sum_i s_i\right)
=\tau_i\left[-m^2+(c_{i,1}m+c_{i,2})\tilde\theta\left(\sum_i s_i\right)\right],
\qquad
\tilde\theta\left(\sum_i s_i\right)
=1-\exp\left[-\frac{u_i\left(\sum_i s_i\right)}{\bar u}\right],
\label{aymardcharfull}
\end{equation}
where $\tau_i$ sets the maturation time-scale of the $i$-th follicle and $\bar u$ is a global control parameter. The function $u_i$, which depends on $7$ parameters ($4$ parameters are related to the the global FSH level and $3$ are tuning the local FSH control), represents the locally bioavailable FSH level. In particular, the function $u_i$ is bounded by 
\begin{equation}
b_i U_m\leq u_i\left(\sum_i s_i\right)\leq U_M,
\label{boundlocalfsh}
\end{equation}
where $0<b_i<1$ is a follicle-dependent constant \citep{aymard16}, while $U_m$ and $U_M$ define the lower and upper bounds, respectively, for the FSH level.
The characteristic equation for the $m$-curves during the first cellular growth (typically identified with G1) and differentiation phase is given by
\begin{equation}
\begin{split}
\dot m^{(i)}_a&=\tilde g_i\left(a,m^{(i)}_a,\sum_i s_i\right)\\ &
=\tau_i\left\{-m^2+(c_{i,1}m+c_{i,2})\left[ 1- \exp\left(-\frac{u_i}{\bar u}\right)\right]\right\}.\label{aymardchar}
\end{split}
\end{equation}
At this point, one can choose $b_i$, $c_{i,1}$, and $c_{i,2}$ such that there exists a positive equilibrium point satisfying $\dot m^{(i)}_a=0$. As a function of the bioavailable FSH level $u_i$, the equilibrium of \eqref{aymardchar} is given by
\begin{equation} \label{equlibrium_maturity}
m^*= \frac{1}{2} \left\{c_{i,1}\left[1-\exp\left(-\frac{u_i}{\bar u}\right)\right]+\sqrt{c_{i,1}\left[1- \exp\left(-\frac{u_i}{\bar u}\right)\right]^2+4 c_{i,2}\left[1- \exp\left(-\frac{u_i}{\bar u}\right)\right]}\right\}.
\end{equation}
Using the parameters $p_{ov}$ and $p_{atr}$ reported in Table C.2 of the Supplemental Material of \cite{aymard16}, together with the bounds in \eqref{boundlocalfsh}, we find that the equilibria of the $m$-characteristics are bounded by
\begin{subequations}\label{eqmaturity}
\begin{align} 
&0.52\lesssim m^*_{ov}\lesssim0.62 \label{eqmaturity:ov},\\ &0.50\lesssim m^*_{atr}\lesssim0.5506, \label{eqmaturity:atr}
\end{align}
\end{subequations}
where the upper bounds correspond to the case of maximal locally bio-available FSH $u_i=U_M$.

From \eqref{equlibrium_maturity}, we observe how the choice of parameters $(c_{i,1},c_{i,2})$ determines the location at which the support of the solution $\tilde \phi_i$ contracts in the $m$-direction. For ovulatory parameter sets, this contraction may occur sufficiently far from the vulnerable zone $V$, see \eqref{eqmaturity:ov}. For atretic parameter sets, by contrast, it remains close to the upper boundary of $V$, even under maximal local FSH availability, see \eqref{eqmaturity:atr}. Consequently, the support of $\tilde{\phi}_i$ may persist within the vulnerable zone for a longer period of time.\footnote{Using the plotted values of $u_i$ in \cite{aymard16}, one can verify that $m^*{\mathrm{atr}}\in V$ for atretic follicles, whereas $m^*_{\mathrm{ov}}$ lies above $V$ during the final stages of development.} This suggests that, for follicles whose dynamics are governed by an atretic parameter set, reaching larger masses $x_i$ may be difficult unless their initial mass is sufficiently large.

As a consequence of the contraction of the support, the corresponding average stages tend to stagnate near their equilibrium values. Hence, the maximal average stage is determined by the upper bound in \eqref{equlibrium_maturity}, considered as a function of $u_i$, and depends on the parameter pair $(c_{i,1},c_{i,2})$.

To illustrate the behavior of the average stage, we numerically simulate two follicles using the model and numerical approach developed in \cite{aymard16}, together with the corresponding parameter sets reported in their Supplemental Material. The average stages of the simulation are shown in Figure~\ref{fig:aymardsoverx}. The two uppermost dashed lines indicate the upper bounds for the maturity levels given in \eqref{eqmaturity}. Both follicles are able to reach their respective maximal average maturity; in particular, the follicle with the atretic parameter set lies numerically close to the vulnerable zone. Since
\[
\tau_{\mathrm{atr}}=0.522 > 0.28=\tau_{\mathrm{ov}},
\]
the atretic follicle reaches its follicle-specific maximal stage first. 

\begin{figure}[t]
\centering
\includegraphics[scale=0.4]{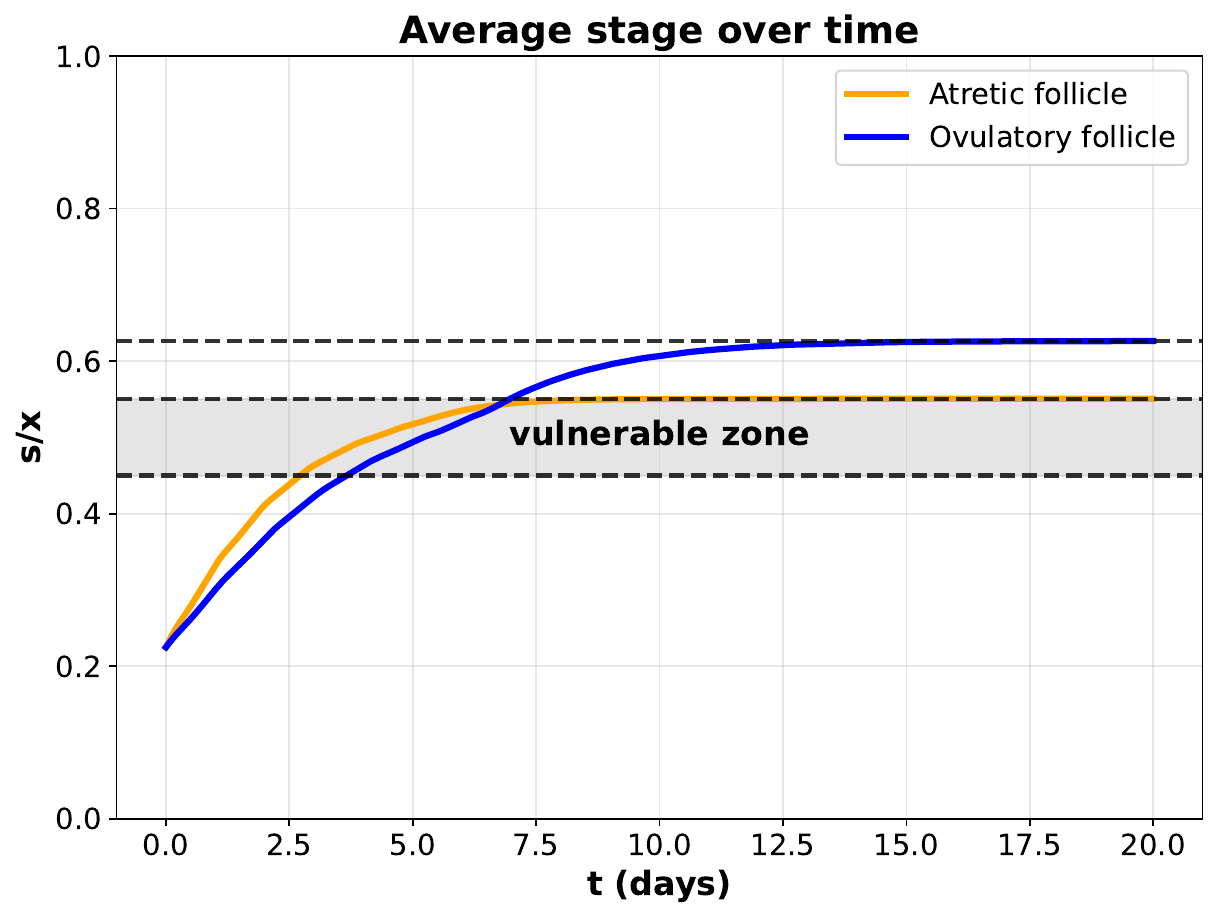}
\caption{Average maturity over time for one ovulatory and one atretic follicle as in ~\cite{aymard16}. Vulnerable zone is shown by the grey shaded area. Similar plots but with different parameters and asymptotics are obtained for the $1D$ model in \cite{dimred}. }
\label{fig:aymardsoverx}
\end{figure}

 We emphasize that, in the modeling framework of \cite{Monniaux2016-yt}, the ovulatory or atretic classification is determined by the follicle's macroscopic state at the time of ovulation $T_{\mathrm{ov}}$ (characterized by the overall follicular maturity). More precisely, the classification is based on whether the follicle's total mass $x_i$ and maturity $s_i$ fall within the empirically determined admissible range at the time $T_\mathrm{ov}$ defined by
\begin{equation}
M_{\mathrm{ov}}=\sum_i \int_0^{T_{\mathrm{ov}}} s_i(\tau)\ud \tau
\end{equation}
for a prescribed ovulatory threshold $M_{\mathrm{ov}}$. This numerical threshold (fixed at $M_{\mathrm{ov}}=15$ in \cite{aymard16}) can be interpreted as a qualitative measure of the cumulative maturation level, which can be coupled with the occurrence of the estradiol peak. As a consequence of this ovulatory threshold, although follicles are assigned ovulatory or atretic parameter sets, these parameters do not, by themselves, determine the follicle's fate.

Because total follicular maturity is not normalized in the SISYPHE model, the numerical value of the threshold $M_{\mathrm{ov}}$ triggering the ovulatory surge depends on the model parametrization and has no direct quantitative interpretation as a measure of follicular (biological) maturity. Accordingly, this threshold should be interpreted primarily as a qualitative descriptor of follicular behavior during ovulation rather than as a quantitatively calibrated measure of cellular and follicular maturity.

With this in mind, the biological interpretation of the auxiliary, non-normalized local maturity variable $m$, and consequently of the average maturity variable $\bar s_i$, is less direct. In particular, a follicle may be classified as ovulatory even when its average maturity is $\bar{s}_i=0.62$ or $\bar{s}_i=0.5506$, i.e. only about $12\%$ and $0.1\%$, resp.,  above the $0.55$ maturity threshold (upper boundary of the vulnerable zone). These values alone should not be interpreted as measures of the (local) biological maturity of the follicle. Rather, the ovulatory classification is determined by the follicle's macroscopic state, in particular its total mass $x_i$ and total maturity $s_i$, together with the endocrine environment experienced during follicular development and the time spent in the vulnerable zone. Therefore, an ovulatory follicle may have the cellular distribution $\tilde{\phi}_i$ partially supported within, or slightly above, the vulnerable zone. This follows from the fact that its average maturity,
\begin{equation}
\bar{s}_i \equiv \frac{s_i}{x_i},
\end{equation}
may remain relatively close to the vulnerable zone, for example $\bar{s}_i\approx 0.55$, while its total mass and total maturity still satisfy the admissibility conditions for ovulation. Note that the negative effect of the vulnerable zone on the follicular mass, depends also on the shape of the sink term $\tilde \lambda_i$ as function of $m$.  Consequently, the model could allow follicles with relatively low average cellular maturity to be classified as ovulatory and to proceed to ovulation. This reflects the fact that the follicular classification should depend not only on the local maturity level, but also on follicular mass, the time spent by the granulosa cells in the vulnerable zone, and the evolving endocrine environment. The distinction between ovulatory and atretic follicles therefore also highlights the importance of an appropriate calibration of the model parameters to ensure consistency between the mathematical classification of follicles and the intended biological interpretation.

At this point, we identify the smallness of the parameter $\bar u$ as the main obstruction to observing a clear competitive regime in the calibrated models \citep{aymard16,dimred}. Indeed, since $\bar u\ll U_M$ and, for most follicles, also $\bar u\ll u_i\left(\sum_i s_i\right)$ over the biologically relevant timescale \citep{aymard16}, we obtain
\begin{equation}
\tilde \theta\left(\sum_i s_i \right)=1-O\left(e^{-\frac{1}{\bar u}}\right).
\end{equation}
Thus, up to an exponentially small correction, the ``competitive'' term is effectively constant, and its dependence on the follicle states is negligible.

Conversely, if $u_i\left(\sum_i s_i\right)$ becomes comparable to  $\bar u$, the fate of a follicle with ovulatory parameters can be substantially influenced by competition through its dependence on \(u_i\). In this regime, the equilibrium of the $m$-characteristics may enter the vulnerable zone, as illustrated, e.g., in \eqref{eqmaturity:ov}. This provides a possible mechanism through which competition can alter the developmental fate of a follicle. Indeed, it is possible that increasing the initial follicular masses alone can shift a follicle from a declining-mass trajectory to an ovulatory one (``rescue mechanism''), as shown by the green dashed line in Fig. 6 of \citet{aymard16}.

Given the complexity of the computational domains and boundary conditions, solutions to \eqref{eq:aymardeqn} are obtained using the finite volume method combined with specialized limiter techniques developed in \cite{aymard16}.  In contrast, the models fitting in the framework \eqref{pdesystem} involve a simpler computational domain and boundary conditions. In the following sections, we adapt the numerical approach of \cite{aymard16} to this simplified setting.

\subsection{Finite‑volume discretization of the PSPM} \label{finite-volume}

In this section, we provide a concise description of the finite-volume scheme used to solve the physiologically structured population model (PSPM) underlying all simulations in this study. The corresponding C++ implementation is available in the source code repository \cite{github_follicular-dynamics}.

For the \(i\)-th follicle, where \(i=1,2,\dots,N_f\), the dynamics of
\(\phi_i\) are described by the following initial-boundary value problem
on the domain \([0,T_{\text{final}}]\times[0,1]\):
\begin{equation}
\begin{dcases}
\frac{\partial \phi_i}{\partial t}
+ \frac{\partial (g_i \phi_i)}{\partial m}
= (p_i-\lambda_i)\phi_i,
\qquad (t,m)\in[0,T_{\text{final}}]\times[0,1],
\label{eq:continuum}
\\
\phi_i(0,m)=\phi_i^{\text{initial}}(m),
\qquad m\in[0,1],
\\
\phi_i(t,0)=\phi_i(t,1)=0,
\qquad t\in[0,T_{\text{final}}].
\end{dcases}
\end{equation}

We assume that all functions are sufficiently smooth; specifically,
\[
g_i,p_i,\lambda_i\in C^{1}\bigl([0,T_{\text{final}}]\times[0,1]\times\Omega^{2N}\bigr),
\]
where $\Omega$ is as in Theorem~\ref{odethm}. We further assume that the solution of \eqref{eq:continuum} exists and is unique \citep{Perthame2007}.
 The initial distribution satisfies 
\(\operatorname{supp}\bigl(\phi_i^{\text{initial}}\bigr)\subset(0,1)\)
. In all simulations we choose the final horizon 
\(T_{\text{final}}\)
 such that the numerical support of 
\(\phi_i(t,\cdot)\)
 remains strictly inside the interval \((0,1)\) for every 
\(t\in[0,T_{\text{final}}]\). 

\begin{figure}
\centering
\begin{tikzpicture}[scale=1.0,decoration={brace,amplitude=6pt}]
% Draw axis
\draw[<->] (-2,0) -- (11,0) node[right] {$m$};
% Cell width
\def\dx{3}
% Cell interfaces (including boundaries)
% Left boundary: m_{1/2} = 0
%\draw (-3.0,0.2) -- (-3.0,-0.2);
% Interior interfaces
\draw (0*\dx,0.2) -- (0*\dx,-0.2);
\draw (1*\dx,0.2) -- (1*\dx,-0.2);
\draw (2*\dx,0.2) -- (2*\dx,-0.2);
\draw (3*\dx,0.2) -- (3*\dx,-0.2);
% Right boundary: m_{N_m+1/2} = 1 (schematic, not to scale)
% Labels for interfaces
\node[below] at (1*\dx, -0.3) {$m_{j-\frac12}$};
\node[below] at (2*\dx, -0.3) {$m_{j+\frac12}$};
% Cell centers (red dots)
\draw[red,fill=red!30] (0.5*\dx,0) circle (.5ex);
\draw[red,fill=red!30] (1.5*\dx,0) circle (.5ex);
\draw[red,fill=red!30] (2.5*\dx,0) circle (.5ex);
% Cell center labels
\node[below] at (0.5*\dx, -0.3) {$m_{j-1}$};
\node[below] at (1.5*\dx, -0.3) {$m_j$};
\node[below] at (2.5*\dx, -0.3) {$m_{j+1}$};
% Cell-centered values
\node at (0.5*\dx,0.8) {$\phi_{i,j-1}^n$};
\node at (1.5*\dx,0.8) {$\phi_{i,j}^n$};
\node at (2.5*\dx,0.8) {$\phi_{i,j+1}^n$};
% Flux labels at interfaces
\node (leftflux) at (1*\dx, 2.0) {$F_{i,j-\frac12}^n$};
\node (rightflux) at (2*\dx, 2.0) {$F_{i,j+\frac12}^n$};
% Arrows to interfaces
\draw[-{Latex[width=3mm]}] (leftflux.south) -- (1*\dx, 0.25);
\draw[-{Latex[width=3mm]}] (rightflux.south) -- (2*\dx, 0.25);
% Downward-facing brace for the central cell
\draw[decorate,decoration={brace,mirror,amplitude=6pt}]
(0*\dx,-1.2) -- (1*\dx-0.05,-1.2)
node[midway,below=8pt] {cell $j-1$};
\draw[decorate,decoration={brace,mirror,amplitude=6pt}]
(1*\dx+0.05,-1.2) -- (2*\dx-0.05,-1.2)
node[midway,below=8pt] {cell $j$};
\draw[decorate,decoration={brace,mirror,amplitude=6pt}]
(2*\dx+0.05,-1.2) -- (3*\dx,-1.2)
node[midway,below=8pt] {cell $j+1$};
\end{tikzpicture}
\caption{Cell-centred discretization of the maturity interval. The red dots denote the cell centres \(m_j\), and the small black vertical lines indicate the cell interfaces \(m_{j\pm\frac12}\). The black arrows indicate the numerical fluxes \(F_{i,j\pm\frac12}^{\,n}\) crossing the corresponding interfaces. The left and right physical boundaries, which are not shown, correspond to \(m_{1/2}=0\) and \(m_{N_m+1/2}=1\).}
\label{fig:cell-centered-grid}
\end{figure}
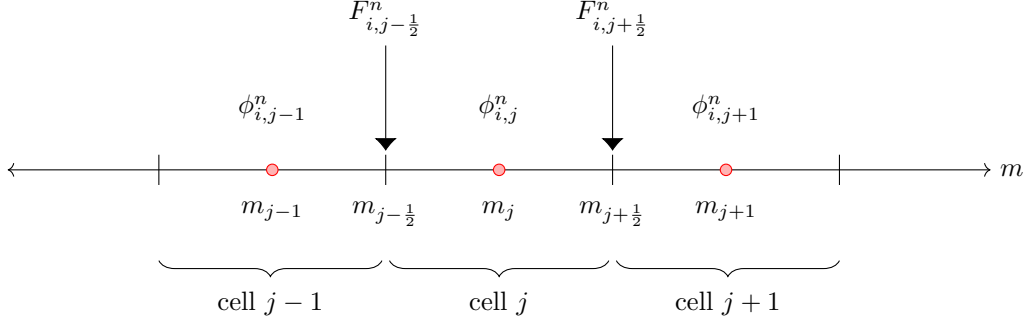

 The finite‑volume discretization follows the classic approach described in
\cite{LeVeque2022} and \cite{Toro2009}.
We partition the maturity interval \([0,1]\) into \(N_m\) uniform control volumes of width 
\(\Delta m = 1/N_m\). The cell interfaces are

\begin{equation*}
m_{j-\frac12} := (j-1)\Delta m,
\qquad
m_{j+\frac12} := j\Delta m,
\qquad \text{for } j=1,2,\dots,N_m.
\end{equation*}
The centre of each control volume is 
\begin{equation*}
m_j := \frac{m_{j-\frac12}+m_{j+\frac12}}{2}
= \left(j-\frac12\right)\Delta m,
\qquad \text{for } j=1,2,\dots,N_m.
\end{equation*}

For every follicle \(i\), we denote by \(\phi_{i,j}^{\,n}\) the
cell-average approximation of \(\phi_i\) over the \(j\)-th control
volume at time \(t^n\), associated with the cell centre \(m_j\).
The discrete transport velocity is evaluated at the cell interfaces, and the growth and decay rates are computed from the complete discrete population state defined in \eqref{eq:complete-discrete-population}. Thus,
we write
\[
\phi_{i,j}^{\,n}\approx\phi_i(t^n,m_j),
\qquad
g_{i,j+\frac12}^{\,n}
\approx
g_i\!\bigl(\phi_i^{\,n},m_{j+\frac12}\bigr),
\qquad
p_{i}^{\,n}\approx
p_i\!\bigl(\boldsymbol{\phi}^{\,n}\bigr),
\qquad
\lambda_{i}^{\,n}\approx
\lambda_i\!\bigl(\boldsymbol{\phi}^{\,n}\bigr),
\]
where
\begin{equation}
\boldsymbol{\phi}^{\,n} =\left\{ \phi_{i,j}^{\,n} \right\}_{i=1,\ldots,N_f;\,j=1,\ldots,N_m}
\label{eq:complete-discrete-population}
\end{equation}
denotes the complete discrete population state.

The temporal grid is defined recursively by
\begin{equation*}
t^{0}=0,
\qquad
t^{n+1}=t^n+\Delta t^{\,n},
\qquad
n=0,1,\ldots,N_t-1,
\label{eq:time-grid}
\end{equation*}
with $t^{N_t}=T_{\text{final}}$. The time steps \(\Delta t^{\,n}\) may vary in time according to the adaptive time-step criterion described below.
The time step is chosen adaptively according to
\begin{equation}
\Delta t^{\,n}
=
\min\left\{
\mathrm{CFL}\,
\frac{\Delta m}
{\displaystyle\max_{\substack{1\leq i\leq N_f\\
1\leq j\leq N_m}}
\left|g_{i,j+\frac12}^{\,n}\right|+\varepsilon},
\;
\frac{1}
{2\displaystyle\max_{1\leq i\leq N_f} \left|p_i^{\,n}-\lambda_i^{\,n}\right| + \varepsilon}
\right\}.
\label{eq:time-step}
\end{equation}
The Courant--Friedrichs--Lewy number $\mathrm{CFL}\in[0,1]$ controls the
transport time-step restriction. In all simulations we use
$\mathrm{CFL}=0.8$. 

The cell average of the continuous density over the \(j\)-th control
volume is
\begin{equation}
\bar{\phi}_{i,j}(t)
:=
\frac{1}{\Delta m}
\int_{m_{j-\frac12}}^{m_{j+\frac12}}
\phi_i(t,m)\,\mathrm{d}m.
\label{eq:cell-average}
\end{equation}
The discrete quantity \(\phi_{i,j}(t)\) is interpreted as an
approximation of this cell average. For sufficiently smooth solutions,
the midpoint value satisfies
\[
\bar{\phi}_{i,j}(t)
=
\phi_i(t,m_j)+\mathcal O(\Delta m^2),
\]
so that cell-centre evaluations provide a second-order approximation
of the cell averages.

Integrating \eqref{eq:continuum} over the \(j\)-th control volume and approximating the cell average of $\phi_i$ by its discrete value $\phi_{i,j}$ gives the semi-discrete finite-volume system
\begin{equation}
\frac{\mathrm{d}\phi_{i,j}}{\mathrm{d}t}
=
-\frac{1}{\Delta m}
\left(
F_{i,j+\frac12}-F_{i,j-\frac12}
\right)
+
\bigl(p_i-\lambda_i\bigr)\phi_{i,j},
\qquad
i=1,\ldots,N_f,\quad j=1,\ldots,N_m.
\label{eq:fvm-semidiscrete}
\end{equation}
Here \(F_{i,j+\frac12}\) denotes the numerical approximation of the
transport flux \(g_i\phi_i\) through the interface \(m_{j+\frac12}\).
At internal interfaces, it is defined by the first-order upwind rule
\begin{equation}
F_{i,j+\frac12}
=
\begin{cases}
g_{i,j+\frac12}\,\phi_{i,j},
& g_{i,j+\frac12}\ge 0,\\[4pt]
g_{i,j+\frac12}\,\phi_{i,j+1},
& g_{i,j+\frac12}<0,
\end{cases}
\qquad
j=1,\ldots,N_m-1.
\label{eq:upwind-flux}
\end{equation}

The homogeneous boundary conditions
are imposed through the upwind numerical fluxes at the physical boundaries. Since the prescribed boundary density is zero, the boundary fluxes are
\begin{equation}
F_{i,\frac12}
=
\begin{cases}
g_{i,\frac12}\,\phi_{i,1},
& g_{i,\frac12}<0,\\[4pt]
0,
& g_{i,\frac12}\geq 0,
\end{cases}
\qquad
F_{i,N_m+\frac12}
=
\begin{cases}
0,
& g_{i,N_m+\frac12}\leq 0,\\[4pt]
g_{i,N_m+\frac12}\,\phi_{i,N_m},
& g_{i,N_m+\frac12}>0.
\end{cases}
\label{eq:boundary-fluxes}
\end{equation}
Thus, an incoming flux is determined by the prescribed boundary value
\(\phi_i=0\), whereas an outgoing flux is determined by the adjacent
interior cell value. This is the standard first-order upwind treatment
of the homogeneous Dirichlet boundary conditions.

The semi-discrete system \eqref{eq:fvm-semidiscrete} can be written compactly as
\begin{equation}
\frac{\mathrm{d}\boldsymbol{\phi}}{\mathrm{d}t}
=
\mathcal{L}(\boldsymbol{\phi}),
\label{eq:semi-discrete}
\end{equation}
where the right-hand-side operator is defined by
\begin{equation}
\mathcal{L}_{i,j}(\boldsymbol{\phi})
=
-\frac{1}{\Delta m}
\left(
F_{i,j+\frac12}(\phi_i)
-
F_{i,j-\frac12}(\phi_i)
\right)
+
\bigl(
p_i(\boldsymbol{\phi})
-
\lambda_i(\boldsymbol{\phi})
\bigr)
\phi_{i,j}.
\label{eq:fv-operator}
\end{equation}
The dependence of the growth and decay rates on
$\boldsymbol{\phi}$ accounts for the nonlinear coupling between the
follicle sub-populations.

The classical fourth-order Runge--Kutta (RK4) method is applied to
\eqref{eq:semi-discrete}. For each time step, the stage derivatives given by
\begin{subequations}
\begin{align}
\boldsymbol{k}_1
&=
\mathcal{L}\left(\boldsymbol{\phi}^{\,n}\right),
\label{eq:k1}
\\
\boldsymbol{k}_2
&=
\mathcal{L}\left(
\boldsymbol{\phi}^{\,n}
+ \frac{\Delta t^{\,n}}{2}\boldsymbol{k}_1
\right),
\label{eq:k2}
\\
\boldsymbol{k}_3
&=
\mathcal{L}\left(
\boldsymbol{\phi}^{\,n}
+ \frac{\Delta t^{\,n}}{2}\boldsymbol{k}_2
\right),
\label{eq:k3}
\\
\boldsymbol{k}_4
&=
\mathcal{L}\left(
\boldsymbol{\phi}^{\,n}
+ \Delta t^{\,n}\boldsymbol{k}_3
\right).
\label{eq:k4}
\end{align}
\end{subequations}
Thus, the finite-volume operator, including the numerical fluxes and
state-dependent coefficients, is evaluated at the corresponding RK4
stage state.

The solution is then advanced according to
\begin{equation}
\boldsymbol{\phi}^{\,n+1}
=
\boldsymbol{\phi}^{\,n}
+
\frac{\Delta t^{\,n}}{6}
\left(
\boldsymbol{k}_1
+ 2\boldsymbol{k}_2
+ 2\boldsymbol{k}_3
+ \boldsymbol{k}_4
\right).
\label{eq:rk4-update}
\end{equation}

\paragraph{Errors}\label{error}
In order to validate the PDE solution, we compare the average stage $\bar s_i^{\mathrm{(PDE)}}$ computed from the PDE solution obtained in the previous section with the average stage $\bar s_i^{\mathrm{(ODE)}}$ computed from the ODE system using the numerical ODE solver implemented in Mathematica. Clearly, the two solutions should agree.

\begin{figure}[t]
\centering
\includegraphics[scale=1]{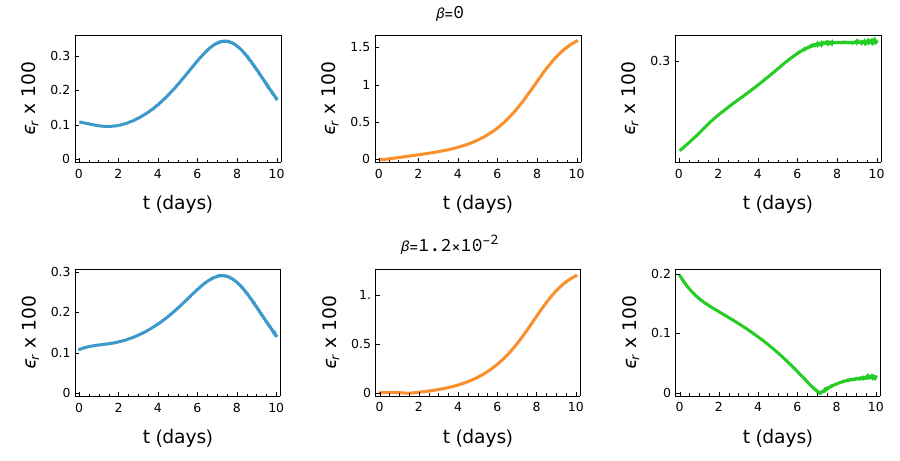}
\caption{Percentage error between the ODE and PDE solutions for different follicles, represented by different colors, and for $\beta=0$ (top row) and $\beta=1.2\times 10^{-2}$ (bottom row) for the model \eqref{langepde}, using $N_m=2500$.}
\label{fig:errors}
\end{figure}

In particular, we consider the model \eqref{langepde} with the parameters given in Fig. \ref{fig:bovine419} and the initial conditions in \eqref{init_pde_lange}. In Fig. \eqref{fig:errors}, we show the percentage error
\begin{equation}
100\epsilon_r
=100\frac{\left\vert \bar s_i^{\mathrm{(PDE)}}-\bar s_i^{\mathrm{(ODE)}}\right\vert}
{\bar s_i^{\mathrm{(ODE)}}}.
\end{equation}

The different colors indicate the follicles, following the color scheme in Fig. \ref{fig:lange}, while the different rows correspond to different values of $\beta$. In all cases, we observe that the error does not exceed $1.5\%$, which validates the PDE solution.

\section{Additional data}\label{additionalcow}

For comparison, we apply the methods of Section~\ref{langeconn} to data from the first follicular wave of another cow, shown in Figure \ref{fig:bovine3674}. In contrast to the cow presented in Figure \ref{fig:bovine419}, we observe larger discrepancies between model predictions and data, both in follicle size and estradiol concentration. Nevertheless, the model captures the overall qualitative trends of both quantities.

Two main differences between cow $419$ and cow $3674$ are noteworthy. First, the estradiol concentration data exhibit a decrease slightly before day $4$. Although the model does not accurately reproduce the concentration peak, it successfully captures this early drop.

Second, the atretic follicle (yellow) shows a slow decrease in size, accompanied by a gradual decline in stage. From a modeling perspective, both the competition coefficient $\kappa$ and the self-regulation coefficient $\eta$ are smaller for cow $3674$, leading to weaker competitive and regulatory effects during the later phases of development. As a consequence of the slow stage decay, the average stage approaches the maximal value $\bar{s} = 1$ even while the follicle is losing mass.

A possible biological interpretation is that, during the gradual degeneration process, the remaining viable cells and structures may continue to mature despite the overall reduction in follicle size. In the absence of histological data, however, this predicted behavior cannot be experimentally verified at this stage.

\begin{figure}[t]
\centering
\includegraphics[width=0.45 \textwidth]{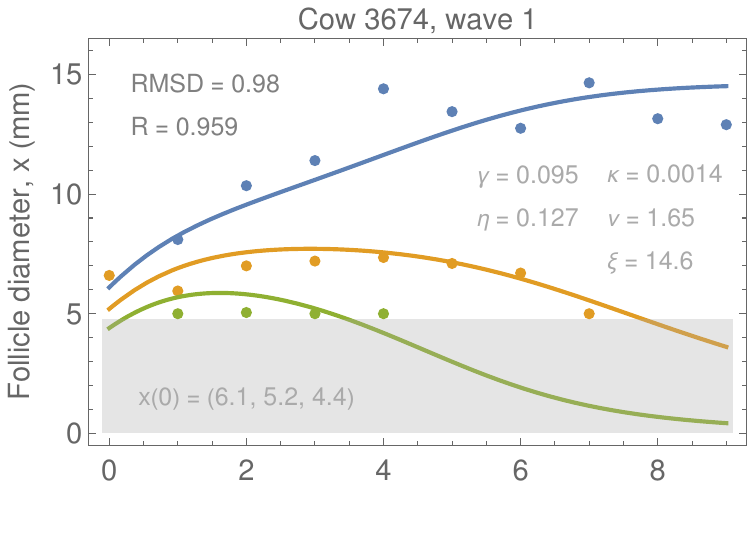} \quad
\includegraphics[width=0.45 \textwidth]{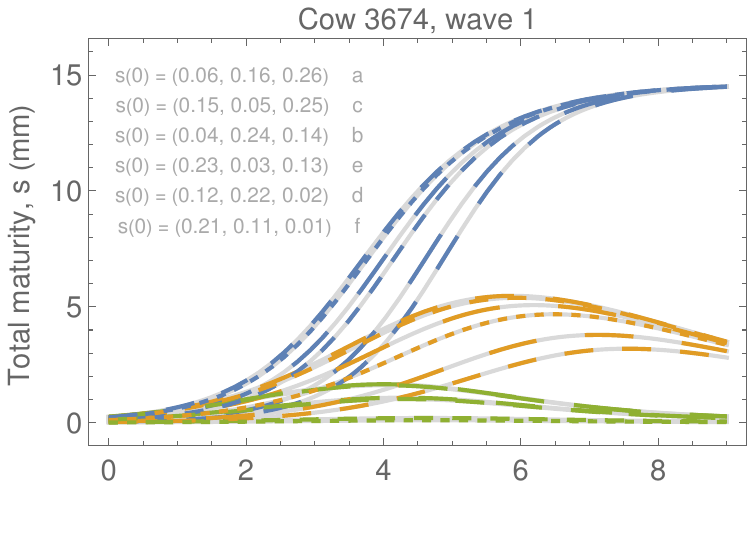}
\\[-4ex]
\includegraphics[width=0.45 \textwidth]{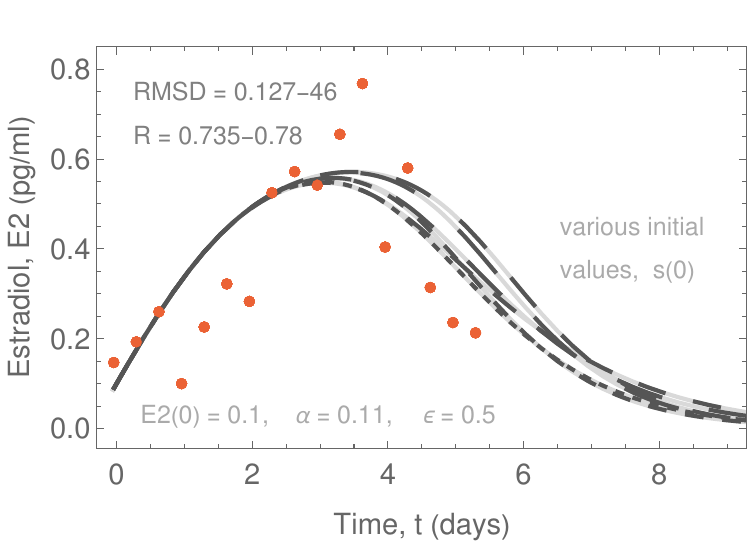} \quad
\includegraphics[width=0.45 \textwidth]{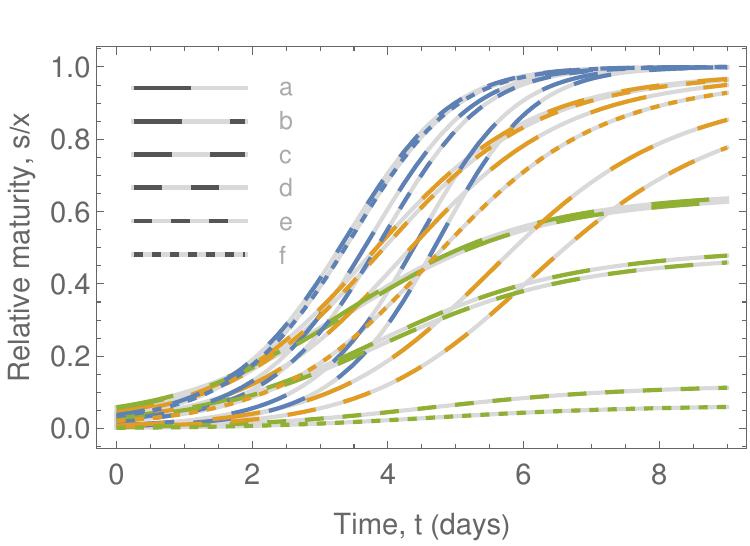} 
\\[2ex]
\caption{Fitting ovarian follicle sizes and blood estradiol concentrations to other cow data.
This is a copy of Fig.~\ref{fig:bovine419}
for cow no.~3674 in the data set of \cite{Cummins2012-oa}; this cow has also been studied in \cite{Lange2018-xf}.}
\label{fig:bovine3674}
\end{figure}    

\section{Characteristics of the PDE}
\label{app:characteristics}

For completeness, we briefly discuss the characteristic equations associated with the PDE \eqref{pdesystem}; cf.~\cite{evans2010partial,griffiths2015essential}. Expansion of the transport term by the product rule yields
\begin{equation}
\partial_t \phi_i
+ g_i \partial_m \phi_i
=
\left(-\partial_m g_i-\lambda_i+p_i\right)\phi_i,
\label{app:pdeexpanded}
\end{equation}
where $g_i=g_i(t,m,x,s)$, $\lambda_i=\lambda_i(t,m,x,s)$ and $p_i=p_i(t,m,x,s)$.
The characteristic curve $m_a^{(i)}(t)$ is defined as the solution to the characteristic equation
\begin{equation}
\dot{m}_a^{(i)}
=
g_i\left(t,m_a^{(i)},x,s\right),
\qquad
m_a^{(i)}(0)=f(a),
\label{app:explicitchar}
\end{equation}
where the supscript $a$ identifies the curve through its initial value via some bijective function $f$. Along a characteristic curve, the chain rule gives
\begin{equation}
\frac{d}{dt}
\phi_i\left(t,m_a^{(i)}(t)\right)
=
\partial_t\phi_i
+
\dot{m}_a^{(i)}\partial_m\phi_i.
\label{app:chainrule}
\end{equation}
Using \eqref{app:explicitchar} and \eqref{app:pdeexpanded}, we obtain
\begin{equation}
\frac{d}{dt}
\phi_i\left(t,m_a^{(i)}(t)\right)
=
\Big[
-\partial_m g_i
-\lambda_i
+p_i
\Big]\Big\vert_{m=m^{(i)}_a(t)}
\phi_i\left(t,m_a^{(i)}(t)\right).
\label{app:phichar}
\end{equation}
Thus, \eqref{app:explicitchar} describes the evolution of maturity along a characteristic curve, whereas \eqref{app:phichar} describes the evolution of the density function along that characteristic curve.
Finally, we observe that by substituting 
\eqref{a2}
in \eqref{app:explicitchar}, the characteristic equation becomes
\begin{equation}
\dot{m}_a^{(i)}
=
g_{i,0}\left(t,x,x\circ\bar{s}\right)
+
m_a^{(i)}
g_{i,1}\left(t,x,x\circ\bar{s}\right),
\label{app:explicit-characteristic}
\end{equation}
which is indeed the form used in the main body of this work.

\end{document}